%% file: tfm.tex
\documentclass[11pt,letterpaper]{article}

\input{macro.tex}

\newcommand{\elaine}[1]{{\footnotesize\color{magenta}[Elaine: #1]}}
\newcommand{\Hao}[1]{{\footnotesize\color{blue}[Hao: #1]}}
\renewcommand{\elaine}[1]{}
\renewcommand{\Hao}[1]{}

\newcounter{cnt:challenge}

\begin{document}
\begin{titlepage}
\title{Foundations of Transaction Fee Mechanism Design}
\author{}
\author{Hao Chung\thanks{Supported by Packard Fellowship, NSF award 2044679,  and a gift from Nikolai Mushegian.} \\ CMU \\ {\tt haochung@andrew.cmu.edu} \and Elaine Shi\footnotemark[1] \\ CMU \\{\tt runting@cs.cmu.edu}}
\date{\vspace{-20pt}}

\maketitle
\thispagestyle{empty}

\begin{abstract}
\input{abstr}

\end{abstract}

\end{titlepage}

\tableofcontents
\thispagestyle{empty}
\newpage
\setcounter{page}{1}
\thispagestyle{empty}

\input{intro-new}

\input{roadmap-new}

\input{related}

\input{defn-new}

\input{newImpossible}
\input{weakic}

\input{proof-burn-2nd-price}

\input{lb-weak}

\input{incleqconf-main}

\input{conclusion}

\input{acks}


\bibliographystyle{alpha}
\bibliography{refs,crypto,bitcoin}

\appendix
\input{ICcomparison}


\input{largeCoalition}

\input{mechComparison}

\end{document}

%% file: macro.tex
\ifdefined\ACM
\else
\usepackage[margin=1in]{geometry}
\usepackage{amsthm}
\fi

\usepackage{cite}
\usepackage[bookmarks=true,pdfstartview=FitH,colorlinks,linkcolor=darkred,filecolor=darkred,citecolor=darkred,urlcolor=darkred]{hyperref}

\usepackage{colortbl}
\usepackage{mdframed}
\usepackage{color,soul}
\usepackage{paralist}
\usepackage{pdfpages}
\usepackage{enumitem}
\usepackage{float}
\usepackage{booktabs}
\usepackage[override]{cmtt}
\usepackage{wrapfig}
\usepackage{algorithm}
\usepackage[noend]{algpseudocode}
\usepackage{algorithmicx}
\usepackage{boxedminipage}
\usepackage{url}
\usepackage{graphicx}
\usepackage[bf,justification=centering]{caption}
\usepackage[justification=centering]{subcaption}
\usepackage{xspace}
\usepackage{multirow}
\usepackage{amsmath,amstext,amsfonts,amssymb,latexsym}
\usepackage{amsbsy}
\usepackage{stmaryrd}
\usepackage{pifont}
\usepackage{color,xcolor}
\usepackage{graphicx,color,eso-pic}
\usepackage[
    lambda,
    advantage,
    operators,
    sets,
    adversary,
    landau,
    probability,
    notions,
    logic,
    ff,
    mm,
    primitives,
    events,
    complexity,
    asymptotics,
    keys]
{cryptocode}
\usepackage[capitalize,noabbrev]{cleveref} 
\usepackage{tikz}
\usetikzlibrary{fit, shapes, positioning, patterns, arrows,shadows}
\usetikzlibrary{graphs,graphs.standard,arrows}
\usetikzlibrary{decorations.markings}

\tikzset{
    master/.style={
        execute at end picture={
            \coordinate (lower right) at (current bounding box.south east);
            \coordinate (upper left) at (current bounding box.north west);
        }
    },
    slave/.style={
        execute at end picture={
            \pgfresetboundingbox
            \path (upper left) rectangle (lower right);
        }
    }
}

\sethlcolor{lightgray}

   \newcommand{\bfb}{\ensuremath{{\mathbf{b}}}\xspace}
\newcommand{\bfC}{\mathbf{C}}   \newcommand{\bfc}{\mathbf{c}}
   \newcommand{\bfd}{\mathbf{d}}
   \newcommand{\bfe}{\mathbf{e}}

\newcommand{\bfI}{\mathbf{I}}

\newcommand{\bfM}{\mathbf{M}}

\newcommand{\bfP}{\mathbf{P}}   \newcommand{\bfp}{\ensuremath{{\mathbf{p}}}\xspace}

   \newcommand{\bfx}{\mathbf{x}}

\DeclareMathOperator*{\newargmax}{arg\,max}
\definecolor{darkred}{rgb}{0.5, 0, 0}
\definecolor{darkgreen}{rgb}{0, 0.5, 0}
\definecolor{darkblue}{rgb}{0,0,0.5}

\newcommand\markx[2]{}

\renewcommand{\path}{\ensuremath{\mathsf{path}}\xspace}

\newcommand{\N}{\ensuremath{\mathbb{N}}\xspace}

\newcommand{\ignore}[1]{}

\newcommand{\R}{\ensuremath{\mathbb{R}}}

\newcommand{\E}{\mathbb{E}}

\definecolor{darkgreen}{rgb}{0,0.5,0}
\definecolor{lightblue}{RGB}{0,176,240}
\definecolor{darkblue}{RGB}{0,112,192}
\definecolor{lightpurple}{RGB}{124, 66, 168}
\definecolor{grey}{RGB}{139, 137, 137}
\definecolor{maroon}{RGB}{178, 34, 34}
\definecolor{green}{RGB}{34, 139, 34}
\definecolor{types}{RGB}{72, 61, 139}
\definecolor{gold}{rgb}{0.8, 0.33, 0.0}

\definecolor{darkgray}{gray}{0.3}

\newcounter{task}

\ifdefined\ACM
\else
    \newtheorem{thm}{Theorem}[section]      
    \newtheorem{theorem}[thm]{Theorem}

    \newtheorem{lemma}[thm]{Lemma}
    \newtheorem*{lemma*}{Lemma}
    \newtheorem{claim}[thm]{Claim}
    \newtheorem{corollary}[thm]{Corollary}
    \newtheorem{fact}[thm]{Fact}
    \newtheorem{proposition}[thm]{Proposition}

    \newtheoremstyle{boxes}
    {2pt}
    {0pt}
    {}
    {}
    {\bfseries}
    {}
    {\newline}
    {\thmname{#1}\thmnumber{ #2}:  
    \thmnote{#3}}

    \theoremstyle{boxes}

    {
    \theoremstyle{definition}
    \newtheorem{definition}{Definition}
    \newtheorem{remark}{Remark}
    }
\fi

%% file: abstr.tex
In blockchains such as Bitcoin and Ethereum, users compete in a transaction fee auction
to get their transactions confirmed in the next block.
A line of recent works set forth the desiderata for a ``dream'' 
transaction fee mechanism (TFM), and explored whether 
such a mechanism existed. 
A dream TFM should satisfy 1)
{\it user incentive compatibility} (UIC), i.e., 
truthful bidding should be a user's dominant strategy;
2) {\it miner incentive compatibility} (MIC), 
i.e., the miner's dominant  strategy 
is to faithfully implement the prescribed mechanism;
and 3) {\it miner-user side contract proofness} (SCP), i.e., 
no coalition of the miner and one or more user(s) can increase
their joint utility by deviating from the honest behavior.
The weakest form of SCP is called $1$-SCP, where 
we only aim to provide resilience against the collusion of the miner 
and {\it a single} user.
Sadly, despite the various attempts, to the best of knowledge, no existing
mechanism can satisfy all three properties in all situations.

Since the TFM  
departs from classical mechanism design in modeling and assumptions, 
to date, our understanding
of the design space is relatively little.
In this paper, we 
further unravel the mathematical structure 
of transaction fee mechanism design by proving the following results: 
\begin{itemize}[leftmargin=5mm]
\item 
{\it Can we have a dream TFM?}
We prove a new impossibility result: {\it assuming finite block size}, 
no single-parameter, non-trivial, 
possibly randomized TFM 
can simultaneously satisfy UIC and $1$-SCP.
Consequently, no non-trivial TFM can satisfy all three desired
properties simultaneously.  
This answers an important open question raised by Roughgarden in his recent work.
\item 
{\it Rethinking the incentive compatibility notions.}
We observe that the prevalently adopted incentive compatibility notions may be   
too draconian and somewhat flawed. 
We rectify the existing modeling techniques, and
suggest a relaxed incentive compatibility notion 
that captures additional hidden costs of strategic deviation. 
We construct a new mechanism called the ``burning second-price auction'',
and show that it indeed satisfies the new incentive compatibility notions.
We additionally prove that
the use of randomness is necessary 
under the new incentive compatibility notions for 
``useful''
mechanisms that resist
the coalitions of the miner and at least $2$ users.
\item 
{\it Do the new design elements make a difference?}
Unlike classical mechanisms, TFMs may employ a couple new design elements
that are idiosyncratic to blockchains.
For example, a burn rule (employed by Ethereum's EIP-1559)
allows part to all of the payment from the users to be burnt rather than paid
to the miner. Some mechanisms also allow unconfirmed transactions to be
included in the block, to set the price for others. 
Our work unveils how these new design elements actually make a difference
in TFM design, allowing us to achieve incentive compatible properties 
that would otherwise be impossible.

\ignore{
Roughgarden 
showed that Ethereum's EIP-1559 achieves UIC, MIC, 
and SCP simultaneously --- but {\it only when the block size is infinite}. 
Intriguingly, Ethereum's EIP-1559 makes use a {\it burn rule}, where 
all users' payment (except possibly a bare minimum tip) is {\it burnt}
rather than paid to the miner.
Such a burn rule is novel and idiosyncratic to blockchains.
We show that even when assuming 
infinite block size, 
having a burn rule is necessary for any (possibly randomized) TFM
that simultaneously satisfies UIC and $1$-SCP; and thus
the burn rule indeed plays a critical role in EIP-1559.
Without a burn 
rule, the only TFM that satisfies both UIC and $1$-SCP is the trivial mechanism
where users always pay nothing and the miner gets nothing.
\item 
{\it Is it useful for blocks to contain unconfirmed transaction?}
There is an ongoing debate in the community whether 
we should allow a block to contain unconfirmed transactions which might be there
just to ``set the price''.
We show that 
the ability for a block to contain 
unconfirmed transactions may be useful. Specifically, we prove 
that if we insisted that all transactions in a block must be confirmed, 
it is impossible to have a non-trivial TFM that satisfies
even weak incentive compatibility. 
}
\end{itemize}


\ignore{
Our impossibility result leads to 
an important implication for Ethereum's recent EIP-1559 proposal.
EIP-1559 is arguably (among) the closest we have
come to in terms of 
satisfying all three properties. 
When there is congestion,  EIP-1559 approximates a first-price auction 
and therefore is not UIC. On the other hand, when the block size 
is infinite (i.e., plentiful),
EIP-1559 approximates a simple posted-price 
auction where all the proceeds (except possibly a bare minimal tip) are {\it burnt}
rather than paid to the miner.  
The burn rule is a novel idea that is idiosyncratic to blockchains.
Our result implies that such a burn rule (or the like) 
is indeed necessary for any non-trivial, single-parameter
mechanism that achieves all three properties,
even if only in the infinite block size regime.
}


%% file: intro-new.tex
\elaine{discuss split bid in one place}

\section{Introduction}
\label{sec:intro}
In decentralized blockchains such as Bitcoin and Ethereum, miners
are incentivized 
to collectively maintain the public ledger, since they can collect block rewards and
transaction fees.
Today, a simple ``pay your bid'' auction is implemented
by major blockchains like Bitcoin. 
In a ``pay your bid'' auction, the miners' dominant strategy
is to take the highest bids. However, 
users may be incentivized to bid strategically, e.g., 
bid close to $0$ when there is no congestion, or bid the minimum possible
to get selected when there is congestion. 
Earlier works~\cite{zoharfeemech,yaofeemech,functional-fee-market} pointed out  
such strategic bidding indeed happens in real life, 
and is considered undesirable.
Consequently, 
several works~\cite{zoharfeemech,yaofeemech,functional-fee-market,eip1559,roughgardeneip1559,roughgardeneip1559-ec,dynamicpostedprice}
call out to the community to rethink the design of transaction  
fee mechanisms (TFMs). These works raise the following  
important question: {\it what is the ideal transaction fee mechanism}?

\paragraph{Desiderata of a dream TFM.}
Partly due to its decentralized nature, transaction fee mechanism (TFM)
design departs from classical mechanism design~\cite{myerson,agt} 
in modeling and assumptions.
We face several challenges
that arise from the strategic behavior of the miner and 
of miner-user coalitions: 
\begin{itemize}[leftmargin=5mm]
\item 
{\it Challenge 1: strategic behavior of the miner.}
The vast majority of work 
in the classical mechanism design literature (with some 
exceptions~\cite{credibleauction,commit-credible-auction,shillbid02,shillbid00,shillbid01}
which we discuss further in Section~\ref{sec:related})
assumes that the auctioneer
is trusted and implements the prescribed mechanism honestly 
---  therefore, we mainly care about how to design mechanisms 
such that the users are incentivized to bid truthfully.
In a decentralized environment, 
the auctioneer is no longer fully trusted. 
In a blockchain transaction fee mechanism, 
the miners and the logic of the blockchain jointly serve
as the ``auctioneer''. 
Although the logic of the blockchain is hard-coded and unalterable, 
miners can deviate from the prescribed mechanism, 
and behave strategically to increase their financial gains. 
As a simple example, consider a classical Vickrey auction~\cite{vickrey}. 
Suppose that each block has size $B$. 
We can then include the top $B$ bids 
into the block, among which the first $B-1$ are considered {\it confirmed}
and they pay the $B$-th price. 
If there are strictly  fewer than $B$ bids, everyone gets 
confirmed and they all pay a price of $0$.
All users' payment goes to the miner that mines the block.
Classical algorithmic game theory~\cite{vickrey,agt} 
tells us that such a Vickrey auction is 
dominant strategy incentive compatible (DSIC) for the users,
assuming that the miner indeed behaves honestly. 
Unfortunately, several prior 
works~\cite{functional-fee-market,roughgardeneip1559} pointed out
that the Vickrey auction is not incentive compatible
for the miner, since the miner may want to inject
a fake transaction whose price is between the $(B-1)$-th
and $B$-th price to increase its revenue.
\item
{\it Challenge 2: miner-user collusion.}
In a decentralized blockchain, it is easy for two or more parties
to form binding side contracts through smart contracts. 
A miner could collude with a user to increase 
the joint utility of the coalition, and 
the two can then split the gains with a binding side contract. 
In the aforementioned Vickrey auction 
example, the miner could alternatively ask the $B$-th bidder
to raise its bid to be infinitesimally smaller 
than the $(B-1)$-th bid, and then split its gains with the 
$B$-th bidder in a side contract.

Most prior works~\cite{zoharfeemech,functional-fee-market,roughgardeneip1559,roughgardeneip1559-ec} focused on 
miner-user rather than user-user collusion, 
likely for the following reason:
it is much
easier to facilitate miner-user  
rendezvous since the big miners are well-known. 
In comparison, users are ephemeral and thus user-user rendezvous
is much more costly to facilitate.
\end{itemize}


With these challenges in mind, prior works~\cite{zoharfeemech,roughgardeneip1559} have 
suggested the following desiderata for 
a ``dream'' transaction fee mechanism:
\begin{enumerate}[leftmargin=6mm]
\item {\it User incentive compatibility (UIC).}
Assuming that the miner implements the mechanism honestly, 
then following the honest bidding strategy or truthful bidding
should be a dominant strategy for the 
users\footnote{Roughgarden~\cite{roughgardeneip1559,roughgardeneip1559-ec}'s 
definition requires 
the honest strategy of the user (not necessarily truthful bidding)
be the dominant strategy, 
but he also pointed out that 
one can always convert such mechanisms to one where truthful
bidding is dominant due to the revelation principle (see
footnote 14, page 11 of \cite{roughgardeneip1559-ec}).
Therefore, we require truthful bidding to be a dominant  
strategy without loss of generality.
}.  

\item {\it Miner incentive compatibility (MIC).}
A miner's dominant strategy should be to implement the prescribed
mechanism faithfully. 
\item {\it Miner-user side contract proofness ($c$-SCP).}
No coalition of the miner and 
up to $c$ users can increase their joint utility through any deviation.
In the above, $c$ is a parameter that specifies an upper bound
on the coalition's size. The larger the $c$, the more side contract resilient.
Note that it is generally harder for a miner 
and a large number of users to engage in  
a side contract, than, say, a miner and a single user.
\end{enumerate}

To the best of our knowledge, 
all prior works~\cite{zoharfeemech,yaofeemech,functional-fee-market,roughgardeneip1559} 
fall short of achieving all three properties at the same time --- see 
Section~\ref{sec:related}
for more detailed discussions on these prior works.
The closest we have come to achieving all three properties 
is Ethereum's recent EIP-1559~\cite{eip1559} proposal.  
The very recent work of Roughgarden~\cite{roughgardeneip1559}
showed that (a close variant of) EIP-1559 
can achieve all three properties 
{\it assuming that the block size is infinite} (or more precisely, assuming
that the base fee is set high enough such that the number of transactions 
willing to pay the base fee is upper bounded by the block size). 
However, in practice, congestions do occur, e.g.,
when there is a peak in demand or when the mining power drops
causing inter-block time to be longer~\cite{gapgame,bitcoininstability}.
It is also well-understood that we cannot arbitrarily increase
the block size since this would  harm the security 
of the underlying consensus~\cite{backbone,pss17,rethinking-csf17}.
Roughgarden~\cite{roughgardeneip1559,roughgardeneip1559-ec}
argued that when there is congestion,  
EIP-1559 acts like a first-price auction and therefore fails to satisfy UIC, i.e., 
strategic bidding could improve an individual user's utility.

\paragraph{\underline{Open question 1:}}
With all these failed attempts, it is natural to ask: 
{\it is it actually feasible to have a ``dream'' transaction fee mechanism
that satisfies all three properties simultaneously?}
Is the community's lack of success 
so far due to a more fundamental mathematical
impossibility?
Roughgarden also 
raised this as a major open question 
in his recent work~\cite{roughgardeneip1559,roughgardeneip1559-ec}.

\paragraph{\underline{Open question 2:}}
If there is indeed a mathematical impossibility, then the natural next question to
ask is: are the current incentive compatibility notions 
overly stringent? If so, 
can we relax the incentive compatibility notion  
to circumvent the impossibilities?

\paragraph{TFM design space enriched by new elements.}
Transaction fee mechanisms often employ a couple interesting features
that are not commonly used in classical mechanisms.
For example, Ethereum's EIP-1559~\cite{eip1559,roughgardeneip1559} 
suggested the usage of 
a {\it burn rule}, where part to all of the 
fees collected from the confirmed transactions
may be ``burnt'' rather than paid to the miner.
Earlier work also considered ``complete burning'' of payments 
in environments 
where money transfer is not possible~\cite{moneyburnmech}. 
By contrast, in TFM, the burning may be partial.

Another design consideration that  
is being debated in the community 
is whether we should allow blocks to contain unconfirmed 
transactions that are just there to ``set the price''. 
Although this approach has been employed by 
some suggested mechanisms~\cite{zoharfeemech,yaofeemech} (see
the paragraph before Section 1.1 in Lavi et al.~\cite{zoharfeemech} 
\elaine{refer to appendix?}), 
an argument against it is that real estate on a blockchain is scarce --- 
therefore, we ideally do not 
waste space including unconfirmed transactions.  
An intriguing question is the following:

\paragraph{\underline{Open question 3:}}
Do these elements idiosyncratic to blockchains actually make a difference
in the design of transaction fee mechanisms?
Can they help achieve incentive compatible mechanism designs 
that would otherwise be impossible?

\ignore{
\subsubsection{to move}
Of course, for the decentralized parties to form side contracts 
would require a rendezvous process. The rendezvous process
is easy to implement for miner-user coalitions 
since the big miners are well-known in major blockchains today. 
In comparison, users are ephemeral, 
and thus user-user rendezvous is much more costly to implement.
Perhaps for this reason, 
prior works on transaction 
fee mechanisms~\cite{zoharfeemech,roughgardeneip1559,roughgardeneip1559-ec} 
focused mostly on resilience to miner-user collusion (as opposed to user-user collusion).
}

\subsection{Our Results and Contributions}

\subsubsection{Impossibility of a Having a ``Dream'' Transaction Fee Mechanism}
We prove an impossibility result 
(Theorem~\ref{thm:intromain})
showing that assuming finite block size, there is no non-trivial 
transaction fee mechanism (TFM) that satisfies 
UIC and 1-SCP, where 1-SCP means resilience against side contracts
between the miner and {\it a single} user.
Consequently, there is also no non-trivial TFM that satisfies
all three desired properties.

\begin{theorem}[Impossibility of a ``dream'' transaction fee mechanism (informal)]
Suppose that the block size is finite.
There does not exist a non-trivial, single-parameter transaction fee mechanism (TFM)
that simultaneously satisfies UIC and 1-SCP.
Moreover, this impossibility holds 
for both deterministic and randomized
mechanisms.
\label{thm:intromain}
\end{theorem}

Another way to understand Theorem~\ref{thm:intromain}
is the following: the only TFM that satisfies UIC and 1-SCP 
simultaneously is the trivial mechanism that always confirms
nothing and pays the miner nothing.
Our impossibility result  
holds no matter whether transactions take up the same
space or not. For example, it also holds for Ethereum's gas model
where each transaction may consume a different amount of space. 

\subsubsection{Definitional Contribution: Incentive Compatibility under $\gamma$-Strict Utility}
\label{sec:introweakic}
While our aforementioned impossibility result
paints a pessimistic outlook, we observe 
that the previously formulated incentive compatibility 
notions appear too draconian and somewhat flawed.
So far, almost all prior  
works~\cite{zoharfeemech,yaofeemech,functional-fee-market,roughgardeneip1559,roughgardeneip1559-ec} 
model the TFM in a standalone setting, where the players 
are myopic and care only about their gain or loss in the current auction instance.
In this setting, if 
a strategic player (which is either a user, a miner, or a miner-user coalition)
injects a fake transaction whose true value is $0$, 
or if it overbids (i.e., bids more than the transaction's true value),  
we assume that the offending transaction is free of charge as long as it is not
confirmed in the present block --- since an unconfirmed transaction need
not pay any fees.

In practice, however, the TFM is executed repeatedly as blocks
get confirmed. 
In a Bitcoin-like cryptocurrency, 
any transaction that has been posted to the network 
cannot be retracted even if unconfirmed 
in the present block\footnote{In some smart-contract capable blockchains such as 
Ethereum, it might be possible for a transaction to declare
a desired block number such that it is no longer considered valid for later blocks.
For our feasibility result, 
we shall focus on cryptocurrencies such as Bitcoin, as well as any cryptocurrency system
where retracting a posted transaction is not possible.
Note that posting a transaction that conflicts with the offending
transaction later does not fundamentally remove the cost, since an honest miner
may include the one with the higher fee.
}.
In particular, a fake or overbid transaction could 
be confirmed in a future block, and thus the strategic player 
would end up paying fees to the future block, potentially mined by a different miner.
For example, consider the Vickrey auction example again, where we include
the $B$ highest bids in the block, among which the top $B-1$ 
are confirmed and pay the $B$-th price.
Suppose that all payment goes to the miner.
In this case, the miner may want to inject a fake transaction whose bid
is in between the $(B-1)$-th and the $B$-th price, to increase its revenue.
In prior works as well as our aforementioned impossibility result, 
we assume that injecting this fake transaction is free because it is unconfirmed.
However, in practice, the injected transaction may be 
confirmed and paying fees in a future block.

A natural question is whether we can capture this cost of cheating
in our model, and thus circumvent the impossibility.
In our new approach, we still model the TFM as a single-shot auction, but
we want to more accurately charge the cost of cheating in the utility model.
Unfortunately, we face a notable challenge: 
accurately predicting the cost of cheating is difficult,
since what the offending transaction actually pays in the future 
depends on the environment, e.g., what other users are bidding, as
well as the mechanism itself.

\paragraph{Defining $\gamma$-strict utility.}
To make progress, we take the following approach. We first ask what is the worst-case cost 
of cheating. This is when the overbidding or fake transaction 
that is unconfirmed in the present ends up paying its full bid in the future,
thus incurring a cost as high as the difference
between the bid and the true value of the transaction.
This setting makes it the hardest for the 
cheater to gain, and the easiest for the mechanism designer to 
satisfy incentive compatibility.
Asking whether there is a mechanism that 
satisfies incentive compatibility under the most strict cost model  
is equivalent to asking: can we at least design mechanisms that defend 
against {\it paranoid} strategic players who only 
want to deviate if there is a sure chance of gain and no chance of losing.
Understanding the feasibility of mechanism design under the worst-case 
cost can shed light on whether this is a worthwhile direction.
Further, it is also useful to adopt the worst-case cost 
model in proving lower   
bounds, since that makes the lower bounds stronger. 

Next, we generalize the cost model and imagine that 
in reality, the offender only needs to pay $\gamma$ fraction of the worst-case cost,
where $\gamma \in [0, 1]$ is also called the {\it discount} factor.
This generalization may be useful because in practice, we can often estimate 
the cost of cheating from the recent historical data, or even  
adjust the choice of $\gamma$ dynamically over time (similar to 
how Bitcoin adjusts the mining difficulty or 
how Etheurem's EIP-1559 dynamically adjusts 
their base fee based on recent 
historical data~\cite{eip1559,roughgardeneip1559,roughgardeneip1559-ec}). 
Another motivation for introducing $\gamma$
is to enable a knob that allows us to engineer a tradeoff between
the efficiency of the mechanism  
and its resilience to strategic behavior. In this sense, {\it estimating
the exact $\gamma$ is not too important}. As mentioned, setting $\gamma = 1$
gives reasonable incentive compatibility guarantees, namely, 
against {\it paranoid} players.


A mechanism that satisfies UIC (or MIC, $c$-SCP, resp.) under 
this cost model is also said to satisfy
UIC (or MIC, $c$-SCP, resp.) under {\it $\gamma$-strict utility}. 
Specifically, when $\gamma = 0$, there is no cost of cheating --- in this case, 
our new incentive compatibility notions would degenerate to the previous notions.
When $\gamma =1$, this is when we are charging the worst-case cost for cheating.
Since we are often particularly interested in  
the case of $\gamma = 1$ (e.g., when proving lower bounds), 
for convenience, a mechanism that satisfies
UIC (or MIC, $c$-SCP, resp.) under
$1$-strict utility  
is also said to satisfy
{\it weak UIC (or weak MIC, $c$-weak-SCP,  resp.)}.

We present our new incentive compatibility notions formally in Section~\ref{sec:weakic}.

\ignore{
\paragraph{Definitional contribution: weak incentive compatibility.}
As a result, a strategic player who is paranoid and risk-averse  
may be deterred from strategic overbidding or injection of fake transactions  
for fear or losing fees to a future block.
We therefore 
suggest a relaxed notion called weak incentive compatibility,  
which captures the potential cost of an overbid/fake and unconfirmed transaction. 
Specifically, 
if an overbid/fake transaction is unconfirmed in the present, 
the strategic player will assume the worst case, i.e., it will cost fees
as high as the bid 
in a future block.
Based on this, we redefine the 
utility function of the strategic player (see Section~\ref{sec:weakic} 
for more details).
We then define 
{\it weak UIC}, {\it weak MIC}, and {\it 1-weak-SCP} 
in the same way as before, except
that we now adopt the new utility function. 
}
\subsubsection{The Mathematical Structure of Incentive Compatibility 
under $\gamma$-Strict Utility}

\noindent\textbf{The burning second-price auction.}
Using our new $\gamma$-strict utility notion, we can circumvent
the aforementioned impossibility (Theorem~\ref{thm:intromain}).
Specifically, we describe a new mechanism called 
the burning second-price auction (see Section~\ref{sec:burn2ndprice})
that achieves UIC, MIC, and $c$-SCP under $\gamma$-strict utility for
any $\gamma \in (0, 1]$, and any choice 
of coalition resilience parameter $c \geq 1$.
The mechanism is randomized, and  
one can view the parameters $c$ and $\gamma$ that allow
us to tradeoff the degree 
of incentive compatibility and
the efficiency of the mechanisms (in terms the expected number of  
bids confirmed).


\ignore{
\begin{mdframed}
\begin{center}
{\bf The burning second-price auction} 
\end{center}
\paragraph{Parameters:} the block size $B$,  
and $0 < k' \leq k < B$
such that $k + k' = B$, where $k$ denotes the number of confirmed
transactions per block, and $k'$ denotes the number of unconfirmed
transactions in a block that are used to set the price and miner revenue.

\paragraph{Mechanism:}
\begin{itemize}[leftmargin=5mm,itemsep=1pt]
\item 
Choose the $B$ 
highest bids to include in the block.
The highest $k$
bids are considered confirmed, and they each pay the $(k+1)$-th price.
\item 
The miner is paid the sum of the $(k+1)$-th 
to the $B$-th prices.
All remaining payment collected from the confirmed
transactions is burnt.
\item 
If the block is not fully filled, any remaining empty slot is treated
as a bid of $0$.
\end{itemize}
\end{mdframed}

In Section~\ref{sec:weakic}, we shall prove that the burning second-price auction
indeed satisfies weak UIC, weak MIC, and 1-weak-SCP, as stated in the following
theorem:
}
Formally, we prove the following theorem:

\begin{theorem}[Burning second price auction]
For any $\gamma \in (0, 1]$ and any $c \geq 1$, 
there exists a TFM that 
satisfies UIC, MIC, and $c$-SCP under $\gamma$-strict utility.
Further, the TFM can support any finite block size,  
and except for the case when $c = 1$ and $\gamma = 1$, the TFM 
is randomized (the paragraph ``Necessity of randomness'' below 
gives more explanations about randomized TFMs).
\label{thm:intro-upper}
\end{theorem}
Our burning second-price auction
enables a smooth trade-off between the efficiency 
of the TFM and its resilience to strategic deviations.
In particular, as $\gamma$ approaches $0$, the number
of transactions confirmed by the TFM approaches $0$ as well --- this is
in some sense inherent due to our earlier impossibility result.

\paragraph{Necessity of randomness.}
As mentioned in Theorem~\ref{thm:intro-upper}, except for the special case
$c = 1$ and $\gamma = 1$, our burning second-price
auction is randomized. 
In general, a randomized TFM allows 
the miner and/or the blockchain to employ random coins
to decide which bids to include, to confirm, 
the payment of each confirmed bid, as well as the miner revenue.
In particular, our burning second price auction (Theorem~\ref{thm:intro-upper})
employs trusted on-chain randomness
to pick a random subset of eligible, included transactions to confirm.
Although unbiased and unpredictable on-chain randomness can be generated
using standard cryptographic 
techniques~\cite{Cachin00randomoracles,randpiper,spurt}, 
such coin toss protocols introduce some extra overhead, and ideally we would
like to avoid them.
Unfortunately, we prove a lower bound that for $c \geq 2$, 
randomness is 
necessary to achieve UIC and $c$-SCP 
for any $\gamma \in [0, 1]$, as long as the mechanism
is ``useful'' in the sense that it sometimes confirms at least $2$ bids.

\begin{theorem}[Necessity of randomness for weak incentive compatibility]
Consider an arbitrary deterministic TFM and assume finite block size.
Suppose that there exists
a bid vector such that the TFM confirms at least two bids.
Then, the TFM
cannot satisfy both weak UIC and $2$-weak-SCP simultaneously.
\label{thm:intro-lb-weak}
\end{theorem}

In the above lower bound, 
the restriction that the mechanism must sometimes confirm $2$ bids  
is necessary. Specifically,
we construct a deterministic mechanism called the solitary mechanism
(Appendix~\ref{section:solitary})
that always confirms a single bid, and  
achieves weak UIC, weak MIC, and $c$-weak-SCP for any $c \geq 1$. 

\subsubsection{Understanding New Design Elements for TFM}

As mentioned, 
TFMs often employ
a couple new design elements that are not so common in classical mechanisms. 
First, the mechanism can employ a ``burn rule'', 
which allows part to all of the users' payment to be ``burnt''
on the blockchain, and not
paid to the miner of the present block.
For example, Ethereum's EIP-1559 makes critical 
use of such a burn rule~\cite{eip1559,roughgardeneip1559,roughgardeneip1559-ec}. 
Second, some prior works~\cite{zoharfeemech,yaofeemech} 
have suggested including transactions in a block
that are not confirmed eventually, but serve the role of setting
the price for others. 
For example, even though we know that the  
Vickrey auction is not an awesome auction in a decentralized environment, 
hypothetically, imagine we want to implement the Vickrey auction on a blockchain. 
This would require the block to include $B$ bids, 
among which only the top $B-1$ are eventually confirmed, whereas the $B$-th 
bid is  included only to set the price.
Moreover, our own burning second-price auction (Theorem~\ref{thm:intro-upper})
also includes some transactions in the block that have no chance
of being confirmed, but are just there to set the price.

Do these new design elements make a difference in TFM design,
and can they help us achieve incentive compatibility designs  
that would otherwise be impossible?
We give a nuanced 
answer to this question. First, we point out that our 
earlier impossibility results (Theorems~\ref{thm:intromain}
and \ref{thm:intro-lb-weak})
hold even when the TFM is allowed to employ both of these design elements. 

\ignore{
As mentioned, the decentralized nature raises new challenges
for the design of transaction fee mechanisms.  
On the other hand, 
there are also features of the blockchain that can potentially help
us in designing TFM.
}
\ignore{
Our results reveal a nuanced answer to this question.
On one hand, we show that 
assuming finite block size, 
the impossibility stated in Theorem~\ref{thm:intromain}
still holds even for (possibly randomized) TFMs 
with a burn rule:
}

On the other hand, we also show scenarios in which these new design
elements do make a difference.
Specifically, we prove the following results.

\paragraph{The burn rule is critical to Ethereum's EIP-1559.}
\ignore{
On the other hand, we show 
having a burn rule does make a difference 
assuming infinite block size. In particular, 
this implies that {\it the burn rule indeed makes a fundamental difference 
in a mechanism like EIP-1559}.
}
Recall that Roughgarden~\cite{roughgardeneip1559,roughgardeneip1559-ec}
argued that {\it assuming infinite block size},
Ethereum's EIP-1559 approximates 
a simple 
``posted price, burn all''  auction: 
there is an a-priori fixed price tag $r$, and
anyone who bids at least $r$ would get their transaction confirmed, 
paying only $r$. All users' payment 
is burnt, and 
the miner gets nothing.
Roughgarden~\cite{roughgardeneip1559,roughgardeneip1559-ec}
also proved that this simple 
``posted price, burn all''  auction
would indeed satisfy UIC, MIC, and $c$-SCP for any $c \geq 1$, 
assuming infinite block size. 
\ignore{
\begin{corollary}[Impossibility for TFMs with burn (informal)]
Suppose that the block size is finite.
The impossibility result stated in Theorem~\ref{thm:intromain}
holds even for (possibly randomized) TFMs with a burn rule. 
\label{cor:intromain}
\end{corollary}
}
We show that without the burn rule, even under infinite block size, 
the only way for a TFM to satisfy both UIC and $1$-SCP
is for users to always pay nothing.
More generally, we prove that even when assuming infinite block size 
and whether we allow a burn rule or not, 
any (randomized) TFM 
that satisfies both UIC and 1-SCP must always pay the miner nothing:

\begin{theorem}[The burn rule makes a difference assuming infinite block size]
Any (randomized) TFM 
that satisfies both UIC and 1-SCP must always pay the miner nothing.
In other words, any non-trivial (randomized) TFM that does not always burn  
all payment cannot be both UIC and 1-SCP.
This impossibility 
holds regardless of whether the block size is infinite or finite, and   
regardless of whether the TFM has a burn rule or not\footnote{In fact, 
in our actual proof, we prove Theorem~\ref{thm:introburn} first,
which is then used as a stepping stone towards proving
Theorem~\ref{thm:intromain}.}.
\label{thm:introburn}
\end{theorem}

\ignore{
Based on Theorem~\ref{thm:intromain}, 
for TFMs {\it without} a burn rule, to satisfy both UIC and 1-SCP, 
it is necessary that users always pay nothing.
In other words, 
}
As a direct implication, 
if we want a mechanism like EIP-1559 that has non-trivial user payment
and satisfies all three properties in the infinite block size regime, 
 a burn rule is necessary.
Therefore, Theorem~\ref{thm:introburn} and  
Roughgarden's result~\cite{roughgardeneip1559,roughgardeneip1559-ec}
together show that 
having a burn rule 
does make a difference assuming 
infinite block size.

\begin{remark}[Regarding burning]
During uncongested periods, EIP-1559 burns almost all user payment 
since the tips should approach $0$ during these times.
The reader may be concerned why miners would still be incentivized
to mine in Ethereum. This is because 
the miners always get a constant block reward which has no effect
on our game theoretic modeling and thus is omitted in our utility definition. 
Burning part to all payments
may be used to introduce deflation or regulate inflation in cryptocurrencies, 
and the burnt payment can also be repurposed 
to offset block rewards for future miners~\cite{roughgardeneip1559}. 
Since such usage
has no effect on our game-theoretic modeling, we simply assume
the part of payment not directly paid to the miner is ``burnt'', just like
Roughgarden~\cite{roughgardeneip1559,roughgardeneip1559-ec}.
\end{remark}

\paragraph{Necessity for blocks to contain unconfirmed transactions.}
\ignore{
Observe that 
in the burning second price auction, not all transactions included in the block
are confirmed. Specifically, the $k$-th bid is included in the block 
but left unconfirmed. In some sense, it is included just to set the price.
}
In the cryptocurrency 
community, there is an ongoing debate whether it is a waste of space for blocks 
to contain unconfirmed transactions.
We argue that the ability for a block to contain unconfirmed transactions  
could indeed make a difference for the mechanism designer.
Specifically, we prove a corollary showing that 
if one insists on confirming all transactions included in a block, then 
even with the weak incentive compatibility notion, 
it is still impossible to construct any 
non-trivial (possibly randomized) TFM 
that satisfies weak UIC 
and 1-weak-SCP.

\begin{corollary}[Allowing unconfirmed transactions in a block can make a difference]
Suppose that all transactions in a block must be confirmed. 
Then, there is no non-trivial (possibly randomized) TFM that satisfies
weak UIC and 1-weak-SCP simultaneously.
\label{cor:intro-incleqconf}
\end{corollary}

\paragraph{Additional results.}
In the appendices, we additionally show 
a variant of Theorem~\ref{thm:intro-lb-weak}
that says if the TFM is not allowed  
to have a burning rule, then Theorem~\ref{thm:intro-lb-weak} holds
even when the block size is infinite.

\paragraph{Additional related work.}
Transaction fee mechanism is incomparable to the credible
auction model proposed by Akbarpour and Li~\cite{credibleauction}
--- see a detailed discussion in \cref{sec:related}.
In \cref{sec:related}, we also discuss additional related work.

%% file: roadmap-new.tex
\section{Technical Roadmap}

\subsection{Transaction Fee Mechanism and Incentive Compatibility}
\label{sec:roadmap-tfm}
In a transaction fee mechanism (TFM), we are selling slots
in a block to bidders who want to get 
their transactions included and confirmed  
in the block.
For simplicity, we assume that all slots are identical commodities, and 
we  often use the terms ``transaction'' and ``bid'' interchangeably.
For convenience, we assume that each bid comes from a different user.

\paragraph{Transaction fee mechanism.}
A transaction fee mechanism (TFM)
includes the following rules:
\begin{itemize}[leftmargin=5mm,itemsep=0pt,topsep=2pt]
\item 
An {\it inclusion rule} executed by the miner. 
Given a bid vector $\bfb = (b_1, b_2, \ldots, b_m)$,
the inclusion rule decides which of the bids to include in the block;
\item  
A {\it confirmation rule} executed by the blockchain. 
The confirmation rule chooses a subset of the included bids 
to be confirmed.
In the most general form, 
not all transactions included in the block are necessarily confirmed, and  
only confirmed transactions are considered final, 
i.e., the money has been transferred to the merchant's account
and the merchant can now provide the promised service.
\item 
A {\it payment rule} and a {\it miner revenue} rule executed by the blockchain,
which decides (using only information recorded
in the block) how much each confirmed bid pays, and  
how much revenue the miner gets.
Any (possibly included) transaction that is not confirmed pays nothing.
Furthermore, we assume that the miner's revenue is upper bounded
by the total payment collected from all 
confirmed bids\footnote{All existing major cryptocurrencies such as 
Bitcoin and Ethereum satisfy this assumption (ignoring the fixed 
block reward paid to the miner on top of the fees).
We discuss what it might mean to relax 
this assumption in Section~\ref{sec:conclusion} and leave this as an exciting
future direction.
}.
In particular, if the miner's revenue is strictly smaller
than the total payment of all bids, then we often say that part of the payment  
is {\it burnt}. 
\end{itemize}

In our model, a strategic miner (possibly colluding with some users) 
may not implement the honest inclusion rule, if  
deviating can benefit the miner (or coalition).
However, the blockchain is trusted 
to implement the confirmation, payment, and miner-revenue rules honestly. 

In comparison with Roughgarden's model~\cite{roughgardeneip1559,roughgardeneip1559-ec}, 
we explicitly distinguish the inclusion rule from the 
confirmation rule in our modeling. 
By contrast,
Roughgarden's model calls the union of the inclusion rule and the  
confirmation rule the {\it allocation rule}.
Making the distinction between 
the inclusion and confirmation rules explicit is 
useful for us since we want to tease out the fine boundaries
between feasibility and infeasibility, depending on whether 
the block size is finite or infinite.

\paragraph{Strategy space and incentive compatibility.}
A {\it strategic player} can be a user, a miner, or the coalition of the miner
and up to $c$ users. 
The strategic player can deviate in the following ways:
1) if one or more users are involved, then some of the users
can decide to bid untruthfully, possibly {\it after} examining some or all other bids; 
2) the strategic player can inject fake bids, possibly
 after examining some or all other bids;
and 3) if the miner is involved, then
the miner may not implement the inclusion rule honestly.

Every user has a true value for its transaction to be confirmed.
If confirmed, its utility is 
its true value minus its payment. An unconfirmed user has utility $0$.
The miner's utility is its revenue.
If the miner colludes with some users, 
the coalition's joint utility is the sum of the utilities of all coalition members. 

\paragraph{Incentive compatibility.}
The honest strategy for a user is to bid its true value.
The honest strategy for a miner is to implement the correct inclusion rule.
A TFM is incentive compatible for a strategic player 
iff 
deviating from the honest strategy cannot increase the strategic player's expected
utility; i.e., 
playing honestly is the strategic player's best strategy (or one
of the best strategies).
A TFM is said to be user incentive compatible (UIC), 
if it is incentive compatible for any individual user.
A TFM is said to be miner incentive compatible (MIC), 
if it is incentive compatible for the miner.
Finally, 
a TFM is said to be $c$-side-contract-proof, 
if it is incentive compatible for any coalition consisting
of the miner and at least $1$ and at most $c$ users.
The notions UIC, MIC, and $c$-SCP are incomparable as shown in Appendix~\ref{sec:ic-compare}.

Note that in a blockchain environment, user-user coalitions are much harder to form:  
since users are ephemeral, rendezvous
between them is challenging. By contrast, there are typically a stable set
of big miners which makes miner-user rendezvous easy. 
For this reason, most works in this space are more interested in defending
against miner-user rather than user-user coalitions.

\begin{remark}
In this paper, for our upper bounds, 
we assume that all transactions are equal size and 
we do not consider the scenario where the utility may depend
on the position in the block~\cite{flashboy}.
However, we stress that our lower bounds hold even when transactions are not of the same
size and even when utility may be an arbitrary 
function of the position in the block.
Similarly, for our upper bounds, we 
assume that a user's utility depends only on whether it is confirmed
and how much it pays, and we do not consider the case where a user's utility may depend on 
whether someone else's transaction being confirmed or not~\cite{LedgerHedger, massExitAttack}.
\end{remark}

\subsection{Impossibility of a ``Dream'' TFM under Finite Block Size}
\label{sec:roadmap-main-lb}
We now sketch how to prove Theorem~\ref{thm:intromain}, that is, 
assuming finite block size, no non-trivial TFM can achieve 
UIC and 1-SCP at the same time.
We shall first sketch how the proof works for deterministic TFMs, then
we explain how to generalize the proof to the randomized case. 

\paragraph{Deterministic case: miner has $0$ revenue.}
Recall that if a TFM satisfies  
UIC, it must respect the constraints imposed by the 
famous Myerson's Lemma~\cite{myerson}. For deterministic mechanisms, 
this means that the confirmation decision is {\it monotone}, and moreover,
every confirmed bid 
pays {\it the minimum price it could have bid and still remained confirmed},
assuming everyone else's bids remain the same.

To prove Theorem~\ref{thm:intromain}, we go through an intermediate stepping
stone: we shall actually prove  
Theorem~\ref{thm:introburn} first, that is, 
any TFM that satisfies both UIC and $1$-SCP {\it must always pay the miner nothing},
regardless whether the block size is finite or infinite.
Henceforth, let $\mu(\bfb)$ denote the miner 
revenue under the bid vector $\bfb$.
We use $p_i(\bfb)$ to denote user $i$'s payment under $\bfb$, and if
user $i$ is unconfirmed, $p_i(\bfb) = 0$.

Consider an arbitrary deterministic TFM
that is UIC and $1$-SCP.
Consider an arbitrary bid vector $\bfb = (b_1, \ldots, b_m)$
and we want to argue that the miner has $0$ revenue under $\bfb$.
To do this, we want to lower each user's bid to $0$ one by one,
and argue that the miner revenue is unaffected in this process.
If this is the case, we can show that the miner revenue is $0$ under $\bfb$, since 
at the very end of this process, when we have lowered
everyone's bid to $0$, the miner revenue must be $0$.

\ignore{
First, we show that 
if user $i$ bids anywhere in between $[p_i, \infty]$ which 
allows user $i$ to be confirmed
by Myerson's Lemma, then 
the miner revenue must be unaffected.
If this is not true, then suppose 
that the miner and user $i$ form a coalition, and 
suppose that user $i$'s true value is at least $p_i$.
In this case, user $i$ should bid not necessarily its true value, the 
amount that maximizes the miner revenue in the range $[p_i, \infty]$. 
The coalition gains under this strategic behavior since 
user is indifferent to any bid in this range, whereas the miner's utility 
is maximized, and this violates $1$-SCP. 
}

It suffices to prove the following.
Let $\bfb = (b_1, \ldots, b_m)$ be an arbitrary bid vector
and $i \in [m]$ be an arbitrary user.
We want to show that $\mu(\bfb) = \mu(\bfb_{-i}, 0)$.
First, we show that if a user changes its bid such that its confirmation status 
remains unaffected, then the miner revenue should stay the same
(Claim~\ref{clm:inconsequentialbidchange}). 
If this is not true, then the miner and the user can collude,
and there is a way for the user to bid untruthfully without affecting
its confirmation status and thus its utility, but increasing the miner
revenue. Overall, the coalition strictly gains and this violates $1$-SCP.
Suppose that $p_i$ is the minimum price that some user
$i$ could bid to let it be confirmed, assuming that everyone else is
bidding $\bfb_{-i}$.
The above means that if user $i$ bids anywhere between
$[p_i, \infty]$ such that it remains confirmed, 
then the miner revenue 
is unaffected. Similarly, if user $i$ bids
anywhere between $[0, p_i)$ such that it is unconfirmed, 
then the miner revenue is unaffected too.

It remains to rule out the possibility that
there is a sudden jump in miner revenue, when user $i$ lowers
its bid from $p_i$ to $p_i - \epsilon$ for an aribitrarily small $\epsilon$.
Suppose for the sake of contradiction 
that there is a sudden $\Delta > 0$ increase in the miner 
revenue when user $i$ lowers
its bid from $p_i$ to $p_i - \epsilon$ 
(and the proof for the other direction is similar). 
From what we proved earlier, 
the entire jump of $\Delta$ must occur within an arbitrarily
small interval $p_i$ and $p_i -\epsilon$, 
and in particular, we may assume that $\epsilon < \Delta$.
\ignore{
We next show that 
$\mu(\bfb_{-i}, p_i) = \mu(\bfb_{-i}, 0)$ which is sufficient to prove the statement. 
Below we rule out the case $\mu(\bfb_{-i}, p_i) > \mu(\bfb_{-i}, 0)$,
since the other case $\mu(\bfb_{-i}, p_i) < \mu(\bfb_{-i}, 0)$
has a similar proof.
Let 
$\epsilon < \Delta := \mu(\bfb_{-i}, p_i) -  \mu(\bfb_{-i}, 0)$
be an arbitrarily small number.
We have shown that
$\mu(\bfb_{-i}, 0) =  \mu(\bfb_{-i}, p_i - \epsilon)$.
This means that if user $i$ bids $p_i$ instead of $p_i-\epsilon$,
the miner can gain $\Delta > \epsilon$ in utility. 
}
In this case, if the miner colludes with user $i$ whose true value is
actually $p_i-\epsilon$, the user 
should bid $p_i$ instead. 
This way, the miner's gain $\Delta$
outweighs the user's loss $\epsilon$, and the coalition strictly gains.
This violates $1$-SCP\footnote{An anonymous reviewer suggested the following alternative way
to think about the proof, 
assuming that the user utility and miner revenue
functions are differentiable (and our proof need not make this assumption). 
Since bidding the true value $v$
maximizes the user's utility, the derivative of the user's utility as a function
of bid is zero at the true value $v$. 
This means that the miner's 
revenue must have derivative $0$ at the bid $v$, since otherwise
the user-miner coalition can profit by deviating.
Observe also that the above argument must hold for any choice of $v$, we conclude
that the miner revenue must be unaffected by the user's bid. 
}.
A formal presentation of this proof can be 
found in Section~\ref{sec:trivial-miner-rev}.

\paragraph{Theorem~\ref{thm:introburn} + finite block size $\Longrightarrow$
Theorem~\ref{thm:intromain}.}
Once we have proven Theorem~\ref{thm:introburn}, 
i.e., the miner always has $0$ revenue, 
we can now throw in the finite block size assumption, 
to prove Theorem~\ref{thm:intromain}.
We show it for the deterministic case below.
Specifically, suppose there is a bid vector $\bfb = (b_1, \ldots, b_m)$ 
under which some bid $b_i$ is confirmed where $i \in [m]$.
Now, imagine that the real world
actually consists of the bids $\bfb$ 
plus sufficiently many users bidding $b_i + \epsilon$, such that 
the number of users bidding $b_i + \epsilon$ exceeds the block size.
We know that one of the users
bidding $b_i + \epsilon$ must be unconfirmed --- let us call this user
$u$.
The miner can now form a coalition with $u$, and ask $u$ to bid $b_i$ instead.
The miner can now pretend that the world consists of the bid vector
$\bfb$
where $b_i$ is replaced with $u$'s bid, and run the honest mechanism. 
This helps the user $u$ get confirmed  
and gain a positive utility, and meanwhile, the miner 
itself always gets $0$ revenue no matter what it does.
Thus, overall, the coalition strictly gains, which violates $1$-SCP.


\ignore{
Next, we argue that user $i$ can change its bid in between $(p_i, 0]$ 
without affecting the miner revenue.
Suppose that user $i$'s true value is $p_i$. Bidding anywhere
in this range gives user $i$ utility $0$, since it is either
unconfirmed or confirmed but paying its true value $p_i$.
If bidding in the range $[p_i, 0]$
}

\paragraph{Generalizing to randomized TFMs.}
At a high level, our earlier impossibility proof 
for {\it deterministic} TFMs use Myerson's Lemma 
as a blackbox. Since the TFM is UIC, we argue that the mechanism
must fall within the solution space characterized by Myerson's Lemma.
Our proof then shows that the constraints imposed by Myerson conflict
with the requirements of 1-SCP. 
We want to generalize the impossibility to even randomized mechanisms,
where the inclusion rule, confirmation rule,  
payment and miner revenue rules are allowed to employ randomness.
For the randomized case, instead of following the same blueprint as before, 
we present an 
alternative proof that uses Myerson's Lemma (the randomized case) 
in a slightly non-blackbox manner --- we review 
 Myerson's Lemma generalized to the randomized case in Section~\ref{sec:myerson}.
Below, keep in mind that the notations  
$p_i(\bfb)$ and $\mu(\bfb)$ can be random variables.

We first give a slightly incorrect intuition. 
As a thought experiment,
imagine that the coalition of the miner and user  
$i$ forms a ``meta-user'' $i$.
Meta-user $i$ 's true value
is $v_i$, i.e., same as user $i$'s true value. 
Meta-user $i$'s payment is $p_i(\bfb) - \mu(\bfb)$. 
Observe that meta-user $i$'s true value minus its payment
is exactly the coalition's utility in the original TFM.
Now, imagine a ``meta-auction''
among a set of meta-users, 
where each meta-user $i$ is the coalition of 
the miner and the user $i$. 
Each meta-user's strategy space is either overbidding or underbidding. 
Since the original 
TFM satisfies 1-SCP, it must be that 
each meta-user does not want to overbid or underbid, i.e., 
the meta-auction is dominant strategy incentive compatible
for each meta-user. 
Now, we can apply Myerson's Lemma 
 to this meta-auction, 
and argue that each meta-user's 
payment $p_i(\bfb) - \mu(\bfb)$ must satisfy the unique payment rule 
stipulated by Myerson's Lemma.
However, since the original TFM is UIC, it must be that each user's payment
$p_i(\bfb)$ 
also satisfies the unique payment rule stipulated by Myerson's Lemma.
This gives us $p_i(\bfb) - \mu(\bfb) = p_i(\bfb)$, i.e., $\mu(\bfb) = 0$.

The above argument is slightly incorrect, though, since the unique payment rule 
of Myerson's Lemma
relies on the border condition that if a user bids $0$, it pays $0$.
When we consider the meta-auction, a meta-user's payment is
of the form $p_i(\bfb) - \mu(\bfb)$ --- and it is not immediately clear that this quantity
is $0$ (even though at the end of the proof, we can see that it is indeed $0$). 
It takes a little more work to make this intuition correct, 
and we give a formal proof below that makes slightly non-blackbox
usage of the proof of the Myerson's Lemma --- see Section~\ref{sec:randomized-lb}
for details.

The above proves Theorem~\ref{thm:introburn}
for the randomized case.
Similarly, we can now rely on Theorem~\ref{thm:introburn}
and additionally throw in the finite block size assumption
to get Theorem~\ref{thm:intromain}
for the randomized case. 
The proof of this is a little more complicated
than the deterministic case, and we defer the formal details to 
Section~\ref{sec:finite}.

\subsection{Incentive Compatibility under $\gamma$-Strict Utility}

\noindent\textbf{$\gamma$-strict utility.}
As observed earlier 
in Section~\ref{sec:introweakic}, the current modeling approach
does not charge for certain costs of cheating. Specifically,  
an overbid or fake transaction that is not confirmed in the present 
is incorrectly assumed to be free of cost.
We therefore refine the model by changing the utility definition
to account for this cost.
As mentioned, since the exact cost is hard to predict, we define 
a parametrizeable
utility notion called $\gamma$-strict utility, where the 
discount factor $\gamma \in [0, 1]$
serves as a knob to tune 
the tradeoff between efficiency and resilience --- see Section~\ref{sec:introweakic}
for more philosophical discussions of introducing
the $\gamma$ parameter and why it is a good idea even when we cannot
obtain an accurate estimate of $\gamma$. 

In comparison with the utility notion
introduced in Section~\ref{sec:roadmap-tfm}, 
the only difference here is that 
for any overbid or fake transaction that is not confirmed
in the present, we 
charge the strategic player 
$\gamma$ times the worst-case cost, 
where the worst-case cost
is the difference between the bid amount and the true value, since
the strategic player may end up paying the full bid
amount in a future block (of which it may not be the miner). 
We may assume that any fake transaction has a true value of zero.

We can define UIC, MIC, and $c$-SCP just like before but now
using the $\gamma$-strict utility  
notion.
The notions UIC, MIC, and $c$-SCP 
under $\gamma$-strict utility 
are incomparable for any $\gamma \in [0, 1]$ 
as shown in Appendix~\ref{sec:ic-compare}.

\input{burn-2nd-price}

\subsection{Necessity of Randomness for Weak Incentive Compatibility}
\label{sec:roadmap-lb-weak}
We present an informal roadmap for the proof of Theorem~\ref{thm:intro-lb-weak}.
Briefly speaking, if there exists a bid vector such that the TFM confirms at least two bids,
we say that the TFM is {\it 2-user-friendly}.
Then, any deterministic and 2-user-friendly TFM 
cannot satisfy weak UIC and 2-weak-SCP simultaneously. 
Recall that weak incentive compatibility corresponds to the case
when $\gamma = 1$. In other words, we are charging the worst-case cost
for cheating, and this makes our lower bounds stronger.

\paragraph{Myerson's lemma holds for deterministic and weak UIC mechanisms.}
Recall that Myerson's Lemma holds for any UIC mechanism.
Since we now are considering a more relaxed notion, namely, weak UIC, 
it may not be immediately clear that 
Myerson's Lemma still holds. 
Fortunately, we can prove that 
assuming {\it deterministic} and no random coins,  then even weak UIC  
mechanisms must satisfy the requirements imposed by Myerson's Lemma
(Fact~\ref{fct:myerson-weakuic}).
We stress that this observation is actually somewhat subtle, since
it is not too clear whether Myerson's Lemma
holds for {\it randomized} mechanisms that satisfy weak UIC.

\paragraph{Weak UIC + $2$-weak-SCP + $2$-user-friendly $\Longrightarrow$ several natural properties.}
Next, we establish a few natural structural properties
for any deterministic, $2$-user-friendly TFM that is both weak UIC and $2$-weak-SCP.
\begin{enumerate}[leftmargin=5mm]
\item 
All confirmed bids must pay the same, and thus there is a universal payment 
(Lemma~\ref{lemma:samePayment});
\item 
The mechanism must confirm the highest bids where the number of confirmed 
bids may depend on the bid vector (Lemma~\ref{lemma:ordered}); 
and 
\item  
The universal payment 
must be at least as high as the top unconfirmed bid (Lemma~\ref{lemma:unconfirmedPayment}).
In other words, anyone bidding strictly higher than the universal payment must be confirmed.
\end{enumerate}

\paragraph{Influence of an individual bidder.}
When we proved the impossibility 
for (strong) incentive compatibility, we used the fact 
that when an individual user moves its bid up or down, 
as long as its confirmation decision is unaffected, 
the user's own utility does not change. 
Now, due to $1$-SCP, the miner's revenue should be unaffected too. 
This statement is not entirely true any more now that
we have changed our utility definition.
In particular, if an unconfirmed user 
increases its bid while still remaining unconfirmed, there is now 
an extra cost to the user.
The key to proving Theorem~\ref{theorem:twoweakSCPwithburn}
is to understand how fast the universal payment and 
miner revenue can change as we change a single user's bid.
There are a few cases (stated informally below): 
\begin{itemize}[leftmargin=5mm]
\item 
{\bf Lemma~\ref{lemma:confirmInvariant}\footnote{We 
in fact need to use this lemma to prove the aforementioned natural properties.
}.}
If a confirmed user changes its bid such that it is still confirmed,
then the miner revenue is unaffected.
This can be shown using the same argument as 
in Section~\ref{sec:trivial-miner-rev} relying on weak UIC and 1-weak-SCP, 
since for a confirmed bid, the new and old utility notions coincide.

Additionally, using $2$-weak-SCP, we can show something even stronger:
if there are two confirmed bids $b_1$ and $b_2$ such that $b_1 > p$ where
$p$ is the universal payment, 
then, $b_1$'s confirmation status and the universal payment
amount 
are also unaffected when $b_2$ changes its bid as long as it remains confirmed. 

\item 
{\bf Lemma~\ref{lem:minerutilkto0}.}
If an individual user changes its bid by $\Delta$, 
then the miner utility cannot change by more than $\Delta$.
Roughly speaking, this is because even under our new utility notion, 
the extra cost to a user is at most $\Delta$ if it changes its bid by $\Delta$. 
If the miner revenue changed 
by more than $\Delta$, then the miner-user coalition has a deviating strategy  
that allows them to strictly gain.

\item 
{\bf Lemma~\ref{lemma:paychangeslow}.} 
If a user $i$ increases its bid from  
$0$ to $\Delta$, the universal payment cannot increase by more than $\Delta/2$.
Had it not been the case, then the coalition of a miner 
and two confirmed users can gain in the following way: 
if user $i$ actually bids $b_i > 0$, the miner can replace $b_i$ with a $0$-bid.
In this way, the two colluding users each pay a lot less, and due to 
the earlier Lemma~\ref{lem:minerutilkto0}, the miner's revenue does not change that much.
So overall, the coalition can strictly gain.

\item 
{\bf Lemma~\ref{lemma:paychangeslow2}.} 
If a user $i$ drops its bid from $b_i$ to $0$, then the universal payment
cannot increase by more than $b_i$. Otherwise, the miner can collude with one paying user, and 
suppose user $i$'s actual bid is $0$, but the miner changes it to a fake bid of $b_i$.
In this case, the paying user 
would pay a lot less which outweighs the cost to the miner is only $b_i$.
\end{itemize}

\begin{figure*}[t]
\[\begin{array}{cccccccl}
& & \text{\bf bids} & & & & \text{\bf universal payment}\\
b_1 > p, & b_2 > p, & \_, & \_, & \ldots, & \_  & p\\[-6pt]
& & & & & & & \biggl.\biggr\} \text{\ Lemma~\ref{lemma:confirmInvariant}}\\[-7pt]
{\color{blue}\Gamma \ \text{(big)}}, & {\color{blue}\Gamma \ \text{(big)}}, & \_, & \_, & \ldots, & \_  & p\\[-6pt]
& & & & & & & \biggl.\biggr\} \text{\ Lemma~\ref{lemma:paychangeslow2}}\\[-8pt]
\Gamma \ \text{(big)}, & \Gamma \ \text{(big)}, & {\color{blue}0}, & \_, & \ldots, & \_  & 
{\color{blue} p_1}\\[-6pt]
& & & & & & & \biggl.\biggr\} \text{\ Lemma~\ref{lemma:paychangeslow2}}\\[-8pt]
\Gamma \ \text{(big)}, & \Gamma \ \text{(big)}, & {0}, & {\color{blue}0}, & \ldots, & \_  &  {\color{blue} p_2}\\[-3pt]
& & & & \vdots & & & \biggl.\biggr\} \text{\ Lemma~\ref{lemma:paychangeslow2}}\\[-4pt]
\Gamma \ \text{(big)}, & \Gamma \ \text{(big)}, & {0}, & {0}, & \ldots, & {\color{blue} 0}  &  {\color{blue} p'}\\[-9pt]
& & & & & & & \biggl.\biggr\} \text{\ Lemma~\ref{lemma:paychangeslow}}\\[-9pt]
\Gamma \ \text{(big)}, & \Gamma \ \text{(big)}, & {\color{blue} p'_1 + \epsilon}, & {0}, & \ldots, & {0}  &  {\color{blue} p'_1}\\[-5pt]
& & & & & & & \biggl.\biggr\} \text{\ Lemma~\ref{lemma:confirmInvariant}}\\[-9pt]
\Gamma \ \text{(big)}, & \Gamma \ \text{(big)}, & {\color{blue} \Gamma \ \text{(big)}}, 
& {0}, & \ldots, & {0}  &  {p'_1}\\[-5pt]
& & & & & & & \biggl.\biggr\} \text{\ Lemma~\ref{lemma:paychangeslow}}\\[-8pt]
\Gamma \ \text{(big)}, & \Gamma \ \text{(big)}, & \Gamma \ \text{(big)}, & 
{\color{blue} p'_2 + \epsilon}, & \ldots, & {0}  &  {\color{blue} p'_2}\\[-5pt]
& & & & & & & \biggl.\biggr\} \text{\ Lemma~\ref{lemma:confirmInvariant}}\\[-8pt]
\Gamma \ \text{(big)}, & \Gamma \ \text{(big)}, & \Gamma \ \text{(big)}, & 
{\color{blue} \Gamma \ \text{(big)}}, & \ldots, & {0}  &  {p'_2}\\[-3pt]
& & & & \vdots & & & \biggl.\biggr\} \text{\ Lemmas~\ref{lemma:confirmInvariant}, \ref{lemma:paychangeslow}}\\[-9pt]
\Gamma \ \text{(big)}, & \Gamma \ \text{(big)}, & { \Gamma \ \text{(big)}}, & { \Gamma \ \text{(big)}}, & \ldots, & {\color{blue}  \Gamma \ \text{(big)}}  &  {\color{blue} p''}\\[-9pt]
\end{array}
\]
\caption{{\bf Proof roadmap for Theorem~\ref{theorem:twoweakSCPwithburn}.}
We construct a sequence of bid vectors, and show that if
the mechanism satisfies the desired properties, 
then, in the last configuration, every bid must 
be confirmed. Since there are more bids than the block size, we reach  
a contradiction. The notation ``\_'' denotes a bid whose value 
we do not care about (as long as $\Gamma$ is big enough w.r.t. all 
these values).}
\label{fig:proofroadmap}
\end{figure*}

\paragraph{Demonstrating the contradiction (Figure~\ref{fig:proofroadmap}).}
With the above key observations, we can finally 
demonstrate a contradiction, assuming that there indeed
exists a deterministic, $2$-user-friendly mechanism that is weak UIC and 
$2$-weak-SCP.
\begin{enumerate}[leftmargin=5mm]
\item 
First, we show that there exists a bid vector $\bfb = (b_1, b_2, \ldots, b_m)$ 
such that there are two (or more) users
confirmed, and both users bid strictly higher than the payment.
Note that $2$-user-friendliness guarantees the existence 
of a vector $\bfb$
such that two users are confirmed, but does not directly guarantee
that both of them bid strictly above the payment --- it actually requires a bit of work to show this
(which we defer to the subsequent formal presentation).
Henceforth, without loss of generality, we may assume that $b_1$ and $b_2$
are the two confirmed bids, and let $p$ be the payment.
We know that $b_1 > p$, and $b_2 > p$.

\item 
Next, using Lemma~\ref{lemma:confirmInvariant}, we can increase both $b_1$ and $b_2$  
to some sufficiently large number $\Gamma$, without affecting the payment $p$
or the miner revenue, and the resulting 
bid vector is $(\Gamma, \Gamma, b_3, \ldots, b_m)$.

\item 
Next, we can lower $b_3, \ldots, b_m$ all to $0$. 
Due to Lemma~\ref{lemma:paychangeslow2} and the sufficiently large choice of $\Gamma$,
the increase in the payment is relatively small in comparison with $\Gamma$.
This means that 
at the end, the first two users' bid amount $\Gamma$ is still much greater than the universal payment,
despite the possible increase in the universal payment.
Therefore,  the first two users must be still confirmed at the end 
(formally showing this requires a bit extra work).
At this moment, we have a bid vector $(\Gamma, \Gamma, 0, 0, \ldots, 0)$,
where the first two users are confirmed, and there is still a sufficiently
large gap between their bid $\Gamma$ and the universal payment $p'$.

\item 
Next, one by one, we shall increase the bids of users $3$ through $m$.
For each user $j \in \{3, 4, \ldots, m\}$, as we increase their bid at some rate
$r$,  the universal payment increases 
at rate at most $r/2$ due to Lemma~\ref{lemma:paychangeslow}.
At some point, $j$'s bid will surpass the universal payment, and at this point,
due to  
the third \elaine{hardcoded ref} natural property mentioned earlier, user
$j$ must become confirmed.
Note that during this entire process, users $1$ and $2$ remain confirmed since
their bids $\Gamma$ is sufficiently large.

\item 
Repeating the above process, 
we will eventually obtain a bid vector such that all $m$ users  
are confirmed. Now, as long as $m$ is strictly greater than the block size $B$, 
we reach a contradiction --- note that this is the only place where 
we use the finite block size assumption in the entire proof.
It turns out that we can safely assume 
$m > B$, since if the initial vector $\bfb$ has fewer than $B$ users, we can always 
append $0$ bids to $\bfb$ ``for free'' (and showing this requires a little extra
work which we defer to the subsequent formal exposition).
\end{enumerate}




%% file: burn-2nd-price.tex
\label{sec:burn2ndprice}


\paragraph{Burning second-price mechanism: special case $c = 1, \gamma = 1$.}
To aid understanding, we first give the special case of the 
mechanism for $c = 1$ and $\gamma = 1$, and we then generalize
it to arbitrary choices of $c$ and $\gamma$. 
\begin{mdframed}
\begin{center}
{\bf The burning second-price auction: special case when $c = 1$, $\gamma = 1$} 
\end{center}
\paragraph{Parameters:} the block size $B$,  
and $0 < k' \leq k < B$
such that $k + k' = B$, where $k$ denotes the number of confirmed
transactions per block, and $k'$ denotes the number of unconfirmed
transactions in a block that are used to set the price and miner revenue.

\paragraph{Mechanism:}
\begin{itemize}[leftmargin=5mm,itemsep=1pt]
\item 
Choose the $B$ 
highest bids to include in the block.
The highest $k$
bids are considered confirmed, and they each pay the $(k+1)$-th price.
Unconfirmed transactions, included or not,  pay nothing.
\item 
The miner is paid the sum of the $(k+1)$-th 
to the $B$-th prices (which cannot exceed the total payment by construction).
All remaining payment collected from the confirmed
transactions is burnt.
\item 
If the block is not fully filled, any remaining empty slot is treated
as a bid of $0$.
\end{itemize}
\end{mdframed}
In the above mechanism, 
the top $k$ users pay the $(k+1)$-th price, the miner gets the 
sum of the $(k+1)$-th to $B$-th prices, and the rest of the payment 
is burnt.
It is easy to see that the mechanism satisfies UIC  since it behaves
exactly like a second-price auction from the user's perspective.
For MIC and $1$-SCP, observe that if a miner or miner-user
coalition overbids at $b$ which is strictly greater than 
the true value $v$, then it may be able to gain $b-v$
extra in miner revenue.
However, when the offending transaction eventually becomes 
confirmed sometime in the future, it may end up paying $b$, thus
incurring a cost of $b-v$, and offsetting the gain in revenue.
The above argument also holds for injecting a fake bid 
which can be viewed as a special case where the true value $v = 0$.
Of course, the above intuition is not a formal proof, 
we provide the formal proof in the subsequent technical sections.


\paragraph{Burning second-price auction: general case.}
We now generalize the mechanism to arbitrary choices 
of $\gamma \in (0, 1]$ 
and $c \geq 1$. The resulting mechanism achieves UIC, MIC, and $c$-SCP
under $\gamma$-strict utility.

\begin{mdframed}
\begin{center}
{\bf The burning second-price auction: general case}
\end{center}
\paragraph{Parameters:}
\begin{itemize}[leftmargin=5mm,itemsep=0pt,topsep=2pt]
\item the block size $B$,
\item the maximum coalition size $c \in \N$,
\item the discount factor $\gamma \in [0,1]$,
\item $k,k' \in \N$ such that $k + k' = B$ \elaine{note: changed to = here} 
and $1\leq k' \leq \lfloor\frac{\gamma k}{c} \rfloor$\footnote{When $\lfloor\frac{\gamma k}{c} \rfloor = 0$, the mechanism reduces to a trivial case where no transaction is confirmed and the miner is paid nothing. Thus, we only specify the case where $\lfloor\frac{\gamma k}{c} \rfloor \geq 1$.}, 
where $k$ denotes the number of included bids that might be confirmed with some probability,
and $k'$ is the number of included 
bids that are not confirmed, but are used to set the price.
(The probability is defined in the confirmation rule below.)

\end{itemize}

\paragraph{Mechanism:}
\begin{itemize}[leftmargin=5mm,itemsep=1pt,topsep=2pt]
\item 
{\it Inclusion rule.}
Choose the $B$ highest bids to include in the block, 
breaking ties arbitrarily.
Let $(b_1,\ldots, b_B)$ denote the included bids where $b_1 \geq \cdots \geq b_B$.
If the block is not fully filled, any remaining empty slot is treated as a bid of $0$.
\item 
{\it Confirmation rule.}
Select a random subset $S \subseteq \{b_1, \ldots, b_k\}$ 
of size exactly $\lfloor\frac{\gamma k}{c}\rfloor$ using (trusted) on-chain randomness.
The set $S$ is confirmed
and all other bids $\{b_1, \ldots, b_B\} \setminus S$ are unconfirmed.
\item 
{\it Payment rule.}
Any confirmed bid pays $b_{k+1}$.
All unconfirmed bids pay nothing. 
\item 
{\it Miner revenue rule.}
The miner is paid $\gamma \cdot (b_{k+1} + \cdots +  b_{k + k'})$.
Burn any remaining payment collected from the confirmed bids.
\end{itemize}
\end{mdframed}

\paragraph{A note about randomized TFM and implementation of the random coins.}
In general, a randomized mechanism may employ random coins  
in the the miner-implemented inclusion rule,
as well as the blockchain-implemented confirmation rule,
payment and miner revenue rules.
In our burning second-price auction specifically,
the inclusion rule executed by the miner is
deterministic, and only the   
confirmation rule that is executed by the blockchain is randomized.
To implement such a mechanism in practice, we will need trusted on-chain randomness.
How to generate unbiased and unpredictable random coins in distributed environment has
been extensively studied~\cite{Cachin00randomoracles,randpiper,spurt}. 
Since such ``trusted'' random coins could be expensive to generate in a decentralized environment, 
we would ideally like to avoid them. 
Unfortunately, 
we will show 
later that randomness is actually necessary
to get weak incentive compatibility when $c \geq 2$. 
\ignore{
For example, in certain consensus protocols, the consensus nodes
can jointly run a coin toss protocol  
for each block mined.
The coin toss outcome is guaranteed to be unbiased and unpredicable as long as
a majority or super-majority of the nodes behave honestly.
}

\paragraph{Some interesting observations.}
We can make a few intereseting 
observations about this mechanism:
\begin{enumerate}[leftmargin=5mm,topsep=2pt]
\item 
First, the larger the coalition resistance parameter $c$, 
the smaller the number of confirmed bids 
$\floor{\frac{\gamma k}{c}}$.
Similarly, when $\gamma$ is larger, i.e., when we are charging harsher costs for  
cheating, the mechanism can confirm more bids. 
In other words, both $c$ and 
$\gamma$ can be viewed as knobs that allow us to smoothly tradeoff 
the {\it strength of incentive compatibility}
and the {\it efficiency} of the mechanism.
We stress that such a tradeoff
is inevitable due to our earlier impossibility result
for strong incentive compatibility (see Corollary~\ref{cor:finiteblocksize}).
Our burning second-price auction
gives a {\it mathematically quantifiable tradeoff}
between the resilience towards strategic behavior 
and the efficiency of the mechanism.

As a special case, 
when $\gamma = 0$, i.e., when there is no cost for overbid/fake unconfirmed bids,  
the number 
of confirmed bids $\floor{\frac{\gamma k}{c}} = 0$ --- in other words, the mechanism becomes
degenerate. This is consistent with our earlier impossibility result
for strong incentive compatibility.

\item 
Second, when $\gamma = 1$ and $c = 1$,  
the mechanism acutally becomes {\it deterministic}, since 
the number of confirmed bids $\floor{\frac{\gamma k}{c}} = k$. In other words,
the top $k$ included bids are surely confirmed.  We 
give a full description of the mechanism 
for this particularly interesting special case below.

On the other hand, if $c > 1$, the mechanism is randomized even for $\gamma = 1$. 
This is no co-incidence, since later, 
we will prove that randomness
is actually necessary for $c > 1$ for any ``interesting'' mechanism.
\elaine{can we prove that randomness is needed for c = 1 and $\gamma < 1$}
\end{enumerate}

\ignore{TODO: comment on how to realize on-chain randomness, randomness expensive.
write the deterministic special case gamma = 1 and c =1}

\begin{theorem}[Burning second-price auction, restatement of Theorem~\ref{thm:intro-upper}]
For any $c \geq 1$ and $\gamma \in (0, 1]$, 
the burning second-price 
auction satisfies UIC, MIC, and $c$-SCP under $\gamma$-strict utility.
\label{thm:burn2ndprice}
\end{theorem}

The proof of Theorem~\ref{thm:burn2ndprice}
is provided in Section~\ref{sec:proof-burn2ndprice}.
We provide some informal intuition about the proof below.
First, if $\gamma = c = 1$, the mechanism is deterministic.
The top $k$ users pay the $(k+1)$-th price, and the miner gets the sum of the $(k+1)$-th to $B$-th prices. 
For this case, the proof can be accomplished through a careful case-by-case analysis.
Next, when $\gamma < 1$ and $c \geq 1$, the mechanism becomes randomized, and the miner's revenue is the sum of the $(k+1)$-th to $B$-th prices, but discounted by $\gamma$. 
The discount factor in miner revenue is necessary since otherwise, the miner (or miner-user coalition) will be incentivized to raise the $(k+1)$-th price by $\Delta$ (through a fake or overbid transaction) for some $\Delta > 0$. 
The cost to the miner is $\gamma \cdot \Delta$, but the miner would get paid $\Delta$ more if there were no discount factor in miner revenue. 
With this newly introduced $\gamma$ factor, we need to randomly sample $\gamma/c$ fraction of the top $k$ transactions to confirm. 
Otherwise, a miner colluding with $c$ users will be incentivized to lower the $(k+1)$-th price by $\Delta$. 
In this way, each of the $c$ colluding users saves $\Delta$ in payment, but the miner loses only $\gamma \cdot \Delta$ in revenue. 
The random sampling ensures that in expectation, only $\gamma/c$ fraction of the colluding users' transactions are confirmed, thus avoiding this problem.
Formalizing the above intuition into a proof requires a careful case by case analysis, and the complete proof is provided in Section~\ref{sec:proof-burn2ndprice}.

Although the tradeoff between resilience and efficiency 
as observed in the burning second-price auction is inevitable,
currently we do not understand whether our mechanism
achieves the optimal tradeoff curve, i.e., whether it  
achieves optimality for every choice of $\gamma$ 
in terms of utilization of on-chain space  
and money burnt. We leave this as an exciting direction
for future work (see also Section~\ref{sec:conclusion}
for numerous other open questions 
in this exciting and little explored space).

%% file: related.tex
\subsection{Additional Related Work}
\label{sec:related}

\paragraph{Transaction fee mechanism.}
We now review some additional related work besides
the most closely related work EIP-1559~\cite{eip1559}
and that of Roughgarden~\cite{roughgardeneip1559,roughgardeneip1559-ec}.
Specifically, we will review the transaction fee mechanisms that have been
proposed, and explain  
which of the three properties they each fail to satisfy.

Lavi, Sattath, and Zohar~\cite{zoharfeemech}
pointed out that today's ``pay your bid'' auction 
has resulted in complex strategic bidding behavior. 
In particular, when there is no congestion, users would 
bid almost $0$, resulting in very little
transaction fee revenue for the miners.
To alleviate the problem, 
\cite{zoharfeemech} suggests
two alternative mechanisms, 
Monopolistic Price, and Random Sampling Optimal Price (RSOP), initially
proposed in 
\cite{competitiveauction}.
As \cite{zoharfeemech} acknowledged, Monopolistic Price
is not strictly user incentive compatible (by the classical DSIC notion),  
and is not even 1-side-contract resilient. 
For RSOP, \cite{zoharfeemech} demonatrated an attack showing that it is not MIC. 
In fact, a slightly modified attack can also 
show that RSOP is not side contract resilient.
Yao~\cite{yaofeemech} proved that 
although Monopolistic Price is not strictly UIC, it is nearly  
UIC assuming any i.i.d. distribution  
of the users' true values, 
and as the number of users goes to infinity.
Further, Yao also proved a conjecture in \cite{zoharfeemech}
regarding the relative revenue of the two mechanisms.

Basu, Easley, 
O'Hara, and Sirer~\cite{functional-fee-market}
suggested mechanism that involves paying the transfaction fees forward 
to some number of future blocks.
Roughgarden~\cite{roughgardeneip1559}
simplified and analyzed their scheme, and argued that it does not satisfy  
any of the three properties, although it is approximately UIC when the number
of users goes to infinity.

Ferreira, Moroz, 
Parkes,  and Stern~\cite{dynamicpostedprice}
suggest a modification to EIP-1559: whereas EIP-1559 approximates a first price auction
in the congested regime and approximates a posted price auction
in the infinite block size regime, 
\cite{dynamicpostedprice}
suggest to adopt a posted price mechanism  
no matter which regime one is in, by modifying the 
reserve price over time.
 \cite{dynamicpostedprice}'s approach does not adopt a burn rule, and fails to satisfy
even $1$-side-contract-proofness.

\vspace{3pt}
\noindent \underline{\it More detailed explanations:}
in Appendix~\ref{sec:detailedrelatedwork}, we explain
each of these known TFMs~\cite{zoharfeemech,functional-fee-market,dynamicpostedprice}
in more detail, and explain why they fail to satisfy
one or more of the desired incentive compatibility properties.

\paragraph{Auctioneer misbehavior in auctions.}
Akbarpour and Li~\cite{credibleauction}
proposed a notion called {\it credible} auctions.
In particular, in their model,
the users each communicate with the auctioneer
over a {\it private channel}. 
Credibility requires that the auctioneer does not have incentives
to implement any ``safe'' deviations, i.e., deviations
where the miner can plausibly explain away to every user without
being implicated (e.g., by lying about the message other users have sent).
Akbarpour and Li~\cite{credibleauction} showed  
a two-out-of-three type impossibility result for {\it optimal} credible auctions,
where optimality implies that the auctioneer's revenue is maximized.
While their definition is somewhat similar in spirit 
to miner incentive compatibility (MIC), 
we stress that 
{\bf their lower bound for credibility does not imply a corresponding
lower bound for MIC in our model, since 
a TFM that is MIC may not be credible in their model}, as shown
in the following counter-example. 
\ignore{
First, {\it a credible auction in their model
does not immediately imply a TFM that is MIC}. Their model
wants to remove incentives for the auctioneer to 
lie about users' messages in justification of its own misbehavior,
whereas in our model, since all messages are posted publicly, 
such lying is inherently prevented.
Instead, the kinds of miner deviations we want to discourage such as injecting 
fake transactions are not easily captured by their model
where the set of users are fixed  and known a-priori (and thus such deviations
are not considered ``safe'' in their model).
Second, {\bf a TFM that is MIC does not necessarily 
give rise to a credible auction in their model}.
}
Consider our own deterministic burning second-price mechanism
(where $k' = 1$) can be viewed as follows:
the top bidder pays the $(k+1)$-th price to the 
auctioneer, every other top $k$ bidder 
pays the $(k+1)$-th price but the payment is burnt.
This mechanism is MIC but not credible in the model of 
Akbarpour and Li~\cite{credibleauction}, 
since the auctioneer can lie about the $(k+1)$-th price
to the top bidder and thus earn more from the top bidder.
Moreover, we also stress that the style of our impossibility
is of a stronger nature than that of  Akbarpour and Li~\cite{credibleauction}:  
we rule out {\it any} 
TFM whatsoever satisfying UIC and 1-SCP even when allowing burning, and 
{\it without any regards to optimality}.
The elegant work of Ferreira and Weinberg~\cite{commit-credible-auction} showed  
that using cryptographic commitments can help overcome some 
of the lower bound results 
shown by Akbarpour and Li~\cite{credibleauction}. 

We want to discourage miners from injecting fake transactions. In the economics
literature, auctioneer injecting transactions is also sometimes referred 
to as shill bidding~\cite{shillbid02,shillbid00,shillbid01}. 
Earlier works have found that shill bidding can help increase
the auctioneer's profit~\cite{shillbid02,shillbid00,shillbid01}.

\paragraph{User collusion in auctions.}
\ignore{
Several works in the mechanism design literature 
are related to our work.
Akbarpour and Li~\cite{credibleauction}
proposed a notion called credible auctions, 
i.e., auctions where the auctioneer does not have incentives
to implement any ``safe'' deviations. 
In particular, safe deviations are for which there exists
a plausible explanation. 
While their definition is somewhat similar in spirit 
to miner incentive compatibility (MIC), 
their modeling is incomptible with TFM. In TFM, 
all transactions included
in the block must be visible to the public, whereas  
Akbarpour and Li~\cite{credibleauction} consider auctions where 
a bidder may not be able to see others' bids --- and the cheating
auctioneer could exploit this to explain his cheating behavior away.
The elegant work of Ferreira and Weinberg~\cite{commit-credible-auction} showed  
that using cryptographic commitments can help overcome some 
of the lower bound results 
shown by Akbarpour and Li~\cite{credibleauction}. 
}
A line of works also consider collusion 
among bidders in auctions~\cite{coalitionic,goldberghartline,chen-collusive,optcollusionproof,econcollusion,decklbaummicali}.
Traditional auctions like the Vickrey auction do not satisfy incentive compatibility
if bidders can collude through binding side contracts.
Therefore, this line of work 
explores under what modeling assumptions or incentive compatibility  
notions is it possible to resist bidder collusion. 
The transaction fee mechanism 
(TFM) line of work has not focused 
on user-user collusion --- as mentioned earlier, user-user
rendezvous is difficult to facilitate since users are ephemeral in 
decentralized blockchain settings.

\elaine{anything else to cite?}



%% file: defn-new.tex
\section{Definitions}\label{section:definitions}

In this section, we define a transaction fee mechanism (TFM) 
formally, as well as incentive compatibility notions.
%
Our modeling choice can be viewed as a generalization of 
that of Roughgarden's~\cite{roughgardeneip1559,roughgardeneip1559-ec}.
Specifically, Roughgarden's model only cares about which transactions
are eventually confirmed, but does not care about which ones
are included in the block. 
By contrast, our modeling explicitly separates the ``inclusion rule''
from the ``confirmation rule''.
Both our lower bound and upper bound 
will demonstrate that explicitly separating the ``inclusion rule''
and the ``confirmation rule'' is important for understanding the  
feasibilities and infeasibilities 
of transaction fee mechanism design.
See also Remarks~\ref{rmk:defn-tfm} and \ref{rmk:defn-ic} for 
additional philosophical discussions 
about the modeling.


\subsection{Transaction Fee Mechanism}

We consider a single auction instance corresponding to 
the action of mining the next block.
Suppose that there is a mempool 
containing the list of pending transactions submitted by users.
We may assume that each transaction is submitted
by a distinct user.
We consider a single parameter environment, i.e.,
each user $i$ has a true value $v_i \in \R$ 
for getting its transaction confirmed in the next block; moreover, 
its bid contains only a single value $b_i \in \R$ as well.
Henceforth, we use 
${\bf b} := (b_1, b_2, \ldots, b_m)$ 
to denote the vector of all bids; we also use the same notation
${\bf b}$ to denote the current mempool.
For convenience, we often use the terms {\it bid} and {\it transaction}
interchangeably, e.g., $b_i$ can be called a bid or a transaction.

As defined earlier in \Cref{sec:roadmap-tfm}, 
a Transaction Fee Mechanism (TFM) consists of 
the following possibly randomized algorithms:
\begin{itemize}[leftmargin=5mm,itemsep=1pt]
\item 
an {\it inclusion} rule
henceforth denoted $\bfI(\cdot)$, 
\item 
a {\it confirmation} rule henceforth denoted $\bfC(\cdot)$, 
\item 
a {\it payment} rule henceforth denoted $\bfP(\cdot)$, 
and 
\item 
a {\it miner revenue} rule henceforth denoted $\bfM(\cdot)$.
\end{itemize}
\ignore{
the following 
(possibly randomized) algorithms:
\begin{itemize}[leftmargin=5mm,itemsep=1pt]
\item 
{\bf Inclusion rule $\bfI(\cdot)$}:
given a bid vector $\bfb$, $\bfI(\bfb)$
outputs a subset 
of the vector $\bfb$, denoting the bids to be included in the block.

\item 
{\bf Confirmation rule $\bfC(\cdot)$}:
given a set of bids
$\bfb'$ included in the block, the confirmation rule
$\bfC(\bfb')$ outputs
which of these bids are confirmed.

\item 
{\bf Payment rule $\bfP(\cdot)$}:
given a set of bids
$\bfb'$ included in the block, the payment rule
$\bfP(\bfb')$ outputs
how much each confirmed bid pays. We assume that 
1) any bid that is not confirmed
pays a price of $0$; and 2) each transaction pays at most what it bids. 

\item 
{\bf Miner-revenue rule $\bfM(\cdot)$}:
given a set of bids
$\bfb'$ included in the block, the miner-revenue rule
$\bfM(\bfb')$ outputs
the total payment received by the miner for mining this block.
We assume that the 
miner revenue 
does not exceed the total payment of all confirmed bids.
\end{itemize}
}
The inclusion rule is implemented by the miner, possibly subject to certain
validity constraints
enforced by the blockchain (e.g., block size limit).
The other rules, including confirmation, payment, and miner-revenue rules
are enforced by the blockchain itself; and they use only on-chain information.
We assume that any unconfirmed transaction must pay $0$ and each transaction
pays no more than the bid amount.

There are a few important things to note about this definition:
\begin{enumerate}[leftmargin=5mm,itemsep=1pt]
\item  
{\it Included vs confirmed:}
In the most general form, 
not all transactions included in the block must be confirmed.
It could be that some transactions are included in the 
block to set the price, but they are not considered confirmed.
For example, consider a Vickrey auction where the $k$ highest
bids are included in the block, among which the $k-1$ highest are considered 
confirmed, paying the $k$-th price.
In this case, the $k$-th transaction
is included just to set the price.
\item 
{\it Encoding the burn rule.}
Not all the payment from the users 
will necessarily go to the miner of the block.
It was pointed out earlier, e.g., in Ethereum's 
EIP-1559~\cite{eip1559,roughgardeneip1559,roughgardeneip1559-ec}
that in a blockchain, 
part to all of the payment 
can be burnt.
In our definition, we require that {\it the miner revenue 
be upper bounded
by the total payment from all confirmed transactions}.
In case the miner's revenue is strictly less
than the total user payment, the difference is essentially ``burnt''.
\end{enumerate}

In some cases, the TFM may need to perform tie breaking\footnote{We assume
that the honest TFM is a randomized algorithm and thus it does not have
any non-deterministic behavior. For example, the mechanism
cannot say ``pick an arbitrary subset of bids among those
that bid at least $5$'' --- the mechanism should fully specify how to break ties
when choosing the subset (possibly using random coins). }. For example,
if there are more bids bidding the same price 
than the block can contain, only a subset of them will be included.
Our formulation implicitly implies that the TFM 
is {\it identity agnostic}, i.e., the TFM 
does not use the bidders' identities for tie-breaking.
In other words, if we swap two users' actions, their outcomes would be swapped too.
More formally, 
given a bid vector $\bfb := (b_1, \ldots, b_m)$ and two different users $i$ and $j$,
let $x_i, x_j \in \{0, 1\}$ denote whether  
each user is confirmed, and let $p_i, p_j$ denote their respective payments.
Now, imagine that  
we swap users $i$ and $j$'s roles as follows.
We make $i$ bid $b_j$ and make $j$ bid $b_i$ instead, and we swap $i$ and $j$'s positions
in the bid vector. 
In other words, we still have the same bid vector 
$\bfb$ as before. However, the $i$-th coordinate now contains the bid from user $j$
and the $j$-th coordinate now contains the bid from user $i$.
In this case, 
the outcomes for $i$ and $j$ would be swapped too, that is, user $i$'s outcome
becomes $x_j, p_j$ and user $j$'s outcome becomes $x_i, p_i$.
\elaine{check that this is what we need in all the proofs later for tie-breaking}

\ignore{
Throughout the paper, we assume that a TFM  
treats all users equally in the following sense:
let $\bfb$ be an arbitrary bid vector where $b_i$ and $b_j$
correspond to users $i$ and $j$'s bids, respectively. 
If we now consider the same bid vector $\bfb$ but $b_i$ encodes
user $j$'s bid, 
and $b_j$ encodes user $i$'s bid
}

\begin{remark}[On separating the inclusion and confirmation rules]
In comparison, Roughgarden~\cite{roughgardeneip1559,roughgardeneip1559-ec}
adopts a simpler notation that does not explicitly differentiate
between the inclusion rule
and the confirmation rule. 
Indeed, parts of our 
impossibility proofs do not care about  
this differentiation --- and in these cases,  
we use a simplified notation that coalesces the inclusion and confirmation
rules
(see Section~\ref{sec:imp-notation}).
However, our results show that it is important to explicitly
separate the inclusion rule and the confirmation rule in the modeling, to further
our understanding about TFMs.
For example, making the inclusion rule explicit is important
for proving the impossibility under {\it finite} block 
size (see Corollary~\ref{cor:finiteblocksize}).
Having this distinction is also useful 
in constructing our upper bounds. 
\label{rmk:defn-tfm}
\end{remark}

\subsection{Strategic Behavior and Utility}
\paragraph{Strategic player.}
We will consider three types of {\it strategic players}, 
1) an individual user; 2) the miner of the current block;
and 3) the miner colluding with a single user.
Henceforth, we will use the term {\it strategic player} to refer
to either a user, the miner, or the coalition of a miner and a single user.

As mentioned earlier, user-user 
rendezvous is much more difficult since users are ephemeral, and this is
likely why this 
line of works~\cite{zoharfeemech,functional-fee-market,roughgardeneip1559} 
focused on miner-user collusion (as opposed to user-user collusion).
Moreover, it is easier for the miner to form a side contract 
with a single user rather than more users. 


\paragraph{Strategy space.}
A strategic player may rely on strategic deviations to improve its utility.
We first define the strategy space in the most general form, capturing  
all possible deviations.  
Our impossibility proof will rely on 
a much more restricted strategy space (which makes the impossibility result 
stronger) --- we will explicitly point out
the strategy space needed by our impossibility 
in Section~\ref{sec:impossible}. 
On the other hand, 
our weakly incentive compatible upper bound in Section~\ref{sec:weakic}
defends against the broad strategy space defined below.

A strategic player can engage in the following
types of deviations or a combination thereof:
\begin{itemize}[leftmargin=5mm,itemsep=1pt]
\item 
{\it Bidding untruthfully}. 
A user or a user-miner coalition can bid untruthfully, possibly after examining some or 
all other users' bids.
\item 
{\it Injecting fake transactions}. 
A user, miner, or a user-miner coalition can inject fake transactions,
possibly after examining 
some or 
all other users' bids.
Fake transactions offer no intrinsic value to anyone, and 
their true value is $0$. 
\Hao{Does user injecting transaction mean spliting transaction attack? That is, a user can announce multiple transactions. If so, the current proof of 1-weak-SCP may not include this notion.}
\item 
{\it Strategically choosing which transactions to include in the block.} 
A strategic miner or a miner-user coalition   
may not implement the inclusion rule faithfully. 
It may choose an arbitrary subset of transactions from the mempool
to include in the block, as long as it satisfies
any block validity rule enforced by the blockchain.
\end{itemize}

Just like the prior work of Roughgarden~\cite{roughgardeneip1559,roughgardeneip1559-ec},
we do not consider the strategic behavior of 
splitting a single bid $b$ into multiple bids $b_1, \ldots, b_k$
whose sum is equal to $b$ --- in this case,
a more complicated utility definition is needed
when only a proper subset of these bids are confirmed.
To capture such strategies may require additional tools, e.g.,  
modeling TFM as a combinatorial auction~\cite{agt}, which we leave as future work. Such
strategies may also be remotely related to an elegant line of work on false-name bids 
in the mechanism design 
literature~\cite{falsenamebid00,falsenamebid01,falsenamebid02}. 
Again, our work, just like Roughgarden's~\cite{roughgardeneip1559,roughgardeneip1559-ec}, 
is among the very first mathematical explorations
of transaction fee mechanism design, and part of our contribution
is to expose the lack of knowledge and the
 abundance of open questions in this space.

\paragraph{Utility.}
The utility of 
the miner or a miner-user coalition 
is computed as the following, where $S$
denotes the set of all real and fake transactions\footnote{Recall the true value of a fake transaction is defined to be zero.}
submitted by 
the miner or the miner-user coalition:
\[
\text{miner revenue} +  
\sum_{\forall b \in S \text{\  and $b$ confirmed}}
(\text{true value of $b$} - \text{payment of $b$} ) 
\]

The utility of a sole user is computed
as the following, where $S$
denotes the set of all real and fake transactions  
submitted by the user:
\[
\sum_{\forall b\in S \text{\ and $b$ confirmed}}
(\text{true value of $b$} - \text{payment of $b$} ) 
\]

\ignore{
\begin{remark}
It is also possible for a user to split a transaction 
that wants to pay $\chi = \chi_1 + \chi_2$ into  
two transactions each paying $\chi_1$ and $\chi_2$  
to the same recipient.
\end{remark}
}

\subsection{Incentive Compatibility}
\label{sec:defn-ic}
We would like to have mechanisms that incentivize honest behavior, i.e., 
no deviation of a strategic player 
can increase its utility.
Depending on whether the strategic player
is a user, the miner, or the coalition of the miner
and a single user, 
we can define user incentive compatibility, miner incentive compatibility,
and side-contract-proofness, respectively.

\begin{definition}[User incentive compatibility]
 A TFM is said to be user incentive compatible (UIC),
iff the following holds: 
assuming that 
the miner implements the mechanism honestly, 
an individual user's (expected) utility is always maximized  
if it bids truthfully, no matter what the other users' bids are.
\label{defn:uic}
\end{definition}

\begin{definition}[Miner incentive compatibility]
 A TFM is said to be miner incentive compatible (MIC),
iff no matter what the users' bids are, 
the miner's (expected) utility is always maximized if 
it creates the block by honestly implementing the inclusion rule.
\label{defn:mic}
\end{definition}

\begin{definition}[$c$-side-contract-proofness]
For any $c \in \N$, a TFM is said to be 
$c$-side-contract-proof ($c$-SCP), 
iff 
for any coalition consisting of the miner and at least one and at most $c$ user(s), 
its (expected) utility is maximized when the colluding users 
bid truthfully and the miner plays by the book, no matter 
what the other users' bids are.
\label{defn:scp}
\end{definition}


\begin{remark}[Comparison with Roughgarden's incentive compatibility notions]
Our UIC and MIC notions are equivalent to 
Roughgarden's notions~\cite{roughgardeneip1559,roughgardeneip1559-ec}.
For the SCP notion, we modify Roughgarden's offchain-agreement-proofness 
notion and parametrize it with the coalition size $c$. 
Note that Roughgarden's notion wants that there is no side contract
that strictly benefits {\it every} coalition member
in comparison with the honest on-chain strategy --- this is equivalent
to saying that the coalition cannot deviate strategically to 
increase their {\it joint} utility. 
If they can increase their {\it joint} utility
there is always a way to split it off using a binding side contract 
such that every coalition member strictly benefits.

\ignore{
Our UIC and MIC notions are equivalent to 
Roughgarden's notions~\cite{roughgardeneip1559,roughgardeneip1559-ec}
when subject to the same strategy space. 
While we describe the strategy space in the most general form, 
Roughgarden omits some possible strategies, e.g., he omits 
the possibility that a user (as opposed to the miner) 
may inject a fake transaction, or the possibility that a miner can arbitrarily
choose which transactions to include in a block, which implies
that the miner can strategically ignore some transactions. 
It turns out that the specific 
mechanism Roughgarden studied indeed 
resists the strategic behaviors 
he omitted --- but this is not necessarily true for all mechanisms, and
there is value in making them explicit in a good definition.
\Hao{Maybe add a remark here saying that Aviv Zohar considers user splitting bid attack>}

As for the SCP notion, Roughgarden's 
definitions~\cite{roughgardeneip1559,roughgardeneip1559-ec}, as written,
appear somewhat incomplete --- see 
Appendix~\ref{sec:roughgarden-defn} for details. 
We therefore fix the SCP definitions, and we believe that our notion is likely 
what Roughgarden actually meant.
\elaine{todo: write this appendix}
}
\label{rmk:defn-ic}
\end{remark}

%% file: newImpossible.tex
\section{Impossibility Results}
\label{sec:impossible}

\subsection{Simplified Notation and Restricted Strategy 
Space for our Impossibility}
\label{sec:imp-notation}

To rule out the existence of a UIC and 1-SCP mechanism  
under finite block size, our proof takes two main steps. 
First, we shall prove
that any TFM that satisfies UIC and 1-SCP simultaneously 
must always have $0$ miner-revenue 
(Theorem~\ref{theorem:deterministic} and \ref{theorem:randomized}), no matter whether the block
size is infinite or finite.
These theorems hold 
even when the strategic player is 
confined to a very restricted
strategy space: 
assuming that the miner always implements the mechanism faithfully;
however, either an individual user or a user colluding with the miner may
bid untruthfully.
In the second part of the proof, 
we additionally throw in the finite block size restriction which 
leads to the stated impossibility result (Corollary~\ref{cor:finiteblocksize}).

\paragraph{Simplified notations for deterministic mechanisms.}
We can simplify the notation in the first part of our proof, since 
this part makes use of a very restricted strategy space as mentioned above.
Instead of using the full tuple $(\bfI, \bfC, \bfP, \bfM)$
to denote the TFM, 
we will use the following simplified notation:

\begin{enumerate}[leftmargin=5mm]
\item  {\bf Allocation rule} ${\bf x}$:
given a bid vector ${\bf b} := (b_1, \ldots, b_m) \in \R^m$, the allocation rule 
${\bf x}({\bf b})$
outputs a vector $(x_1, x_2, \ldots, x_m) \in \{0, 1\}^m$, indicating
whether each transaction (i.e., bid)  
in ${\bf b}$ is {\it confirmed} in the next block.

\item  {\bf Payment rule} ${\bf p}$:
given a bid vector ${\bf b} := (b_1, \ldots, b_m) \in \R^m$, the payment rule
${\bf p}({\bf b})$
outputs a vector $(p_1, p_2, \ldots, p_m) \in \R^m$, indicating
the price paid by each transaction in ${\bf b}$.
It is guaranteed that 
$p_i \leq b_i$ for $i \in [m]$, i.e., a user never pays more than its bid.

\item {\bf Miner-revenue} rule ${\mu}$:
given a bid vector ${\bf b} := (b_1, \ldots, b_m) \in \R^m$, the miner-revenue 
rule 
$\mu(\bfb)$
outputs a single value in $\R$ denoting the 
amount paid to the miner.
\end{enumerate}

More specifically, one can 
view:
\begin{itemize}[leftmargin=5mm,itemsep=1pt]
\item $\bfx$ as the composition
of the inclusion rule $\bfI$ and the  
blockchain-enforced confirmation rule $\bfC$; 
\item $\bfp$
as the composition 
of the inclusion rule $\bfI$ and the blockchain-enforced payment rule $\bfP$;  and 
\item 
$\mu$ as the 
composition of the inclusion rule $\bfI$ and the blockchain-enforced 
miner-revenue rule $\bfM$.
\end{itemize}

\paragraph{Additional notations.}
For convenience, we often use the notation $x_i({\bf b})$ 
and $p_i({\bf b})$ to denote whether the $i$-th transaction in ${\bf b}$ 
is confirmed in the next mined block, 
and what price it actually pays.
We assume that if $x_i({\bf b}) = 0$, 
then, $p_i({\bf b}) = 0$ --- in other words, 
if the $i$-th transaction is not confirmed in the next block, then 
the $i$-th user pays nothing.
Let $\bfb = (b_1, b_2, \ldots, b_m)$ be
a bid vector. 
We often use the notation $\bfb_{-i} = (b_1, b_2, \ldots, b_{i-1}, b_{i+1}, \ldots, b_m)$ 
to denote everyone except user $i$'s bids; 
and the notation $(\bfb_{-i}, b_i)$
and $\bfb$ are used interchangeably.

\paragraph{Notations for randomized mechanisms.}
We use the same notations $({\bf x}, {\bf p}, \mu)$ 
to denote a randomized mechanism but their meaning is modified
as follows.
The allocation rule now outputs the probability that each bid is confirmed,
that is, $x_i(\bfb) \in [0,1]$ is the probability 
that user $i$'s bid is confirmed given the included bids are $\bfb$.
Also, we view $p_i(\bfb)$ as the expected payment of user $i$ 
and $\mu(\bfb)$ as the expected miner-revenue.

\elaine{TODO: move this to defn section in the end.}

\ignore{
\begin{itemize}
\item  {\bf Burning rule} ${\bf q}$:
    given a bid vector ${\bf b} := (b_1, \ldots, b_m) \in \R^m$, the burning rule
    ${\bf q}({\bf b})$
    outputs a vector $(q_1, q_2, \ldots, q_m) \in \R^m$, indicating
    the amount of currency being destroyed.
\end{itemize}
}

We say that a TFM
enjoys {\it non-trivial miner revenue}
iff $\mu(\cdot)$ is not the constant $0$ function, i.e.,
the miner sometimes can receive positive revenue.

\subsection{Preliminary: Myerson's Lemma}
\label{sec:myerson}

If a single-parameter
TFM satisfies UIC (even when the user's strategy space
is restricted only to untruthful bidding), the mechanism's allocation rule ${\bf x}$
and payment rule ${\bf p}$ must satisfy the
famous Myerson's Lemma~\cite{myerson}.
\ignore{
Note that our UIC notion 
adopts a richer strategy space than the classical
notion of Dominant Strategy Incentive Compatible (DSIC) --- we consider
untruthful bidding and 
injecting fake bids as possible strategic deviations, 
whereas the traditional DSIC notion considers
only untruthful bidding. The enriched strategy space makes the mechanism
design harder, and therefore Myerson's Lemma should still apply just like
for traditional, single-parameter, DSIC auctions.
}
Specifically, we only need a special case of Myerson's Lemma:
the mechanism can be randomized, and each user's bid is either confirmed or unconfirmed.
In this case, the allocation rule $x_i$ returns a real number in $[0,1]$, which is the probability that user $i$'s bid is confirmed.
Additionally, $p_i$ is the expected payment of user $i$.
\ignore{
Consider a single-parameter auction defined
by the pair $({\bf x}, {\bf p})$.
}
Myerson's Lemma implies the following:

\begin{lemma}[Myerson's Lemma]
Let $({\bf x}, {\bf p}, \mu)$ be a single-parameter TFM that is UIC.
Then, it must be that  
\begin{enumerate}[leftmargin=5mm]
\item 
The allocation rule 
${\bf x}$ is {\it monotone},  
where monotone is defined as follows.
Consider ${\bf b} := (b_1, \ldots, b_m)$, and let ${\bf b}_{-i}$ be
the vector obtained when we remove $b_i$ from ${\bf b}$.
An allocation rule ${\bf x}$ is said to be monotone
iff for any ${\bf b} := (b_1, \ldots, b_m)$, 
and any $b'_i > b_i$, 
it must be that $x_i({\bf b}_{-i}, b'_i) \geq x_i({\bf b}_{-i}, b_i)$.
\item 
The payment rule ${\bf p}$ is defined as follows.
For any user $i$, bids $\bfb_{-i}$ from other users, and bid $b_i$ from user $i$,
it must be
\begin{equation}\label{eq:payment}
	p_i(\bfb_{-i}, b_i) = b_i \cdot x_i(\bfb_{-i}, b_i) - \int_0^{b_i} x_i(\bfb_{-i}, t) dt.
\end{equation}
\end{enumerate}
\label{lem:myerson}
\end{lemma}

\ignore{
Fix an allocation rule ${\bf x}$, if there is a corresponding payment rule
${\bf p}$ such that the TFM $({\bf x}, {\bf p}, \_)$
is 
UIC, 
then we say that ${\bf x}$ is {\it implementable} 
(note that 
we do not care about 


\begin{lemma}[Myerson's Lemma]
Myerson's lemma characterizes the space of 
DSIC auctions.
Let $({\bf x}, {\bf p})$ a single-parameter 
auction where ${\bf x}$ gives binary outcomes.
Then, 
\begin{enumerate}[leftmargin=5mm]
\item 
The allocation rule 
${\bf x}$ is implementable iff ${\bf x}$ is {\it monotone},  
where monotone is defined as follows.
Consider ${\bf b} := (b_1, \ldots, b_m)$, and let ${\bf b}_{-i}$ be
the vector obtained when we remove $b_i$ from ${\bf b}$.
An allocation rule ${\bf x}$ is said to be monotone
iff for any ${\bf b} := (b_1, \ldots, b_m)$, 
and any $b'_i > b_i$, 
it must be that $x_i({\bf b}_{-i}, b'_i) \geq x_i({\bf b}_{-i}, b_i)$.
\item 
If ${\bf x}$ is implementable, then  
there is a unique payment rule ${\bf p}$ that makes the auction DSIC, and moreover
this ${\bf p}$ is defined as follows.
For any user $i$, bids $\bfb_{-i}$ from other users, and bid $b_i$ from user $i$,
it must be
\begin{equation}\label{eq:payment}
	p_i(\bfb_{-i}, b_i) = b_i \cdot x_i(\bfb_{-i}, b_i) - \int_0^{b_i} x_i(\bfb_{-i}, t) dt.
\end{equation}
\end{enumerate}
\label{lem:myerson}
\end{lemma}
}

\paragraph{Deterministic special case.}
When the mechanism is deterministic, the allocation rule $x_i$ returns either $0$ or $1$. 
In this case, the unique payment rule can be simplified as 
\[
p_i(\bfb_{-i}, b_i) = 
	\left\{\begin{matrix}
	\min \{z \in [0, b_i]: x_i(\bfb_{-i}, z) = 1\}& \text{ if $x_i(\bfb_{-i}, b_i) = 1$,} \\ 
	0& \text{ if $x_i(\bfb_{-i}, b_i) = 0$.}
	\end{matrix}\right.
\]
Conceptually, user $i$ only needs to pay the minimal price which makes its bid confirmed.

To prove our impossibility for randomized mechanisms, 
we need to open up Myerson's Lemma and use
the following technical lemma that is used in the 
proof of Myerson's Lemma.
More specifically, the proof of Myerson's Lemma showed
that if a mechanism is UIC, then a user $i$'s payment
must satisfy the following inequality (also called a ``payment sandwich'')
where the allocation rule ${\bf x}$ is monotone:
\[
        r \cdot \left(x_i(\bfb_{-i}, r') - x_i(\bfb_{-i},r)\right)
        \leq p({\bfb_{-i}}, r') - p({\bfb_{-i}}, r)
        \leq r' \cdot \left(x_i(\bfb_{-i}, r') - x_i(\bfb_{-i},r)\right)
\]
Assume that the above payment sandwich holds for a non-decreasing
function $x_i(\bfb_{-i}, \cdot)$, and moreover, 
$p(\bfb_{-i}, 0) = 0$, then Myerson showed that the payment rule
is of a unique form as shown in Equation~(\ref{eq:payment}).
To prove this, Myerson essentially proved the following technical lemma.

\begin{lemma}[Technical lemma implied by the proof of Myerson's Lemma~\cite{myerson,myerson-lecture-hartline}]
Let $f(z)$ be a non-decreasing function. 
Suppose that $ z \cdot (f(z')-f(z)) \leq g(z') - g(z) \leq z' \cdot (f(z')-f(z))  $ 
for any $z' \geq z \geq 0$, and moreover, $g(0) = 0$.
Then, it must be that 
\[
g(z) = z \cdot f(z) - \int_0^{z} f(t) dt.
\]
\label{lem:sandwich}
\end{lemma}

\subsection{Deterministic Mechanisms: UIC + 1-SCP $\Longrightarrow$ Zero Miner Revenue}
\label{sec:trivial-miner-rev}

As a warmup, we first prove a lower bound for deterministic mechanisms.
Then, in Section~\ref{sec:randomized-lb}, we generalize the proof
to randomized mechanisms.
The following theorem states that no deterministic TFM 
with non-trivial miner revenue 
can achieve UIC and $1$-SCP simultaneously, no matter whether
the block size is finite or infinite.

\begin{theorem}[Deterministic TFM: UIC + 1-SCP $\Longrightarrow$ 0 miner revenue]
\label{theorem:deterministic}
There is no deterministic TFM with non-trivial miner revenue 
that achieves UIC and $1$-SCP at the same time.
Moreover, the theorem 
holds no matter whether the block size is finite or infinite.
\end{theorem}

The rest of this section will be dedicated to proving the theorem.
The following claim states that if an individual user changes
its bid in a way that does not affect whether it is confirmed, 
then the miner's revenue should not change.

\begin{claim}
Suppose that a TFM $({\bf x}, {\bf p}, \mu)$
satisfies UIC and 1-SCP.
Suppose that 
$x_i(\bfb_{-i}, b_i) = x_i(\bfb_{-i}, b'_i)$.
Then, it must be that 
$\mu(\bfb_{-i}, b_i) = \mu(\bfb_{-i}, b'_i)$.
\label{clm:inconsequentialbidchange}
\end{claim}
\begin{proof}
Since the TFM satisfies UIC, the tuple $(\bfx, \bfp)$ satisfies Myerson's Lemma.
We know that 
$x_i(\bfb_{-i}, b_i) = x_i(\bfb_{-i}, b'_i) = 0$ or 
$x_i(\bfb_{-i}, b_i) = x_i(\bfb_{-i}, b'_i) = 1$.
In the former case, 
$p_i(\bfb_{-i}, b_i) = p_i(\bfb_{-i}, b'_i) = 0$.
In the latter case, by Myerson's Lemma, 
no matter whether user $i$'s bid is $b_i$ or $b'_i$, 
its payment equals the minimal amount it bids that still allows the transaction
to be confirmed. 
Therefore, in either case, we have that 
$p_i(\bfb_{-i}, b_i) = p_i(\bfb_{-i}, b'_i)$.

Suppose that 
$\mu(\bfb_{-i}, b_i) \neq \mu(\bfb_{-i}, b'_i)$.
Without loss of generality, we may assume that 
$\mu(\bfb_{-i}, b'_i) > \mu(\bfb_{-i}, b_i)$.
In this case, imagine that 
all users' true values are represented by the vector $(\bfb_{-i}, b_i)$.
Now, consider the coalition of the miner and user $i$.
If user $i$ bids $b_i$ truthfully, 
the coalition's joint utility is $U := \mu(\bfb_{-i}, b_i) + b_i - p_i(\bfb_{-i}, b_i)$.
However, if user $i$' strategically bids $b'_i$
instead, 
the coalition's joint utility is $U' := \mu(\bfb_{-i}, b'_i) + b_i - p_i(\bfb_{-i}, b'_i)$.
Since $p_i(\bfb_{-i}, b_i) = p_i(\bfb_{-i}, b'_i)$, 
$U' - U = \mu(\bfb_{-i}, b'_i) - \mu(\bfb_{-i}, b_i) > 0$.
This shows that the coalition can gain if user $i$ bids untruthfully, thus
violating $1$-SCP. 
\end{proof}

\begin{lemma}
Let $({\bf x}, {\bf p}, \mu)$ be any TFM with non-trivial miner revenue.
Then, there exists a bid vector 
$\bfb = (b_1,\ldots,b_m)$ and a user $i$ such that $\mu(\bfb_{-i},0) < \mu(\bfb)$.
\label{lem:decrease}
\end{lemma}
\begin{proof}
Since the mechanism enjoys non-trivial miner revenue, 
there exists a bid vector $\bfb^{(0)} = (b_1,\ldots,b_m)$ such that $\mu(\bfb^{(0)}) > 0$.
Now, consider the following sequence of
bid vectors:  
for $i \in [m]$, let $\bfb^{(i)}$
be obtained by setting the first $i$ coordinates of $\bfb^{(0)}$ to $0$. 
Observe that $\bfb^{(m)} = {\bf 0}$.

Since a user can pay at most its bid, we have $\mu(\bfb) \leq |\bfp(\bfb)|_1
\leq |\bfb|_1$ for any bid vector $\bfb$.
Therefore, $\mu(\bfb^{(m)}) \leq |\bfb^{(m)}|_1 = 0$.
Since $\mu(\bfb^{(0)}) > 0$, there exists an $i\in [m-1]$ 
such that $0 = \mu(\bfb^{(i)}) < \mu(\bfb^{(i-1)})$.
\end{proof}

\begin{lemma}
    If there exists a bid vector $\bfb = (b_1,\ldots,b_m)$ and a user $i$ such that $\mu(\bfb_{-i}, 0) < \mu(\bfb)$,
    then the TFM $({\bf x}, {\bf p}, \mu)$ is either not UIC or 
not 1-SCP.
    \label{lem:nodecrease}
\end{lemma}
\begin{proof}
For the sake of reaching a contradiction, suppose that $({\bf x}, {\bf p}, \mu)$ is both UIC and 1-SCP.
By Myerson's Lemma, we have $x_i(\bfb_{-i}, 0) \leq x_i(\bfb)$.
Due to Claim~\ref{clm:inconsequentialbidchange}, 
it must be $x_i(\bfb_{-i}, 0) = 0$ and $x_i(\bfb) = 1$.

Let $\Delta = \mu(\bfb) - \mu(\bfb_{-i}, 0) > 0$ and $\epsilon = \frac{1}{2} \cdot \min(\Delta, p_i(\bfb)) > 0$.
Imagine that everyone else except user $i$ is bidding $\bfb_{-i}$, and 
user $i$'s true value is $v_i = p_i(\bfb) - \epsilon > 0$.
Due to Myerson's Lemma, since $v_i < p_i(\bfb)$, 
user $i$'s bid would be unconfirmed if it were to bid truthfully.
In this case, 
by Claim~\ref{clm:inconsequentialbidchange}, 
the miner's utility is $\mu(\bfb_{-i}, 0)$ and user $i$'s utility is zero.

However, the miner can sign a side contract and ask user $i$ to bid $p_i(\bfb)$ instead. 
By Myerson's Lemma, 
at this moment, user $i$'s bid will indeed be confirmed. 
By Claim~\ref{clm:inconsequentialbidchange}, 
the miner's utility is now $\mu(\bfb)$ and user $i$'s 
utility is now $v_i - p_i(\bfb) = -\epsilon$.
Consequently, their joint utility becomes $\mu(\bfb) - \epsilon$, 
which has increased by $\Delta - \epsilon > 0$.
This violates 1-SCP.
\end{proof}

\paragraph{Proof of Theorem~\ref{theorem:deterministic}.}
Theorem~\ref{theorem:deterministic} follows directly from the combination of 
Lemma~\ref{lem:decrease}
and Lemma~\ref{lem:nodecrease}.

\subsection{Randomized Mechanisms: UIC + 1-SCP $\Longrightarrow$ Zero Miner Revenue}
\label{sec:randomized-lb}

We now generalize Theorem~\ref{theorem:deterministic}
to even randomized mechanisms.
In a randomized TFM, the random coins
could come from either the miner or the blockchain itself.
Since we are proving an impossibility, without loss of generality,
we may assume that the blockchain 
comes with an unpredictable random source. 
Our impossibility result actually does not care 
where the random coins come from.

Earlier in Section~\ref{sec:roadmap-main-lb}, we presented
the intuition for this impossiblity.  
Therefore, below, we directly jump to the formal description.
 
\paragraph{Notations for randomized mechanisms.}
Recall that for randomized mechanisms,
the allocation rule now outputs the probability that each bid is confirmed;
that is, $x_i(\bfb) \in [0,1]$ is the 
probability that user $i$'s bid is confirmed given the included bids are $\bfb$.
Also, $p_i(\bfb)$ is now the expected payment of user $i$ 
and $\mu(\bfb)$ is the expected miner-revenue.

\elaine{moved a bunch of text to roadmap}

For convenience, we define the following quantity:
$$
\pi_{\bfb_{-i}}(r) = p_i(\bfb_{-i},r) - \mu(\bfb_{-i},r) 
$$
One can think of $\pi_{\bfb_{-i}}(r)$
as a meta-user $i$'s payment in the meta-auction 
(see Section~\ref{sec:roadmap-main-lb}).
The following theorem is a generalization of Theorem~\ref{theorem:deterministic}
to even randomized mechanisms.
\begin{theorem}[Randomized TFM: UIC + 1-SCP $\Longrightarrow$ 0 miner revenue]
\label{theorem:randomized}
	There is no randomized TFM with non-trivial miner revenue 
	that achieves UIC and $1$-SCP at the same time.
	Moreover, the theorem 
	holds no matter whether the block size is finite or infinite.
\end{theorem}

We will now prove this theorem. First, we introduce a useful lemma.
\begin{lemma}\label{lemma:randomInequality}
Let $({\bf x}, {\bf p}, \mu)$ be any randomized TFM.
If $({\bf x}, {\bf p}, \mu)$ is 1-SCP, then,
for any bid vector $\bfb$, user $i$, and $r, r'$ such that $r < r'$, it must be \[
	r \cdot \left(x_i(\bfb_{-i}, r') - x_i(\bfb_{-i},r)\right)
	\leq \pi_{\bfb_{-i}}(r') - \pi_{\bfb_{-i}}(r)
	\leq r' \cdot \left(x_i(\bfb_{-i}, r') - x_i(\bfb_{-i},r)\right).
\]
\end{lemma}
\begin{proof}
First, we prove the case of $r \cdot \left(x_i(\bfb_{-i}, r') - x_i(\bfb_{-i},r)\right)
\leq \pi_{\bfb_{-i}}(r') - \pi_{\bfb_{-i}}(r)$.
For the sake of reaching a contradiction, suppose there exists a vector $\bfb$, a user $i$ and $r < r'$ such that
\begin{equation}\label{eq:sidecontract}
	r \cdot \left(x_i(\bfb_{-i}, r') - x_i(\bfb_{-i},r)\right)
	> \pi_{\bfb_{-i}}(r') - \pi_{\bfb_{-i}}(r).
\end{equation}
Imagine that the real bid vector is $(\bfb_{-i}, r)$ and user $i$'s true value is $r$.
If they do not have a side contract, the miner's expected utility is $\mu(\bfb_{-i},r)$ and user $i$'s expected utility is $r\cdot x_i(\bfb_{-i},r) - p_i(\bfb_{-i},r)$.
However, the miner can sign a contract with user $i$ and ask user $i$ to bid $r'$ instead.
In this case, the miner's expected utility becomes $\mu(\bfb_{-i},r')$ and user $i$'s expected utility becomes $r\cdot x_i(\bfb_{-i},r') - p_i(\bfb_{-i},r')$ since the user's 
true value is still $r$.
By Eq.(\ref{eq:sidecontract}), their joint expected utility increases by 
$r \cdot \left(x_i(\bfb_{-i}, r') - x_i(\bfb_{-i},r)\right) - 
(\pi_i(\bfb_{-i},r') - \pi_i(\bfb_{-i},r))
> 0$.
This violates 1-SCP.

The other case $\pi_{\bfb_{-i}}(r') - \pi_{\bfb_{-i}}(r)
\leq r' \cdot \left(x_i(\bfb_{-i}, r') - x_i(\bfb_{-i},r)\right)$ can be proven by a similar argument, so we only sketch the proof.
Suppose the inequality does not hold, that is,
suppose that $\pi_{\bfb_{-i}}(r') - \pi_{\bfb_{-i}}(r)
> r' \cdot \left(x_i(\bfb_{-i}, r') - x_i(\bfb_{-i},r)\right)$. 
Imagine that the real bid vector is $(\bfb_{-i}, r')$ and user $i$'s true value is $r'$.
The miner can sign a contract with user $i$ and ask user $i$ to bid $r$ instead.
In this case, their joint expected utility increases by $\pi_{\bfb_{-i}}(r') 
- \pi_{\bfb_{-i}}(r) - r' \cdot \left(x_i(\bfb_{-i}, r') - x_i(\bfb_{-i},r)\right) > 0$.
This violates 1-SCP.
\end{proof}

\paragraph{Proof of Theorem~\ref{theorem:randomized}}
We now continue with the proof of Theorem~\ref{theorem:randomized}.
Consider the following quantity:
$$
\widetilde{\pi}_{\bfb_{-i}}(r) = p_i(\bfb_{-i},r) - \mu(\bfb_{-i},r) - 
(p_i(\bfb_{-i},0) - \mu(\bfb_{-i},0))
$$
By Lemma~\ref{lemma:randomInequality}, 
and the fact that definition of $\widetilde{\pi}_{\bfb_{-i}}(r)$
and ${\pi}_{\bfb_{-i}}(r)$ differs by only a fixed constant,
it must be that 
\begin{equation}
        r \cdot \left(x_i(\bfb_{-i}, r') - x_i(\bfb_{-i},r)\right)
        \leq \widetilde{\pi}_{\bfb_{-i}}(r') - \widetilde{\pi}_{\bfb_{-i}}(r)
        \leq r' \cdot \left(x_i(\bfb_{-i}, r') - x_i(\bfb_{-i},r)\right).
\label{eqn:sandwich}
\end{equation}
Now, observe that the above expression exactly 
agrees with the ``payment sandwich'' in 
the proof of Myerson's Lemma~\cite{myerson,myerson-lecture-hartline}.
Furthermore, we have that 
$\widetilde{\pi}_{\bfb_{-i}}(0) = 0$
by definition; and 
${\bf x}$ must be monotone because the TFM is UIC and satisfies
Myerson's Lemma. 
Due to Lemma~\ref{lem:sandwich}, 
it must be that 
$\widetilde{\pi}_{\bfb_{-i}}(\cdot)$ obeys the unique payment rule
specified by Myerson's Lemma, that is, 
\[
\widetilde{\pi}_{\bfb_{-i}}(r) = 
 b_i \cdot x_i(\bfb_{-i}, b_i) - \int_0^{b_i} x_i(\bfb_{-i}, t) dt.
\]
On the other hand, since the TFM is UIC, its payment rule  
itself must also satisfy the same expression, that is, 
\[
p_i(\bfb_{-i}, r) = 
 b_i \cdot x_i(\bfb_{-i}, b_i) - \int_0^{b_i} x_i(\bfb_{-i}, t) dt.
\]
We therefore have that 
\[
\widetilde{\pi}_{\bfb_{-i}}(r) = p_i(\bfb_{-i},r) - \mu(\bfb_{-i},r) - 
(p_i(\bfb_{-i},0) - \mu(\bfb_{-i},0))
= p_i(\bfb_{-i}, r)
\]
In other words, $\mu(\bfb_{-i},r) =  \mu(\bfb_{-i},0) - p_i(\bfb_{-i},0)$,
which 
is a constant that is independent of user $i$'s bid $r$ when $\bfb_{-i}$ is fixed.

We now argue that this 
actually implies $\mu(\bfb_{-i},r) = 0$, i.e., a possibly 
randomized TFM that is UIC and 1-SCP must 
always have 0 miner revenue.
Suppose this is not true, i.e., 
suppose there exists a randomized TFM with non-trivial miner revenue $(\bfx, \bfp, \mu)$ that is UIC and $1$-SCP.
Since it enjoys non-trivial miner revenue, there exists a bid vector $\bfb^{(0)} = (b_1,\ldots,b_m)$ such that $\mu(\bfb^{(0)}) > 0$.
Now, consider the following sequence of
bid vectors:  
for $i \in [m]$, let $\bfb^{(i)}$
be obtained by setting the first $i$ coordinates of $\bfb^{(0)}$ to $0$. 
Observe that $\bfb^{(m)} = {\bf 0}$.    

Recall that we have argued for a fixed $\bfb_{-i}$, 
the miner revenue $\mu(\bfb_{-i}, \cdot)$,
is a constant function independent of user $i$'s bid. 
Thus, 
$\mu(\bfb^{(i-1)}) = \mu(\bfb^{(i)})$ for all $i \in [m]$.
Consequently, we obtain $\mu(\bfb^{(0)}) = \mu(\bfb^{(m)})$.
However, users can only pay their bids at most, so we have $\mu(\bfb^{(m)}) \leq |\bfb^{(m)}|_1 = 0$.
This contradicts the assumption that $\mu(\bfb^{(0)}) > 0$.

\subsection{UIC + 1-SCP + Finite Block Size $\Longrightarrow$ Impossibility}
\label{sec:finite}

\elaine{TODO: in roadmap, write about implication for eip1559}

Theorem~\ref{theorem:randomized}
holds no matter whether the block size is finite or infinite.
In this section, we prove a corollary stating
that if the block size is finite, then no non-trivial TFM 
can satisfy UIC and 1-SCP
simultaneously. 
Particularly, assuming finite block size, 
the only TFM that satisfies both UIC and 1-SCP 
is the one that never confirms any transaction, and always pays the 
miner nothing.
This corollary holds assuming the following strategic behavior is possible:
an individual user 
or a user colluding with the miner can bid untruthfully;
and the miner can arbitrarily decide 
which transactions to include in the block (as long as 
it respects the block's validity constraint).

\begin{corollary}[UIC + 1-SCP + finite block size $\Longrightarrow$ impossibility]
Suppose the size of a block is finite.
Then, the only randomized TFM $(\bfx, \bfp, \mu)$ that satisfies both UIC and 1-SCP
is the trivial mechanism that 
never confirms any transaction no matter how users bid, and always 
pays the miner nothing.
\label{cor:finiteblocksize}
\end{corollary}
\begin{proof}
For the sake of reaching a contradiction, 
suppose that there is a non-trivial TFM that satisfies UIC and 1-SCP.
By Theorem~\ref{theorem:randomized}, any TFM that satisfies
both UIC and 1-SCP 
must have constant zero miner revenue. 
Henceforth, we may assume that the miner always gets zero payment.

Let $B$ denote an upper bound on the block size. 
Since the TFM is non-trivial, 
there exists a bid vector $\bfb = (b_1,\ldots,b_m)$ and a user $i^*$ such that $x_{i^*}(\bfb) > 0$.
Now, let $\epsilon$ be any positive number, let
$n > \frac{B \cdot (b_{i^*} + \epsilon)}{x_{i^*}(\bfb) \cdot \epsilon}$
be a sufficiently large integer.
Consider another bid vector $\bfb' = (b_1,\ldots,b_m, b_{m+1}, \ldots, b_{m+n})$ 
where $b_j = b_{i^*} + \epsilon$ for all $j \in [m+1, m+n]$.
Imagine that the real bid vector is actually $\bfb'$ and each user bids truthfully, 
i.e., user $j$'s true value is $v_j = b_j$ for all $j \in [m+n]$.
Since the block size is at most $B$, 
there must be a user $j \in [m+1, m+n]$ 
who bids $b_j$ is included 
with probability at most $B/n < 
\frac{x_{i^*}(\bfb) \cdot \epsilon}{b_{i^*} + \epsilon}$.

Consider the coalition of the miner and user $j$.
If everyone bids truthfully and the miners runs the honest
mechanism, then their joint utility is strictly less than  
$b_j \cdot \frac{B}{n} < (b_{i^*} + \epsilon) 
\cdot \frac{x_{i^*}(\bfb) \cdot \epsilon}{b_{i^*} + \epsilon} = x_{i^*}(\bfb) 
\cdot \epsilon$ --- since the miner always gets 0 revenue
and user $j$'s utility is upper bounded by 
$b_j \cdot \frac{B}{n}$.
However, the miner can sign a contract with user $j$.
The contract asks user $j$ to change the bid from $b_j$ to $b_{i^*}$, 
and the miner pretends that the actual bid vector is $\bfb$, 
where the coordinate $b_{i^*}$ actually comes from user $j$.
In this case, the coalition's joint utility 
is $(v_j - b_{i^*}) \cdot x_{i^*}(\bfb) = \epsilon \cdot x_{i^*}(\bfb)$.
Therefore, 
the coalition can increase its expected utility by deviating.
This violates $1$-SCP.
\ignore{
Notice that the miner's on-chain revenue is always zero, 
no matter which bids are included in the block.
However, the probability that user $j$'s bid is confirmed becomes $x_{i^*}(\bfb) > 0$.
Notice that the user's payment is zero when the bid is not confirmed, and the payment is upper bounded by its bid when the bid is confirmed.
Thus, the user's utility is zero when the bid is not confirm, and the utility is at least $v_j - b_{i^*} = \epsilon$ when the bid is confirmed.
Consequently, user $j$'s expected utility is at least $x_{i^*}(\bfb) \cdot \epsilon > 0$.
This violates 1-SCP.
}
\end{proof}

%% file: weakic.tex
\section{Rethinking the Incentive Compatibility Notions} 
\label{sec:weakic}

So far in our impossibility results, we have assumed 
it is free of charge for 
a strategic player to inject a fake transaction or overbid
(i.e., bid higher than its true value),  
as long as the offending transaction is not confirmed in the present block.
\ignore{
Consider the Vickrey auction example again: suppose we 
include in the block the highest $k$ transactions, among which
the top $k-1$ are confirmed, paying the $k$-th price; suppose that all payment
goes to the miner.
In this case, the miner may want to inject a fake transaction
$b'$ whose bid is in between the $(k-1)$-th and $k$-th price.
So far, we have assumed that since $b'$ is not confirmed, the miner 
can inject $b'$ free of charge.
}
Not only so, in fact,  
the same model was implicitly or 
explicitly adopted in earlier works
on transaction fee 
mechanism 
design~\cite{roughgardeneip1559,roughgardeneip1559-ec,functional-fee-market}, too.

Such a model, however, 
may be overly draconian, since there is actually some cost associated with cheating
that the existing model does not charge. 
In reality, the TFM is not a standalone auction, it is repeatedly
executed as blocks get confirmed.
Although 
an overbid or fake transaction need not pay fees to the present miner
if it is not confirmed,  
in real life, 
any transaction that has been submitted
to the network cannot be retracted.
Therefore, the offending transaction 
could be 
confirmed and paying fees in a future block
(e.g., paid to a different miner or simply burnt). 
\ignore{
Observe that in the proof of Claim~\ref{claim:burningBelowClaim}, 
we used the following strategy 
from a miner-user coalition: 
the miner asks the user to raise its bid to more
than its true value. 
Although this offending transaction 
is not confirmed in the present block, 
it can possibly increase the miner's revenue --- note that this means that the offending
transaction must be included (but not confirmed) in the block, since the blockchain
eventually uses included transactions  
to decide the prices $\bfp$ and miner revenue $\mu$. 
If this offending transaction 
is not confirmed in the present block, 
it need not pay any fees to the current miner.
In our current model as well as those adopted
in prior work~\cite{roughgardeneip1559,roughgardeneip1559-ec}, 
we simply assumed that such an offending transaction
that is not confirmed in the present incurs no cost. 

In practice, however, 
the TFM is not run in a standalone fashion 
as in our model. Any transaction submitted 
to the network can never be revoked, so an offending transaction
that is not confirmed in the present may be confirmed in the future. 
}
Consequently, a risk-averse miner-user  
coalition may be deterred from such deviations 
for fear of losing the offending transaction's fees to a future block. 

Therefore, a natural and interesting question is:
\begin{itemize}[leftmargin=5mm]
\item[] 
{\it If we fix the existing model and more carefully 
account for the cost of such cheating, can this
help us circumvent the impossibility results?}
\end{itemize}

One challenge we are faced with, however,  
is the difficulty of accurately 
characterizing the cost of such cheating. 
If an overbid or fake transaction is confirmed in a future block, 
it is hard for us to predict how much the offending transaction
will end up paying, since the payment amount may not be equal to the bid, and
the payment amount depends on the environment (e.g., the other bids), 
as well as the mechanism itself.

Despite this difficulty, we still want to understand whether this 
direction is worth exploring. 
A reasonable approach is to start by asking what is 
the worst-case cost. Once we understand what is the worst-case cost, 
we can consider how to define a more general, parametrized cost model. 
\begin{enumerate}[leftmargin=6mm]
\item 
{\it Worst-case cost.}
A worthwhile first step is to consider
the {\it worst-case cost} for the aforementioned deviation.
Specifically, 
whenever a strategic player 
injects a fake transaction or overbids and the offending
transaction cannot be confirmed in the present block, the strategic
player assumes the worst case scenario, i.e., the offending
transaction can end up paying fees as high as its bid in the future.

\Hao{Maybe it is helpful to mention that the worst-case cost is ``the easiest case'' for mechanism design. Thus, that why it gives us the strongest lower-bound.}
Assuming the worst-case cost is useful 
in several ways. First, it is {\it useful for proving 
lower bounds}. If we can prove lower bounds even for 
the worst-case cost, it would directly imply
lower bounds if in reality, the cost is actually smaller than the worst case.
Second, 
assuming the worst-case cost is also equivalent to considering 
strategic players who are {\it paranoid} --- they only want to deviate
if they will surely benefit, and there is no possible scenario
in which they will lose.
In other words, we are asking whether there is a mechanism that can at least
discourage such paranoid players 
from deviating.
\item 
{\it General, parametrized cost model.} 
As mentioned, it is challenging to accurately capture or predict
the cost of overbid or fake transactions that are unconfirmed in the present.
In practice, however, one might be able
to measure the cost of such cheating from historical data.
This motivates a more generalized cost model,
where we assume that there is some discount factor $\gamma \in [0, 1]$, 
and the cost of such cheating is actually $\gamma$ times the  
worst-case cost.
\ignore{
In particular, when $\gamma = 0$, there is no cost to cheat --- in this case,
our new definitions would be equivalent to the incentive compability  
notions in Section~\ref{section:definitions}.
When $\gamma = 1$, we are effectively assuming worst-case cost. 
}
\end{enumerate}

\ignore{
In this section, we will define 
a weaker notion of incentive compatibility.  
In this weaker variant, we 
assume that the strategic player is paranoid or risk averse.
Whenever it injects a fake transaction or overbids, and the offending
transaction cannot be confirmed in the present block, the strategic
player will assume the worst case, i.e., the offending
transaction will lose fees as high as its bid in the future.
}


\subsection{Defining $\gamma$-Strict Utility}
\ignore{We define a weak incentive compatibility  
notion which can be applied to a user, a miner, or a miner-user coalition.
We will then describe a mechanism that satisfies 
this weak incentive compatibility notion.
}
As we argued, the utility 
notions in prior work ignore certain costs associated with cheating. 
We therefore define a more refined utility notion that charges 
such cost parametrized by a ``strictness'' parameter $\gamma \in [0, 1]$. 
In other words, 
when $\gamma = 1$, we are charging the worst-case cost,
and equivalently, we are asking 
whether there are incetive compatible TFMs against {\it paranoid} players
who only want to deviate if there is a sure gain and no risk of losing.
We will also be using $\gamma = 1$ to prove lower bounds, and this gives stronger lower bound results.
When $\gamma = 0$, we are charging no cost --- in this case,
our new incentive compatibility definitions 
would be equivalent to the old notions
in Section~\ref{section:definitions}.

Recall that the term ``strategic player''
can refer to a user, a miner, or a miner-user coalition. 
An {\it offending transaction}
is one whose bid exceeds the transaction's true value: 
it can be an untruthful bid
or an injected fake transaction, since we 
may assume that a fake transaction's true value is $0$.
In the {\it worst-case} scenario, 
an offending transaction 
that is not confirmed in the present block may be charged 
a transaction fee 
equal to its full bid, when it is confirmed in a future block (possibly mined
by a different miner).
Let $v$ be the true value of the offending transaction (and $v=0$ if the
offending transaction is fake), 
and let $b \geq v$ be the bid value. 
Therefore, in the worst-case scenario, the offending transaction can cost  
$b - v$ in utility, due to losing fees to a future block.

In practice, if we can measure the actual cost from historical data, 
we may be able to learn a parameter $\gamma \in [0, 1]$, and model the actual 
cost as $\gamma$ times the worst-case cost, 
that is, $\gamma \cdot (b - v)$.

\elaine{in practice, the true value may even decrease 
if confirmed in future blocks.}

\paragraph{$\gamma$-strict utility.}
We now formally define the utility function 
of a strategic player:
\begin{mdframed}
\begin{center}
{\bf $\gamma$-strict utility}
\end{center}
\begin{itemize}[leftmargin=5mm,itemsep=1pt]
\item 
If the strategic player includes the miner, then let $u 
\leftarrow \mu$ where $\mu$ is the miner's revenue in the present block;
else let $u \leftarrow 0$.
\item 
For any real or fake transaction
the strategic player has submitted with true value $v$
and a bid of $b$: 
\begin{itemize}[leftmargin=5mm,itemsep=0pt]
\item 
if the transaction is confirmed in the present block, 
let $u \leftarrow u + v - p$ where $p$ denotes its payment.
\item 
if the transaction is {\it not} confirmed in the present block
and moreover
$b > v$, 
then let $u \leftarrow u - \gamma\cdot(b-v)$. See also Remark~\ref{rmk:overbid-weakic}.
\end{itemize}
\item 
Output the final utility $u$.
\end{itemize}
\end{mdframed}

\begin{definition}[Incentive compatibility under $\gamma$-strict utility]
Let $X \in \{\text{UIC, MIC, $c$-SCP}\}$.
We can now define $X$ under $\gamma$-strict utility 
just like 
in Definitions~\ref{defn:uic},
\ref{defn:mic}, and \ref{defn:scp}, respectively, except
that now we adopt the aforementioned $\gamma$-strict utility.
\end{definition}

\begin{definition}[Weak incentive compatibility]
For convenience, for the special case $\gamma = 1$, 
we also refer to our incentive compatibility notions 
as {\it weak} incentive compatibility. More specifically, we use the following aliases:
\begin{align*}
\text{weak UIC} & = \text{UIC under $1$-strict utility}\\
\text{weak MIC} & = \text{MIC under $1$-strict utility}\\
\text{$c$-weak-SCP} & = \text{$c$-SCP under $1$-strict utility}
\end{align*}
\end{definition}

\ignore{
\begin{definition}[Weak incentive compatibility]
\ignore{
We say that a TFM is weakly incentive compatible 
for a strategic player, 
iff the following holds:
no deviations 
restricted to the aforementioned strategy space 
can increase the strategic player's utility 
in comparison with following the rules and bidding truthfully.

If a TFM is weakly incentive compatible for any  
user (or the miner, resp.),
we say that it is {\it weakly UIC} (or {\it weakly MIC}, resp.).
If a TFM is weakly incentive compatible for 
any miner-user coalition containing 
up to $c$ users, 
we say that the TFM is {\it $c$-weakly-SCP}.
}
The notions of weak UIC, 
weak MIC, and $1$-weak-SCP 
are defined exactly like in Definitions~\ref{defn:uic},
\ref{defn:mic}, and \ref{defn:scp}, respectively, except
that now we adopt the aforementioned new utility function.
\end{definition}
}

\begin{remark}
Since the
miner has the ability to include 
an arbitrary set of transactions in the block,
without loss of generality, 
we may assume that the following deviations never take place since they do not
help the miner or the miner-user coalition:
\begin{enumerate}[leftmargin=5mm,itemsep=1pt]
\item the miner or the miner-user coalition
never bids untruthfully for any transaction {\it not}
included in the block; 
\item miner or the miner-user coalition 
never 
injects a fake transaction that is {\it not} included in the block.
\end{enumerate}
Therefore, one can equivalently view our new utility definition 
as only charging an additional cost for overbid or fake
transactions that are unconfirmed but included in the block.
Note that any transaction included in the block must have been broadcast
to the network and cannot be retracted.
\label{rmk:overbid-weakic}
\end{remark}

\ignore{
It turns out that weak UIC and UIC are equivalent, but
weak MIC and weak SCP are strict relaxations
of MIC and SCP. 
In the following fact, we prove that weak UIC and UIC are equivalent. 
This means that the Myerson's lemma holds for weak UIC too.

\begin{fact}[Weak UIC = UIC]
A TFM is weak UIC if and only iff it is UIC.
\end{fact}
\begin{proof}
The direction that UIC implies weak UIC is straightforward.
Below, we prove that any weak UIC mechanism is also UIC.
For the sake of reaching a contradiction,
suppose there is some TFM 
that is weak UIC but not UIC.
This means that some user is incentivized
to deviate under the old utility notion, but would not be incentivized
any more under the new utility notion. 
In comparison with the old utility definition, 
the only difference in the new utility 
is that we now charge a cost for overbid or fake transactions
that are unconfirmed.
If there is some user deviation that is profitable
under the old utility but not the new one, such a deviation
must involve overbidding or injecting a fake transaction that is not confirmed. 
However, observe that under the old utility, 
\elaine{this is not true. but myerson's lemma still holds}
\end{proof}
}

%% file: proof-burn-2nd-price.tex
\subsection{Burning Second-Price Mechanism}
\label{sec:proof-burn2ndprice}

Earlier in Section~\ref{sec:burn2ndprice}, 
we presented the burning second-price
mechanism which 
can be parametrized with any $\gamma \in (0, 1]$ and $c \geq 1$;
we additionally explained how to realize the on-chain random coins needed
for the mechanism.
We now prove 
Theorem~\ref{thm:burn2ndprice}, that is, 
for any $c \geq 1$ and $\gamma \in (0, 1]$, 
the burning second-price 
auction satisfies UIC, MIC, and $c$-SCP under $\gamma$-strict utility.

We prove the properties one by one.
Throughout this proof, we assume the $\gamma$-strict utility notion. 
We may assume that $\gamma \in (0, 1]$, since if 
$\gamma = 0$, the burning second-price auction 
always confirms nothing and it trivially satisfies all these properties.

\paragraph{UIC.}
According to the utility definition for the user, 
any injected fake transaction 
cannot lead to an increase in the user's utility. 
Therefore, we may assume that the user 
does not inject any fake transactions, and the only strategic
behavior is bidding untruthfully. 

Let $\bfb = (b_1,\cdots,b_m)$ be an arbitrary bid vector, 
where $b_1\geq\cdots \geq b_m$ and user $i$ bids truthfully (i.e.~$b_i = v_i$).
\elaine{fixed some tie breaking issues.}
Suppose $i \geq k+1$ and 
thus $b_i \leq b_{k+1}$. 
If user $i$ bids honestly, its utility is $0$ since it is unconfirmed.
Imagine that 
user $i$ changes its bids to $b_i'$. 
There are two cases.
First, user $i$'s new bid $b'_i$ is 
still not ranked among the top $k$ (possibly after the tie-breaking).
In this case, 
its utility is either zero if $b'_i < v_i$ or negative  
if $b'_i > v_i$.
Second, the new bid $b'_i$ is now ranked among the top $k$. 
We have $b'_i \geq b_{k+1}$. 
Further, user $i$'s utility becomes \[
    (1 - \frac{\gamma}{c}) \cdot \gamma (b'_i - b_i) + \frac{\gamma}{c} \cdot (b_i - b_{k}),
\]
where the first term captures the cost if $b'_i$ is not confirmed, and the second term captures the cost if $b'_i$ is confirmed.
Since $b'_i - b_i < 0$ and $b_i - b_{k} \leq 0$, its utility only decreases.
\elaine{the 2nd ineuquality should be $\leq$.}

The case of $i \leq k$ can be shown by a similar argument. 
In this case, if user $i$ bids honestly, it is among the top $k$, and its utility 
is at least $0$.
Now, imagine user $i$ changes its bid to $b'_i$.
There are two cases. First, if $b'_i$ cause user $i$ to be no longer among the top $k$,
then its utility is $0$.
Second, with the new bid $b'_i$, user $i$ is still among the top $k$.
If it underbids its utility is the same as bidding honestly.
If it overbids, its utility is the same as bidding honestly conditioned on it is confirmed, and its utility decreases by $\gamma (b'_i - b_i)$ conditioned on it is not confirmed.
\Hao{I explain the case "with the new bid $b'_i$, user $i$ is still among the top $k$" more.}

\Hao{We may need to add a remark that if the user's fake transaction has true value, then the spliting bid attack can happen. For example, a user with true value $10$ may want to bid $3,3,3$. Once three bids are all among the top $k$, it gets better chance to be confirmed (user wants ``at least'' one of them to be confirmed).}

\paragraph{MIC.}
A miner has two kinds of strategies to deviate from honest behavior: not to choose the highest bids and to inject fake bids.
Without loss of generality, we assume that the miner chooses the included bids first, 
and then replaces some of the real bids with fake bids.
We may also assume that all injected fake bids are included in the block.
We will show that both steps would not increase the miner's utility.
Let $(c_1,\ldots,c_B)$ be the highest $B$ bids in 
the bid vector, where $c_1 \geq \cdots \geq c_B$.
The miner's revenue is $\gamma(c_{k+1} + \cdots + c_B)$.
Now, suppose the miner does not choose the highest bids.
Let $(d_1,\ldots,d_B)$ be the resulting bids, where $d_1 \geq \cdots \geq d_B$ ---
we may assume that 
there are always infinitely many $0$-bids that are ``for free'', and the miner
can choose these $0$-bids too.
Then, miner's revenue becomes $\gamma(d_{k+1} + \cdots + d_B)$.
Since $(c_1,\ldots,c_B)$ are the highest $B$ bids, 
we have $c_j \geq d_j$ for all $j \in [B]$.
Thus, miner's revenue does not increase.

We will next show that whenever the miner replaces an included 
real bid with a fake bid, the miner's utility does not increase.
Notice that if the fake bid is confirmed, it costs the $(k+1)$-th price among the included bids.
If the fake bid $b$ is unconfirmed, it costs $\gamma \cdot b$, since its true value is zero.
Let $\bfe = (e_1,\ldots,e_B)$ be an arbitrary bid vector where $e_1 \geq \cdots \geq e_B$.
The bids $\bfe$ may or may not be the highest bids and some of them may be fake.
Suppose the miner replaces $e_i$ with the fake bid $f$.
There are four possible cases.
\begin{enumerate}
\item $i \leq k$ and $f$ is among the top $k$
\item $i \leq k$ and $f$ is not among the top $k$
\item $i > k$ and $f$ is among the top $k$
\item $i > k$ and $f$ is not among the top $k$
\end{enumerate}
Henceforth, no matter which case, let $e'_1 \geq \cdots \geq e'_B$
denote the included bids after replacing $e_i$ with $f$.
Let $\mu := \gamma(e_{k+1} + \cdots + e_B)$
be the miner's revenue before replacing $e_i$ with the fake bid $f$,
and let $\mu' := \gamma(e'_{k+1} + \cdots + e'_B)$
be the miner's revenue after replacing $e_i$ with the fake bid $f$.

\Hao{removed: it must be $f \geq e_{k+1}$. we don't need this condition.}
In the first case, for each bid among the top $k$, the probability that it is confirmed is $\gamma/c$, so the extra cost for the miner is \[
    (1 - \frac{\gamma}{c}) \cdot \gamma \cdot f + \frac{\gamma}{c} \cdot e_{k+1} \geq 0 
\]
where the first term captures the expected cost if $f$ is not confirmed, and the second term captures the expected cost if $f$ is confirmed.
In this case, it is easy to see that 
$e'_{k + j} = e_{k + j}$ for any $j > 0$, and thus $\mu' = \mu$. 
Therefore, miner's expected utility does not increase.

In the second case, $f$ must be unconfirmed, so 
it costs the miner $\gamma\cdot f$ additionally to inject $f$.
Also, since $f$ is not among the top $k$, it must be $f \leq e_{k+1}$.
Because $e_{k+1}$ (or another bid equal to $e_{k+1}$) 
is among the new top $k$, the miner's revenue becomes $\gamma(e_{k+2} + \cdots e_B + 
f)$. \elaine{note: i changed this expression}
Thus, the miner revenue decreases by $\gamma(e_{k+1} - f)$. 
\elaine{i changed this expression}
Including the extra cost $\gamma\cdot f$, miner's utility actually decreases by \[
\gamma(f + e_{k+1} - f) \geq \gamma \cdot e_{k+1}.
\]

In the third case, it must be $f \geq e_k$.
Moreover, $e_k$ 
becomes the largest definitely unconfirmed bid, 
so miner's extra cost is 
$(1 - \frac{\gamma}{c}) \cdot \gamma \cdot f 
+ \frac{\gamma}{c} \cdot e_{k}\geq \gamma \cdot e_k$. 
However, it is not hard to see that $\mu' - \mu \leq \gamma \cdot e_k$.
Therefore, overall, the miner's expected utility does not increase.

In the fourth case, $f$ is unconfirmed, so it costs the miner $\gamma \cdot f$ additionally.
If $f \leq e_i$, then it must be $e'_{k+j} \leq e_{k + j}$ for any $j  > 0$.
Therefore, we have $\mu' \leq \mu$,
and the miner's revenue does not increase. 
Otherwise, if $f > e_i$, we have $(e'_{k+1} + \cdots + e'_B) - (e_{k+1} + \cdots + e_B) = f - e_i$.
Thus, the increase in miner revenue is $\gamma(e'_{k+1} + \cdots + e'_B) - \gamma(e_{k+1} + \cdots + e_B) \leq \gamma(f - e_i) \leq \gamma \cdot f$, which is strictly smaller than the extra cost.
Thus the miner's expected utility does not increase.

Finally, because $\bfe$ is an arbitrary vector which may include fake bids already, 
we conclude that the miner's expected utility does not increase even if there are multiple fake bids.

\ignore{
Let $i' \in [B]$ be the rank of $f$ after replacing $e_i$ with $f$.
For any $t \in [B - k]$, 
conditioned on $T = t$, 
there are the following cases:
\begin{itemize}[leftmargin=5mm]
\item 
if $i > k+t$ and $i' > k+t$, then the miner's expected revenue is unaffected 
after replacing $e_i$ with $f$;
\item  
if $i \leq k + t$ and $i' \leq k+t$, 
then the miner's revenue increases by at most $\gamma \cdot \max(0, f - e_i) \leq 
\gamma f$;
\item 
if $i \leq k + t$ and $i' > k+t$, the miner's revenue cannot increase;
\item 
if $i > k+t$ and $i' \leq k+t$, 
the miner's revenue increases by 
$$\gamma \cdot 
\left((e_{k+1} + \ldots + e_{k + t - 1} + f) - (e_{k+1} + \ldots + e_{k + t})\right)
\leq \gamma \cdot f
$$
\end{itemize}
Since the above holds for any $t \in [B-k]$, 
the miner's increase in revenue is not enough to outweight its cost $\gamma f$. 
}

\paragraph{$c$-SCP.}
A coalition of a miner and up to $c$ user(s) has three kinds of strategies to deviate from the honest behavior: 
the miner may not include the highest bids, 
the miner can inject fake bids, and some of 
the user(s) can bid untruthfully.
Let $C$ be the set of colluding users, where $|C| \leq c$.
Without loss of generality, we assume the coalition prepares the block in the following order.
\begin{enumerate}
\item The miner chooses the included bids arbitrarily. We may
imagine that there are infinitely many $0$-bids that are ``for free''
and the miner can choose from these as well.
\item The miner replaces some of the included real bids (not including the users in $C$) with fake bids.
Without loss of generality, we may assume that all injected fake bids are 
included in the block.
\item 
A subset of users in $C$ change their bids and bid untruthfully. 
\end{enumerate}
\elaine{removed: C included assumption, does not need to be explicit?}
We now show that the joint utility of the coalition does not increase after each step.

\paragraph{The first step of $c$-SCP.}
We may imagine that the miner deletes  
the real bids one by one, and then includes the highest among the remaining bids. 
We argue that after deleting each bid, 
the coalition's expected utility does not increase.
Let $\bfe = (e_1, \ldots, e_m)$ be the current bid vector where $e_1 \geq \cdots \geq e_m$, 
which may already have some bids
deleted from the real bid vector. 
Suppose the miner deletes a bid from $\bfe$.
If the deleted bid is not among the top $B$, then it does not affect the coalition's utility.
If the deleted bid is ranked between $[k+2, B]$, then the miner's utility cannot increase
and no user's utility increases. 

If the deleted bid is among the top $k$, 
then the miner's revenue decreases by at least $\gamma \cdot (e_{k+1} - e_{k+2})$.
Every user who was among the top $k$ before and after
this deletion has $\frac{\gamma}{c} \cdot (e_{k+1} - e_{k+2})$
increase in expected utility.
The bid $e_{k+1}$ (or another bid of equal value) 
now becomes among the top $k$, and its increase in expected utility
is also $\frac{\gamma}{c} \cdot (e_{k+1} - e_{k+2})$.
The utility of the user who got deleted decreases.
All other users' utilities are unaffected.
Thus, as long as the number of colluding users $|C| \leq c$,  
the increase in utility for users in $C$ is upper bounded by 
$\gamma \cdot (e_{k+1} - e_{k+2})$. Overall, the coalition does not gain in expected utility. 

If the deleted bid is ranked $k+1$ in $\bfe$, 
the miner's decrease in revenue 
is at least $\gamma \cdot (e_{k+1} - e_{k+2})$.
For each user among the top $k$ in $\bfe$,  
its increase in utility is $\frac{\gamma}{c} \cdot (e_{k+1} - e_{k+2})$.
The utility of all other users are unaffected.
Thus, as long as $|C| \leq c$, 
the increase in utility for users in $C$ is upper bounded by
$\gamma \cdot (e_{k+1} - e_{k+2})$. Overall, the coalition does not gain in expected utility.
\elaine{i rewrote the first step.}
\Hao{Looks good to me}

\paragraph{The second step of $c$-SCP.}
Let $\bfe$ be an initial bid vector which may already have some bids deleted,
and some real bids replaced with fake bids.
Suppose the miner replaces some $e_i$ where $i \in [B]$ with a fake bid $f$.
Due to the proof of MIC, the miner's utility does not increase after the second step.
If no user's expected utility increases after replacing a real bid with a fake one, then the coalition's
expected utility cannot increase.
Therefore, we only need to consider the cases in which there exists some user 
whose expected utility increases after replacement.
Recall that in the proof of MIC, we divided into 
four possible cases. In cases 1 and 3, no user's expected utility would increase.
Below, we focus on cases 2 and 4. 

In case 2, $i \leq k$ and $f$ is not among the top $k$.
In this case, for every user $j \in [k]$ and $j\neq i$, its expected
utility increases by $\frac{\gamma}{c}(e_{k+1} - \max(f, e_{k+2}))$.
The bid $e_{k+1}$ now becomes the top $k$, and 
its utility also increases by  $\frac{\gamma}{c}(e_{k+1} - \max(f, e_{k+2}))$.
The bid $e_i$'s expected utility decreases, and all other users' expected utilities are unaffected.
However, the miner's utility decreases by at least 
$\gamma \cdot (e_{k+1} - \max(f, e_{k+2}))$.
Therefore, as long as $|C| \leq c$, the coalition's expected utility does not increase.

In case 4, $i > k$ and $f$ is not among the top $k$.
For some user's utility to increase, it must be that 
$i = k+1$ and $f < e_{k+1}$, i.e., the payment price must have decreased to $\max(f, e_{k+2})$.
Similarly, for every user $j \in [k]$, its increase 
in expected utility is $\frac{\gamma}{c}(e_{k+1} - \max(f, e_{k+2}))$, and every other user's utility
is unaffected. The miner's decrease in utility is 
at least $\gamma \cdot (e_{k+1} - \max(f, e_{k+2}))$.
Therefore, as long as $|C| \leq c$, the coalition's expected utility does not increase.
\elaine{i rewrote this proof to be more clear}
\Hao{Looks good to me}

\paragraph{The third step of $c$-SCP.}
At this step, the colluding users change their bids one by one.
Without loss of generality, we assume the colluding users change their bids in 
an ascending order according to their true values; that is, the 
users with lower true values change their bids first.
Let $\bfe = (e_1,\ldots,e_B)$ be an arbitrary bid vector included
in the block, where $e_1 \geq \cdots \geq e_B$.
The bids $\bfe$ may or may not be the highest bids and some of them may be fake or overbidding bids.
Note that if any user whose bid is not 
included in the block changes its bid, the coalition's joint utility cannot increase.
Thus, we may assume that a colluding user $i$ included in the block changes its bid.
Since the colluding users change their bids one by one, that means user $i$ has not changed its bid before, and $e_i$ must be the user's true value.
Henceforth, we often use $e_i$ to refer to the user that placed this bid without risking ambiguity.
We will show that the joint utility of miner and all users in $C$ would not increase 
if $e_i$ changes its bid to $b_i$.
\Hao{I rewrite the narrative here}

When $e_i$ is replaced with $b_i$, there are four possible cases.
\begin{enumerate}[leftmargin=5mm,itemsep=1pt]
    \item $i > k$ and $b_i$ is among the top $k$
    \item $i > k$ and $b_i$ is not among the top $k$
    \item $i \leq k$ and $b_i$ is among the top $k$
    \item $i \leq k$ and $b_i$ is not among the top $k$
\end{enumerate}
Henceforth, no matter which case, let $e'_1 \geq \cdots \geq e'_B$
denote the included bids after replacing $e_i$ with $b_i$.

In the first case, $e_k$ becomes the largest unconfirmed bid.
$e_i$'s utility was zero before, and it becomes 
$(1 - \frac{\gamma}{c}) \cdot \gamma \cdot (e_i - b_i) + \frac{\gamma}{c} \cdot (e_i - e_k)$
afterwards. Since $b_i \geq e_k \geq e_i$, 
the decrease in utility is at least 
$\gamma (e_k - e_i)$. Besides $e_i$, all other users' utilities cannot increase.
The miner's revenue increases by  
at most $\gamma (e_k - e_{k+1}) \leq \gamma (e_k - e_i)$. 
Thus, the coalition's expected joint utility does not increase.
\ignore{
Because $e_k \geq e_{k+1}$, other users' payments never decrease and their utilities never increase.
Miner's on-chain revenue was $\gamma (e_{k+1} + \cdots + e_{k+k'})$ before, and it becomes $\gamma (e'_{k+1} + \cdots + e'_{k+k'})$ afterward.
However, $\sum_{z = k+1}^B e'_{z} - \sum_{z = k+1}^B e_{i} = e_k - v_i$.
Thus, miner's revenue increases at most by $\gamma(e_k - v_i)$, so the joint utility does not increase.
}

In the second case, we have $\sum_{i = 1}^{k'} e'_{k+i} - \sum_{i = 1}^{k'} e_{k+i} = b_i - e_i$.
Thus the miner's change in revenue is $\gamma \cdot (b_i - e_i)$.
For the users who are among the top $k$, their payment become $e'_{k+1}$.
There are two subcases.
\begin{itemize}[leftmargin=5mm,itemsep=1pt]
\item 
If $e_i$ is overbidding ($b_i > e_i$), 
it must be $e'_{k+1} \geq e_{k+1}$, so the utilities of the users among top $k$ do not increase.
$e_i$'s utility reduces from zero to $\gamma (e_i - b_i) < 0$.
Since the miner's revenue increases at most by $\gamma (b_i - e_i)$, 
the coalition's joint expected utility does not increase.
\item 
If $e_i$ is underbidding ($b_i < e_i$), it must be $e_{k+1} \geq e'_{k+1}$.
In this case, it must be $e_j \geq e'_j$ for all $j$.
Further, $e_i$'s utility is still zero; and  
the miner's revenue decreases at least 
$\gamma(e_i - b_i) \geq \gamma \cdot (e_{k+1} - e'_{k+1})$.
The utility of each user among top $k$ increases only by $\gamma(e_{k+1} - e'_{k+1})/c$.
All other users' utilities are unaffected.
Thus, even if the miner colludes with $c$ users, their joint expected utility does not increase.
\end{itemize}

In the third case, $e_i$'s utility and miner's utility do not change individually. 
Moreover, all other users' utilities do not change either, because $e_{k+1} = e'_{k+1}$.
Thus, the joint utility of the coalition does not change.

In the fourth case, 
the miner's revenue reduces from $\gamma(e_{k+1} + \cdots + e_B)$ to $\gamma(e_{k+2} + \cdots + e_B 
+ b_i) = \gamma(e_{k+1} - b_i)$.
We now consider each user's change in utility.
\begin{itemize}[leftmargin=5mm,itemsep=1pt]
\item 
Since the user $e_i$'s bid is replaced with $b_i$
and now becomes unconfirmed, its utility reduces from $\gamma(e_i - e_{k+1})/c$ to zero.
Also, note that $e_i$ (who now bids $b_i$) must belong to $C$.
\item 
For anyone that was among top 
$k$ before and after the replacement, its new payment is $\max(b_i, e_{k+2})$ if confirmed.
Thus its expected utility 
increases by $\frac{\gamma}{c} \cdot (e_{k+1} - \max(b_i, e_{k+2}))$.
\Hao{This case should be increases}
\item 
Now consider the user that bids $e_{k+1}$. This is the most complicated case.
Let $v$ be this user's true value. 
If this user is a coalition member, 
and its bid $e_{k+1}$ was previously changed,  
then we know that $v \leq e_i$ since 
we are changing the coalition users' bids 
in ascending order of their true value. 
In all other cases, $v = e_{k+1} \leq e_i$.  

After the replacement of $e_i$ with $b_i$, 
conditioned on not being confirmed, the user's utility 
does not change since previously it was always unconfirmed.
Conditioned on being confirmed, the user's utility increases by 
at most $v - \max(b_i, e_{k+2}) + \max(0, \gamma \cdot (e_{k+1} - v))$, where 
the part $\max(0, \gamma \cdot (e_{k+1} - v))$ is because the user might be overbidding, i.e., 
$v < e_{k+1}$, and before the replacement it was always unconfirmed. 
Therefore, the user's expected gain in 
utility is $\frac{\gamma}{c} (v - \max(b_i, e_{k+2}) + \max(0, \gamma \cdot (e_{k+1} - v)))$.

\item 
For every other user, its utility is unaffected.
\end{itemize}

Now, suppose that 
the user bidding $e_{k+1}$ belongs to the coalition.
We know that $e_i$, whose bid is being changed to $b_i$, belongs to the coalition too.
The joint utility of $e_{k+1}$ and $e_i$ increases by
$\frac{\gamma}{c} (v - \max(b_i, e_{k+2})
 + \max(0, \gamma \cdot (e_{k+1} - v))) - \frac{\gamma}{c} (e_i - e_{k+1})$.
If $e_{k+1} \geq v$, their increase in utility is upper bounded by  
\begin{align*}
& \frac{\gamma}{c} \big[ v - \max(b_i, e_{k+2})
 + \gamma \cdot (e_{k+1} - v)
\big] - \frac{\gamma}{c} (e_i - e_{k+1}) \\
= & \frac{\gamma}{c} \big[v - \max(b_i, e_{k+2}) 
 + \gamma \cdot (e_{k+1} - v)
- e_i + e_{k+1}\big]\\
= & 
\frac{\gamma}{c} \big[e_{k+1} - \max(b_i, e_{k+2})
+ \gamma \cdot (e_{k+1} - v)
 - (e_i - v) \big]\\
\leq & 
\frac{\gamma}{c} \big[e_{k+1} - \max(b_i, e_{k+2})
+ \gamma \cdot (e_{k+1} - v)
 - (e_{k+1} - v) \big]\\
\leq &  
\frac{\gamma}{c} (e_{k+1} - \max(b_i, e_{k+2}))
\end{align*}
If $e_{k+1} < v$, their increase in utility is upper bounded by  
$
\frac{\gamma}{c} (v - \max(b_i, e_{k+2}) - e_i + e_{k+1})
\leq \frac{\gamma}{c} (e_{k+1} - \max(b_i, e_{k+2}))$, too.

All other users' expected utilities are either unaffected 
or increases by at most 
$\frac{\gamma}{c} (e_{k+1} - \max(b_i, e_{k+2}))$.
Therefore, as long as $|C| \leq c$, the coalition's joint utility cannot increase. 
Suppose that
the user bidding $e_{k+1}$ does not belong to the coalition.
This case is easier since all users' expected utilities
cannot increase by more than $\frac{\gamma}{c} (e_{k+1} - \max(b_i, e_{k+2}))$.
Therefore, as long as $|C| \leq c$, the coalition's joint utility cannot increase.
\Hao{The proof looks good to me}

%% file: lb-weak.tex
\section{Randomness is Necessary for Weak Incentive Compatibility}
\label{sec:lb-rand-weak}

Recall that 
when $c = 1$, our burning 2nd price auction becomes deterministic; but for all 
$c > 1$, the mechanism is randomized. 
In this section, we show that 
the randomness is in fact necessary 
for $c \geq 2$. To state this impossibility result, 
we first need to introduce a new notion 
that captures ``non-degenerate'' mechanisms,
that is, 
we consider mechanisms that sometimes confirm $2$ or more transactions: 

\begin{definition}[2-user-friendly]
    We call a mechanism is \emph{2-user-friendly} if there exists a bid vector $\bfb$ such that $x_i(\bfb) = x_j(\bfb) = 1$ for some $i \neq j$.
\end{definition}

We prove 
the following impossibility result --- throughout this section, we will
assume that $\gamma = 1$ (also called weak incentive compatibility), 
since this makes our impossibility result stronger. 
The same impossibility result trivially extends to the case when $\gamma < 1$ as well. 

\begin{theorem}\label{theorem:twoweakSCPwithburn}
	Suppose the block size is finite. 
	Let $(\bfx, \bfp,\mu)$ be a deterministic mechanism.
	If $(\bfx, \bfp,\mu)$ is 2-user-friendly, then it cannot achieve weak 
UIC and $2$-weak-SCP at the same time.
\end{theorem}

We stress that the $2$-user-friendly restriction is in fact necessary for 
the above impossibility to hold. In particular, 
in Appendix~\ref{section:solitary}, we give a deterministic 
mechanism 
that always confirms only one transaction, and satisfies weak UIC, weak MIC, and 
$c$-weak-SCP for any $c$. 
\elaine{double check this statement} \Hao{Isn't it has been proven?}
Moreover, in Appendix~\ref{sec:solitarypostedprice}, we additionally show
that the finite block size requirement is also necessary for the above impossibility
to hold.

We presented a roadmap of the proof of Theorem~\ref{theorem:twoweakSCPwithburn} 
in Section~\ref{sec:roadmap-lb-weak}. Therefore, we now directly
jump to the detailed proof.
To prove Theorem~\ref{theorem:twoweakSCPwithburn}, we first 
prove that Myerson's lemma still holds for any 
deterministic,  weakly UIC mechanism.
\begin{fact}
Myerson's lemma holds for any deterministic, weakly UIC mechanism.
\label{fct:myerson-weakuic}
\end{fact}
\begin{proof}
\elaine{maybe just say user-dsic is equivalent to uic and suppress the use of this term?}
Recall that in the definition of UIC or weak UIC, a user's strategy space
involves not only bidding untruthfully, but also injecting fake transactions.
To prove that Myerson's lemma holds for weak UIC, we only care about
bidding untruthfully, and we do not care about injecting fake transactions.
Henceforth, if a mechanism disincentivizes an individual user from overbidding
or underbidding under the {\it old} utility notion, 
we say that it is user-DSIC (short for dominant-strategy-incentive-compatible).
Similarly, 
if a mechanism disincentivizes an individual user from overbidding
or underbidding under the {\it new} utility notion, 
we say that it is weakly user-DSIC.
Clearly, UIC implies user-DSIC and weak UIC implies
weakly user-DSIC. 
Since Myerson's lemma holds for user-DSIC, 
it suffices to show that any {\it deterministic} TFM that is 
weakly user-DSIC must be user-DSIC, too. 

Suppose for the sake of contradiction 
that there is a deterministic TFM that is weakly user-DSIC
but not user-DSIC. 
This means that there is an untruthful bidding strategy
that is profitable under the old utility notion (i.e., $0$-strict utility)
but not profitable any more under the new utility notion (i.e., $1$-strict utility). 
In comparison with the old utility, the only difference 
in the new utility is that ``overbidding but unconfirmed'' is charged 
an additional cost.
Therefore, such an untruthful bidding strategy 
as mentioned above must be overbidding but unconfirmed.
However, we know that under the old utility notion, such an untruthful bidding strategy
results in utility $0$ and thus is not profitable.
Thus the user does not want to 
adopt this strategy even under the old utility.
This leads to a contradiction.
\ignore{
First, observe that even with weak UIC, 
Myerson's Lemma still holds.
This is because no matter whether an 
overbid and unconfirmed transaction comes for free, 
an individual user 
would never want to overbid if the corresponding transaction is not confirmed.
Only when the overbid transaction is confirmed, 
is it possible for the user 
to increase its utility by overbidding.
}
\end{proof}

\ignore{
\elaine{change narrative after restructuring}
The solitary mechanism in Section~\ref{section:solitary} can resist the coalition
of the miner and any number of users, 
however, it confirms 
only one transaction.
It is natrual to ask whether we can design a mechanism 
sometimes confirm more than one user, 
while still achieving $c$-SCP for any $c$ as in the solitary mechanism.
This notion is formalized as follow.
}

\begin{lemma}\label{lemma:confirmInvariant}
    Let $(\bfx, \bfp,\mu)$ be a deterministic mechanism which is weak UIC and $2$-weak-SCP.
    Let $\bfb = (b_1, \ldots, b_m)$ be an arbitrary bid vector, where there exists a user $i$ having a confirmed bid, i.e., $x_i(\bfb) = 1$.
    Then, for any bid vector $\bfb' = (\bfb_{-i}, b'_i)$ such that $x_i(\bfb') = 1$, the followings holds.
    \begin{enumerate}
        \item Miner's revenue does not change; that is, $\mu(\bfb) = \mu(\bfb')$.
        \item For any user $j$, if $x_j(\bfb) = 1$ and $b_j > p_j(\bfb)$, it must be $x_j(\bfb') = 1$ and $p_j(\bfb) = p_j(\bfb')$.
    \end{enumerate}
\end{lemma}
\begin{proof}
    Because $x_i(\bfb) = x_i(\bfb') = 1$, we know that $p_i(\bfb) = p_i(\bfb')$ by Myerson's lemma.
Recall that there is no cost for overbidding as long as the bid is confirmed. Therefore,
    user $i$'s utility does not change no matter it bids $b_i$ or $b'_i$.
    However, if $\mu(\bfb) \neq \mu(\bfb')$, the miner can sign a side contract to ask user $i$ to bid the price that makes miner's revenue higher, thus violating 2-weak-SCP.
    For example, suppose $\mu(\bfb) < \mu(\bfb')$, 
    then, in case the actual bid vector is $\bfb$ (where everyone's bidding its true value), 
    the coalition of user $i$ and the miner can gain by having user $i$ bid $b'_i$ instead of
    its true value $b_i$.
    Similarly, if $\mu(\bfb) > \mu(\bfb')$, a symmetric argument holds. 
    Thus, it must be $\mu(\bfb) = \mu(\bfb')$.
    
    We next prove that $x_j(\bfb') = 1$ for any user $j$ with $x_j(\bfb) = 1$ and $b_j > p_j(\bfb)$.
    For the sake of reaching a contradiction, suppose that $x_j(\bfb') = 0$.
    We now show that the coalition of the miner, user $i$, and user $j$ 
    can gain if everyone's true value is $\bfb'$.
    Suppose user $i$ were to  bid its true value $b'_i$, user $j$'s utility would be $0$
    since $x_j(\bfb') = 0$.
    Therefore, the 
    coalition is better off having user $i$ bid $b_i$ instead.
    In this case, user $j$'s utility would be $b_j - p_j(\bfb) > 0$.
    Furthermore, as we have shown, $\mu(\bfb') = \mu(\bfb)$, and moreover,
    by Myerson's
    Lemma, user $i$'s payment and utility 
    do not change as long as it bids high enough to be confirmed.
    Thus the coalition gains positively by having user $i$ bid $b_i$ instead of its
    true value $b'_i$.
    This violates $2$-weak-SCP.
    
    Finally, we prove that $p_j(\bfb) = p_j(\bfb')$ for any user $j$ with $x_j(\bfb) = 1$ and $b_j > p_j(\bfb)$.
    Because we have shown $x_j(\bfb) = x_j(\bfb') = 1$, user $j$'s utility is $v_j - p_j(\bfb)$ if user $i$ bids $b_i$, and $v_j - p_j(\bfb')$ if user $i$ bids $b'_i$.
    Suppose for the sake of contradiction that 
     $p_j(\bfb) \neq p_j(\bfb')$.
    There are two cases. First, suppose that 
     $p_j(\bfb) > p_j(\bfb')$. Imagine now that everyone's true value is $\bfb'$.
    In this case, 
    the miner can collude with both user $i$ 
    and user $j$, and have user $i$ bid $b_i$ rather than its true value  
    to increase the coalition's joint utility. This violates $2$-weak-SCP.
    Similarly, we can rule out the casw where 
    $p_j(\bfb) < p_j(\bfb')$ due to a symmetric argument.
\end{proof}

\begin{lemma}\label{lemma:samePayment}
Let $(\bfx, \bfp,\mu)$ be a deterministic mechanism that achieves weak UIC and $1$-weak-SCP.
Then, for all users $i,j$, if $x_i(\bfb) = x_j(\bfb) = 1$, it must be $p_i(\bfb) = p_j(\bfb)$.
In other words, all confirmed users must pay the same price.
\end{lemma}
\begin{proof}
Suppose $i$ and $j$ are two confirmed users; that is $x_i(\bfb) = x_j(\bfb) = 1$.
For the sake of reaching a contradiction, we assume that $p_i(\bfb) \neq p_j(\bfb)$.
There are two possible cases:
\begin{enumerate}
\item 
At least one user's bid is higher than its payment; that is, either $b_i > p_i(\bfb)$ or $b_j > p_j(\bfb)$ (or both).
\item 
Both users pay their bids; that is, $b_i = p_i(\bfb)$ and $b_j = p_j(\bfb)$.
\end{enumerate}

We start from the first case.
Without loss of generality, assume $b_i > p_i(\bfb)$.
According to Lemma \ref{lemma:confirmInvariant}, user $j$ can increase its bid without changing user $i$'s confirmation and payment.
Furthermore, by Myerson's lemma, user $j$'s payment should not change.
Thus, we have another bid vector $\bfb' = (\bfb_{-j}, b'_j)$ such that $b_i > p_i(\bfb')$ and $b'_j > p_j(\bfb')$.
Using Lemma \ref{lemma:confirmInvariant} again, we can increase user $i$'s and user $j$'s bid arbitrarily while remaining their confirmation and payment.
Consequently, we have a bid vector $\bfc = (c_1,\ldots,c_m)$ such that user $i$ and user $j$ have the same bid.
Formally, $x_i(\bfc) = x_j(\bfc) = 1$, $c_i > p_i(\bfc) = p_i(\bfb)$, $c_j > p_j(\bfc) = p_j(\bfb)$ and $c_i = c_j$.
Without loss of generality, we assume $p_i(\bfc) > p_j(\bfc)$.
Imagine that the real bid vector is $\bfc$.
In this case, miner's utility is $\mu(\bfc)$, and user $i$'s utility is $c_i - p_i(\bfc)$.
The miner can sign a contract with user $i$, and switch user $i$'s and user $j$'s positions 
in the bid vector.
Since users $i$ and $j$ are bidding the same, the miner's revenue is unaffected
if their positions are switched. On the other hand, user $i$ and user $j$'s payments
will be switched as a result.  \elaine{note: we are dealing with tie-breaking 
implicitly in the modeling. the same approach was used in the 
finite-block size impossibility for strong ic}
Thus, user $i$'s utility has increased to $c_i - p_j(\bfb)$.
This violates $1$-weak-SCP.

Next, we analyze the second case.
Without loss of generality, we assume $b_i = p_i(\bfb) > b_j = p_j(\bfb)$.
By Myerson's lemma, user $j$ can increase its bid without changing its payment.
Thus, user $j$ can increase its bid to $b_i$, 
and we have a bid vector $\bfb' = (\bfb_{-j}, b_i)$.
By Lemma \ref{lemma:confirmInvariant}, miner's revenue should not change, so we have $\mu(\bfb) = \mu(\bfb')$.
If $x_i(\bfb') = 1$, it goes back to the first case, so we assume $x_i(\bfb') = 0$.
Now, imagine that the real bid vector is $\bfb'$.
In this case, miner's utility is $\mu(\bfb') = \mu(\bfb)$, and user $i$'s utility is zero.
However, the miner can sign a contract with user $i$, and ask it to bid $b_j$ instead.
Consequently, the miner prepares a bid vector $\bfb$, where $b_j$ comes from user $i$.
In this case, miner's utility is still $\mu(\bfb') = \mu(\bfb)$, while user $i$'s utility becomes $b_i - p_j(\bfb) > 0$.
This violates $1$-weak-SCP.
\end{proof}

\begin{lemma}\label{lemma:ordered}
Let $(\bfx, \bfp,\mu)$ be a deterministic mechanism that achieves weak UIC and $2$-weak-SCP.
Then, for all users $i,j$, if $x_i(\bfb) = 1$ and $x_j(\bfb) = 0$, it must be $b_i \geq b_j$.
In other words, 
the confirmed bids must be the highest $k$ bids for some $k \in \mathbb{N}$ where
$k$ may be a function of the bid vector. 
\end{lemma}
\begin{proof}
For the sake of reaching a contradiction, suppose that there exist two users $i,j$ such that $x_i(\bfb) = 1$ and $x_j(\bfb) = 0$, while $b_i < b_j$.
By Myerson's lemma, user $i$ can increase its bid to $b_j$ without changing its confirmation and payment.
Thus, we have a bid vector $\bfb' = (\bfb_{-i}, b_i = b_j)$ such that $x_i(\bfb') = 1$ and $p_i(\bfb') = p_i(\bfb)$.
There are two possible cases: either $x_j(\bfb') = 1$ or $x_j(\bfb') = 0$.

First, we assume $x_j(\bfb') = 1$.
Imagine the real bid vector is $\bfb$ which also represents everyone's true value.
In this case, miner's utility is $\mu(\bfb)$, user $i$'s utility is $b_i - p_i(\bfb)$, and user $j$'s utility is zero.
However, the miner can sign a contract with user $i$ and user $j$, and ask user $i$ to bid $b_j$ instead.
By Lemma \ref{lemma:samePayment}, we have $p_i(\bfb') = p_j(\bfb')$.
Because $p_i(\bfb') = p_i(\bfb)$ and $p_i(\bfb) \leq b_i < b_j$, we have $p_j(\bfb') < b_j$.
Besides, by Lemma \ref{lemma:confirmInvariant}, $\mu(\bfb) = \mu(\bfb')$.
Therefore, after signing the contract, miner's utility is still $\mu(\bfb)$, user $i$'s utility is still $b_i - p_i(\bfb)$, while user $j$'s utility becomes $b_j - p_j(\bfb') > 0$.
This violates $2$-weak-SCP.

Next, we assume $x_j(\bfb') = 0$.
Imagine the real bid vector is $\bfb'$ which also represents everyone's true value.
Recall that in $\bfb'$, 
user $i$ and user $j$ are bidding the same; however, user $i$ is confirmed
but user $j$ is not.
Furthermore, user $i$ 
is bidding strictly higher than its payment as we have shown above.
The coalition of the miner and user $j$
can strictly benefit, 
if the miner switched 
user $i$ and user $j$'s positions in the bid vector; since this does not affect
the miner's utility, but user $j$'s utility now becomes positive.
This violates $1$-weak-SCP.

\ignore{
In this case, miner's utility is $\mu(\bfb) = \mu(\bfb')$, and user $j$'s utility is zero.
However, the miner can sign a contract with user $j$, and ask user $j$ to bid $b_i$ instead.
The miner prepares a bid vector $\bfb$, where $b_i$ comes from user $j$ and $b_j$ comes from user $i$.
In this case, miner's utility is still $\mu(\bfb) = \mu(\bfb')$, while user $j$'s utility becomes $b_j - p_i(\bfb) > 0$ because $b_j > b_i \geq p_i(\bfb)$.
This violates $2$-weak-SCP.
}
\end{proof}

\paragraph{Notation for the universal payment.}
According to Lemma \ref{lemma:samePayment}, all confirmed users must pay the same price.
Thus, we may simplify the notation, 
and define $p(\bfb)$ to be the universal payment 
price for all confirmed users under the bid vector $\bfb$. 
If no one is confirmed under $\bfb$, then 
we define $p(\bfb) = 0$.
\elaine{i changed it to 0 rather than infty} 

\begin{lemma}\label{lemma:unconfirmedPayment}
Let $(\bfx, \bfp,\mu)$ be a deterministic mechanism that achieves weak UIC and $2$-weak-SCP.
Let $\bfb$ be a bid vector such that at least one user is confirmed.
Then, for any unconfirmed user $i$, it must be $b_i \leq p(\bfb)$.
\end{lemma}
\begin{proof}
For the sake of reaching a contradiction, suppose there exists an unconfirmed user $i$ such that $b_i - p(\bfb) = \Delta$ for some $\Delta > 0$.
Let $j$ be a confirmed user in $\bfb$.
By Lemma~\ref{lemma:ordered}, we know that $b_j \geq b_i$.
Now, consider another bid vector $\bfb' = (\bfb_{-j}, p(\bfb) + \Delta/2)$.
By 
Myerson's Lemma, user $j$ is still confirmed and is still paying $p(\bfb)$.
However, notice that user $j$'s new bid $p(\bfb) + \Delta/2 < b_i$. 
By Lemma~\ref{lemma:ordered}, user $i$ must be confirmed too under
the bid vector $\bfb'$, and by Lemma~\ref{lemma:samePayment}, 
it would be paying 
the same as user $j$, which is $p(\bfb)$.
Now, 
imagine that everyone's true value is the vector $\bfb$.
If everyone bids honestly, 
the miner's utility is $\mu(\bfb)$, user $j$'s utility is $b_j - p(\bfb)$, 
and user $i$'s utility is zero.
If the miner colludes with users $i,j$, 
the coalition can benefit by having 
user $j$ bid $p(\bfb) + \Delta/2$ instead.
In this case, miner's utility is still $\mu(\bfb)$ 
due to Lemma~\ref{lemma:confirmInvariant}, 
user $j$'s utility is still $b_j - p(\bfb)$, 
while user $i$'s utility increases to $\Delta/2$.
This violates $2$-weak-SCP.
\end{proof}


\begin{lemma}\label{lem:minerutilkto0}
Let $(\bfx, \bfp,\mu)$ be a deterministic, weak UIC, and $1$-weak-SCP mechanism.
Then, for any bid vector $\bfb$, any user $k$, and any $0< \Delta \leq b_k$, 
it must be that 
$\mu(\bfb) - \Delta
\leq \mu(\bfb_{-k}, b_k - \Delta) \leq \mu(\bfb)$.
\end{lemma}
\begin{proof}
We first prove the direction $\mu(\bfb_{-k}, b_k -\Delta) \leq \mu(\bfb)$. 
We want to show that if the users' true value is $\bfb$, 
but now user $k$ bids  
$b_k - \Delta$ instead of its true value $b_k$, 
then the miner revenue should not increase. 
There are two cases. 
\begin{itemize}[leftmargin=5mm,itemsep=1pt]
\item First, if user $k$ is unconfirmed under $\bfb$ or confirmed
but paying its full bid $b_k$, then its utility is $0$ under $\bfb$. 
In this case, obviously 
decreasing user $k$'s bid should not make the miner benefit; since otherwise
the coalition of user $k$ and the miner can benefit by having user $k$ bid 
$b_k - \Delta$ instead,
thus violating $1$-weak-SCP.  
\item 
Second, suppose that user $k$ is initially confirmed under $\bfb$
and moreover, $b_k > p(\bfb)$.
We can first decrease user $k$'s bid to exactly $\max(b_k - \Delta, p(\bfb))$, and let
$\bfb' := (\bfb_{-k}, \max(b_k - \Delta, p(\bfb)))$
be the resulting new bid vector.
Due to Myerson's Lemma, Lemma~\ref{lemma:confirmInvariant}, 
and Lemma~\ref{lemma:samePayment}, 
$x_k(\bfb') = 1$, $\mu(\bfb') = \mu(\bfb)$,
and $p(\bfb') = p(\bfb)$.
Then, we can decrease user $k$'s bid from $\max(b_k - \Delta, p(\bfb))$ 
to $b_k - \Delta$,
and due to the same argument as the first case, the miner's revenue should
not increase.
\end{itemize}

We next prove the other direction, that is, 
$\mu(\bfb) - \Delta \leq \mu(\bfb_{-k}, b_k - \Delta)$.
If no one is confirmed under $\bfb$, the statement trivially holds.
Henceforth, we assume that at least one user is confirmed under $\bfb$.
Again, there are two cases.
\begin{itemize}[leftmargin=5mm,itemsep=1pt]
\item 
First, suppose that user $k$ is not confirmed under $\bfb$ or confirmed but paying
its full bid $b_k$. Due to 
Lemma~\ref{lemma:unconfirmedPayment}, we know that $b_k \leq p(\bfb)$.
By Myerson's Lemma, user $k$ should be unconfirmed or confirmed but paying full bid 
under $\bfb' := (\bfb_{-k}, b_k - \Delta)$.
Now, suppose everyone's true value is actually $\bfb'$,  
notice that user $k$'s utility is $0$. 
We argue that if user $k$ bids $b_k$ instead of its true value $b_k -\Delta$,
it should not make the miner revenue increase by more than $\Delta$. 
If so, the coalition of user $k$ and the miner can strictly benefit 
by having user $k$ bid $b_k$ instead of its true value $b_k - \Delta$, since the cost
of such overbidding is at most $\Delta$ under the new utility notion.
This violates $1$-weak-SCP.
\item 
Second, suppose that user $k$ is confirmed under $\bfb$ and moreover
$b_k > p(\bfb)$. In this case, due to 
Myerson's Lemma, Lemma~\ref{lemma:confirmInvariant},
and Lemma~\ref{lemma:samePayment}, 
the miner's revenue should not be affected when 
we reduce user $k$'s bid to $\max(p(\bfb), b_k - \Delta)$.
We now further decrease user $k$'s bid from $\max(p(\bfb), b_k - \Delta)$ 
to $b_k - \Delta$ --- 
due to the same analysis as the first case, the miner's revenue 
should not increase 
by more than $\Delta$ in this process.
\end{itemize}

\ignore{
There are two cases. 
user $k$ is initially unconfirmed under $\bfb$, 
and the miner revenue increases
when user $k$ bids $0$ instead, then  
the miner's revenue increase when user 
$k$ decreases its bid from its true value $b_k$ to $0$, then the coalition
of 
}

\end{proof}

\begin{lemma}\label{lemma:paychangeslow}
Let $(\bfx, \bfp,\mu)$ be a deterministic mechanism which is weak UIC and $2$-weak-SCP.
Let $\bfb = (b_1, \ldots, b_m)$ be an arbitrary bid vector.
Suppose that there exist three different users $i,j,k$ such that $x_i(\bfb_{-k}, 0 ) 
= x_j(\bfb_{-k}, 0) = 1$, and moreover,
$b_i - p(\bfb_{-k}, 0) > b_k$ 
and $b_j - p(\bfb_{-k}, 0) > b_k$.
Then, it must be 
that $x_i(\bfb) = x_j(\bfb) = 1$ and
$p(\bfb) \leq p(\bfb_{-k}, 0) + b_k/2$.
\end{lemma}
\begin{proof}
We first prove that $x_i(\bfb) = x_j(\bfb) = 1$.
Suppose not, without loss of generality, let us suppose $x_i(\bfb) = 0$
since the case $x_j(\bfb) = 0$ has a symmetric proof.
Imagine that the real bid vector is $\bfb$ which also 
represents everyone's true value.
Suppose the miner and user $i$ form a coalition, 
and they replace user $k$'s bid with an injected $0$-bid.
Due to Lemma~\ref{lem:minerutilkto0}, the miner's utility decreases
by at most $b_k$ as a result.
However, user $i$ now 
becomes confirmed and 
its utility is $b_i - p(\bfb_{-k}, 0) > b_k$. Therefore, the coalition
strictly gains which violates $1$-weak-SCP.

We next prove that 
$p(\bfb) \leq p(\bfb_{-k}, 0) + b_k/2$.
For the sake of reaching a contradiction, suppose $p(\bfb) > p(\bfb') + b_i/2$.
Imagine the real bid vector is $\bfb$ which is also everyone's true value.
In this case, miner's utility is $\mu(\bfb)$, user $i$'s utility is $b_i - p(\bfb)$, and user $j$'s utility is $b_j - p(\bfb)$.
However, the miner can collude with user $i$ and user $j$, 
and miner replaces $b_k$ with an injected $0$.
Notice that injecting a $0$-bid costs nothing.
In this case, miner's utility becomes $\mu(\bfb') \geq \mu(\bfb) - b_k$, 
where the inequality follows from Lemma~\ref{lem:minerutilkto0}.
On the other hand, user $i$'s utility becomes $b_i - p(\bfb')$, and user $j$'s utility becomes $b_j - p(\bfb')$.
Because $p(\bfb) > p(\bfb') + b_k/2$, user $i$'s and user $j$'s utilities each increases more than $b_k/2$.
Consequently, the coalition's joint utility increases, which violates $2$-weak-SCP.
\end{proof}

\begin{lemma}\label{lemma:paychangeslow2}
Let $(\bfx, \bfp,\mu)$ be a deterministic mechanism which is weak UIC and $1$-weak-SCP.
Let $\bfb = (b_1, \ldots, b_m)$ be an arbitrary bid vector.
Let $i$ and $k$ be two different users, 
and suppose that $x_i(\bfb) = 1$ and $b_i - p(\bfb) > b_k$.
Then, $x_i(\bfb_{-k}, 0) = 1$
and $p(\bfb_{-k}, 0) \leq p(\bfb) + b_k$.
\end{lemma}
\begin{proof}
For the sake of contradiction, suppose 
either $x_i(\bfb_{-k}, 0) = 0$
or $p(\bfb_{-k}, 0) > p(\bfb) + b_k$.
Imagine that the real bid vector is $\bfb' := (\bfb_{-k}, 0)$, which also
represents everyone's true value.
Now, the miner replaces user $k$'s $0$-bid in $\bfb'$ with an injected
bid $b_k$. 
Injecting this bid costs at most $b_k$.
Due to Lemma~\ref{lem:minerutilkto0}, the miner's utility cannot decrease,
i.e., $\mu(\bfb) \geq \mu(\bfb')$.
However, consider user $i$'s utility. 
If $x_i(\bfb') = 0$  
but $x_i(\bfb) = 1$, user $i$'s utility has increased from $0$ to 
$b_i - p(\bfb) > b_k$.
Else, 
if $x_i(\bfb') = x_i(\bfb) = 1$, but 
$p(\bfb') > p(\bfb) + b_k$, then user $i$'s utility 
increases by strictly more than $b_k$ too. 
In either case, the coalition of the miner and user $i$
can strictly increase their 
joint utility by replacing user $k$'s bid with the injected $b_k$ bid, which violates
$1$-weak-SCP.
\end{proof}

\begin{lemma}\label{lemma:twointervals}
    Let $(\bfx, \bfp,\mu)$ be a deterministic mechanism which is $2$-user-friendly, weak-UIC and $2$-weak-SCP.
    Then, there exists a bid vector $\bfb = (b_1,\ldots,b_m)$ where two different 
users $i,j$ are confirmed, and moreover, $b_i > p_i(\bfb)$ and $b_j > p_j(\bfb)$.
\end{lemma}
\begin{proof}
   Since $(\bfx, \bfp,\mu)$ is 2-user-friendly, there exists a bid vector $\bfb$ such that at least two users' bids are confirmed. 
    There are three possible cases:
    \begin{enumerate}
        \item There are two users $i,j$ such that $b_i > p_i(\bfb)$ and $b_j > p_j(\bfb)$.
        \item 
Only a single confirmed 
user bids strictly above its payment. Without loss of generality,
we may assume 
user $j$ is confirmed and $b_j > p_j(\bfb)$; however, 
for any confirmed user 
$i \neq j$, $b_i = p_i(\bfb)$, and there exists
at least one such $i$. 
        \item For any confirmed bid $b_i$, it holds that $b_i = p_i(\bfb)$.
    \end{enumerate}
    The first case is exactly what we want. For the second case, 
we can raise $i$'s bid
by an arbitrary amount $\Delta > 0$, and the new bid vector
is  $(\bfb_{-i}, b_i + \Delta)$. 
By Myerson's Lemma, $i$ should still be confirmed  
and paying the same price. Due to Lemma~\ref{lemma:confirmInvariant},
$j$ should still be confirmed and paying the same price, too.
Therefore, the bid vector $(\bfb_{-i}, b_i + \Delta)$
satisfies the claim we want to prove.
Moreover, due to Lemma~\ref{lemma:confirmInvariant}, it must be $\mu(\bfb) = \mu(\bfb_{-i}, b_i + \Delta)$.\footnote{$\mu(\bfb) = \mu(\bfb_{-i}, b_i + \Delta)$ is not important for this proof, while this fact will be useful in the proof of Lemma \ref{lemma:twointervals2}.}

We now focus on the third case which is the trickiest.
Due to Lemma~\ref{lemma:samePayment}, it must be that
everyone confirmed has the same bid, and thus $b_i = b_j$. 
Fix an arbitrary $\Delta > 0$. 
Consider the bid vector $\bfb^*$ which is the same as $\bfb$ except
that user $i$ and user $j$'s bids are replaced with $b_i + \Delta$.
We claim that the bid vector $\bfb^*$ satisfies the claim we want to prove.
In other words, $x_i(\bfb^*) = x_j(\bfb^*) =1$
and $b_i + \Delta > p(\bfb^*)$. 
Suppose this is not the case. There are three cases:
\begin{enumerate}[leftmargin=5mm,itemsep=1pt]
\item both $i$ and $j$ are not confirmed under $\bfb^*$;
\item 
exactly one of them is not confirmed under $\bfb^*$ --- without loss of generality,
we may assume that $j$ is not confirmed under $\bfb^*$. 
In this case, by Lemma~\ref{lemma:unconfirmedPayment}, 
it must be that $p_i(\bfb^*) = b_i + \Delta$;
\Hao{I change the wording from ``at least'' to ``exactly''}
\item both $i$ and $j$ are confirmed under $\bfb^*$ but 
$b_i + \Delta = p(\bfb^*)$.
\end{enumerate}
In all of these cases, 
user $i$ and $j$ both have utility $0$ if the true values are $\bfb^*$.

Let $\bfb' := (\bfb_{-i}, b_i + \Delta)$.
\ignore{
By Myerson's Lemma, $i$ should still be confirmed
and paying the same price $p(\bfb') = p(\bfb)$.
It must be that $j$ is unconfirmed under $\bfb'$. 
Otherwise, 
$j$ must be paying the same price  
$p(\bfb') = p(\bfb)$ due to Lemma~\ref{lemma:samePayment} --- 
we are now back to the second case, and 
and using the same argument as the second case, we can conclude that $\bfb^*$
satisfies the claim, which violates our assumption.
}
By Lemma~\ref{lem:minerutilkto0}, 
$\mu(\bfb^*) \geq \mu(\bfb')  \geq \mu(\bfb^*) - \Delta$.
By Lemma~\ref{lemma:confirmInvariant}, $\mu(\bfb') = \mu(\bfb)$.
Thus, $\mu(\bfb^*) \geq \mu(\bfb)  \geq \mu(\bfb^*) - \Delta$.\footnote{$\mu(\bfb^*) \geq \mu(\bfb)$ is not important for this proof, while this fact will be useful in the proof of Lemma \ref{lemma:twointervals2}.}
Now, suppose the true values are $\bfb^*$.
The miner can collude with users $i$ and $j$, and have them bid $b_i = b_j$ instead. 
Both users $i$ and $j$ are paying $b_i$ in this case.
Therefore, each of them has utilty $\Delta$ now.
On the other hand, 
the miner's utility decreases by at most $\Delta$, and therefore
the coalition's utility increases.
This violates $2$-weak-SCP.
\end{proof}

\paragraph{Proof of Theorem \ref{theorem:twoweakSCPwithburn}.}
By Lemmas~\ref{lemma:samePayment}
and \ref{lemma:twointervals}, there exists a bid vector $\bfb = (b_1,\ldots,b_m)$ such that $x_1(\bfb) = x_2(\bfb) = 1$, $b_1 > p(\bfb)$ and $b_2 > p(\bfb)$ --- we can always relabel the bids to make any two confirmed users with positive utility 
labeled as users $1$ and $2$.

Let $B$ denote the block size, and we define 
$\Gamma = 2^{B + 8} \cdot |\bfb|_1 \cdot \max(m, B+1)$ 
to be
a sufficiently large number.
Now, we consider a bid vector $\bfc = (c_1,c_2,\ldots,c_m)$, where $c_1 = c_2 = \Gamma$ and $c_i = 0$ for all $i \geq 3$.
We are going to show that $x_1(\bfc) = x_2(\bfc) = 1$.
By the Myerson's Lemma and Lemma~\ref{lemma:confirmInvariant}, 
from $\bfb$, we can increase $b_1$ without 
changing the confirmation and the payment of 
$b_1$ and $b_2$, 
so $x_1(\Gamma, b_2, \ldots, b_m) = 1$ and $x_2(\Gamma, b_2, \ldots, b_m) = 1$.
Similarly, we can then increase $b_2$ and we obtain $x_1(\Gamma, \Gamma, b_3, \ldots, b_m) = 1$, $x_2(\Gamma, \Gamma, b_3, \ldots, b_m) = 1$ and $p(\Gamma, \Gamma, b_3, \ldots, b_m) = p(\bfb)$.
Next, we reduce all remaining bids to zero one by one.
Repeatedly applying Lemma~\ref{lemma:paychangeslow2}
and observing that $\Gamma$ is sufficiently large, 
we have that 
$x_1(\bfc) = x_2(\bfc) = 1$, 
and the payment increases by $\sum_{i = 3}^m b_i$ at most, 
so $p(\bfc) \leq p(\bfb) + \sum_{i = 3}^m b_i < |\bfb|_1$.

Without loss of generality, 
we may henceforth assume that $m > B$.
If not, that is, if $m < B$, 
we can always add $0$ bids one by one until there are at least $B+1$ bids
 --- we claim that this does not change
the miner's utility nor 
user $1$ or $2$'s confirmation status and payment.
To see this, notice that adding or removing a $0$-bid is free of charge for the miner
or a miner-user coalition.
Therefore, adding or removing a $0$-bid should not 
change the joint utility of the miner and user $1$ by $1$-weak-SCP. 
This implies that user $2$'s confirmation status and payment
should not change, since otherwise, user $2$'s utility would change, 
and thus the joint utility of the miner and  
users $1$ and $2$ would change. This means that the coalition
of the miner and users $1$ and $2$ can cheat by adding or removing a $0$-bid to increase 
their joint utility, thus violating $2$-weak-SCP.
By a symmetric argument, user $1$'s confirmation status or payment should not change either.
Now, since the joint utility of the miner 
and user $1$ should not change due to adding or removing a $0$-bid, the miner's utility should 
be unaffected too.
\ignore{
To see this, suppose adding a $0$-bid to $\bfc$ changes the miner's utility.
Suppose the miner's utility decreases. 
Then, suppose the actual bid is actually $\bfc || 0$, the 
miner can simply remove the $0$ 
to increase its utility, which violates weak MIC. 
\Hao{we use MIC here. Need to change theorem statement?}
Similarly, we can argue that the miner's utility does not increase. 
Now, suppose that adding a $0$-bid changes 
user $1$'s confirmation status or payment --- since miner's utility does not change,
this means that adding a $0$-bid changes the utility of the coalition
of the miner and user $1$. Using a similar argument but now
for the miner-user coalition, we can rule this out, too, due to $1$-weak-SCP.
}

Now, we consider another bid vector $\bfd = (d_1,d_2,\ldots,d_m)$, 
where $d_i = \Gamma$ for all $i\in [B + 1]$ and $d_i = 0$
for $i > B+1$.
We are going to show that $x_i(\bfd) = 1$ for all $i \in [B+1]$ --- 
note that this is sufficient for reaching a contradiction since the block size
is only $B$.
To see this, we start from $\bfc$, and increase 
the bids of each user $j \in \{3, \ldots, m\}$ one by one.
Intuitively, Lemma~\ref{lemma:paychangeslow} guarantees that 
if the payment grows at all during the process, it must grow slower 
than the increase in a user's bid, so 
at some point, user $j$'s bid will catch up with the payment, as long
as there is still a large enough gap left between $\Gamma$ and the payment.
Formally, 
since $\Gamma - p(\bfc) > 2|\bfb|_1$, 
by Lemma~\ref{lemma:paychangeslow}, it must be that 
$p(\Gamma, \Gamma, 2|\bfb|_1,0,\ldots,0) \leq p(\bfc) + 
|\bfb|_1 < 2|\bfb|_1$, and 
$x_3(\Gamma, \Gamma, 2|\bfb|_1,0,\ldots,0) = 1$.
We now further increase user $3$'s bid to $\Gamma$, and 
by Lemma~\ref{lemma:confirmInvariant}, 
users $1$ to $3$ remain confirmed and 
$p(\Gamma, \Gamma,\Gamma, 0,\ldots,0) < 2|\bfb|_1$.
By the same reasoning, 
since $\Gamma - p(\Gamma, \Gamma,\Gamma, 0,\ldots,0) > 4 |\bfb|_1$,  
by Lemma~\ref{lemma:paychangeslow}, it must be that 
$p(\Gamma, \Gamma, \Gamma, 4|\bfb|_1 ,\ldots,0) \leq 
p(\Gamma, \Gamma,\Gamma, 0,\ldots,0) + 
2|\bfb|_1 < 4|\bfb|_1$, and 
$x_4(\Gamma, \Gamma, \Gamma, 4|\bfb|_1,0,\ldots,0) = 1$.
Therefore, 
$p(\Gamma, \Gamma, \Gamma, \Gamma, 0,\ldots,0) < 4|\bfb|_1$.
We 
can now repeat this process and raise
the bid of each user $i \in [B+1]$ to $\Gamma$.
It is not hard to check that our choice of $\Gamma$ is sufficiently large
for the reasoning to go through in all steps.

\ignore{
\elaine{TO FIX: why is m greater than B?}

In this case, each user has the utility at least $3|\bfb|_1$.
Let $\tau$ denote the user's utility when the bid vector is $\bfd$.

Finally, imagine that there are $B+1$ users, each with a bid $\Gamma$.
No matter what the inclusion rule is, there must be a user cannot be included and gets the utility zero.
In expectation, user $i$'s bid is chosen by a honest miner is at most $B/(B+1)$, so the expected utility is $\frac{B}{B+1} \cdot \tau$.
The miner can sign a contract with user $i$, and guarantees that user $i$'s bid is included.
In this case, miner's utility is still $\mu(\bfd)$, while user $i$'s utility becomes $\tau$.
This violates $2$-weak-SCP.
}

%% file: incleqconf-main.tex
\section{Necessity for Blocks to Contain Unconfirmed Transactions}
Observe that in our buring second-price auction, 
not all transactions in the block are confirmed.
In particular, only the 
top $k$ have a chance of being confirmed, and 
remaining $B-k$ included bids are  
not confirmed. Instead, they serve the role of setting the price, i.e.,
they are used by the blockchain to compute the payment for each confirmed
bid and the miner revenue.
In the cryptocurrency community, there is an ongoing debate
whether including unconfirmed transactions in a block is a good idea.
The argument against this approach is that 
``real estate'' on the blockchain is a scarce resource, 
so we ideally do not
want to waste space including 
unconfirmed transactions in the block.

We argue that having unconfirmed transactions in the block 
indeed can lead to more versatile mechanisms.
To show this, we argue that 
if ``included'' must be equal to ``confirmed'',
then, even weakly incentive mechanisms are not possible. 
More specifically, we prove the following corollary: 

\begin{corollary}[Impossibility for ``included = confirmed'']
Assume that all transactions included in the block must be confirmed. 
Then, no (possibly randomized) TFM $({\bf x}, {\bf p}, \mu)$ 
with non-trivial miner revenue
can satisfy weak UIC, weak MIC, and 1-weak-SCP at the same time 
--- this impossibility holds no matter whether the block size
is finite or infinite.

Moreover, if the block size is finite, then the only (possibly randomized) TFM 
that achieves weak UIC, weak MIC, and 1-weak-SCP 
is the trivial mechanism that always confirms nothing  
and pays the miner nothing.
\label{cor:includedeqconfirmed}
\end{corollary}

\paragraph{Myerson's lemma still holds.}
To prove this corollary, an important stepping stone is
to prove that Myerson's lemma 
still holds for any weak UIC, weak MIC, 
and 1-weak-SCP (randomized) mechanism where ``included = confirmed''. It turns out that this 
is somewhat non-trivial to prove.

Recall that in a randomized mechanism, the random coins come from two sources:
1) the miner can flip random coins to decide which transactions 
to include in the block;
2) once the inclusion choices are made,
the blockchain flips random coins to determine  
which of the included transactions are confirmed, 
how much each confirmed transaction pays, and how much the miner gets.
In other words, the randomness in the inclusion rule
is chosen by the miner, whereas the randomness
in the confirmation, payment, and miner-revenue rules are chosen
by the blockchain. 
A strategic miner may choose its random coins  
arbitrarily and not uniformly at random, to increase its expected gain.
On the other hand, we assume that the blockchain's randomness
is trusted. 
In other words, 
we assume that the blockchain can toss fresh random coins 
that are revealed {\it after}
the miner commits to its inclusion decision --- this makes our impossibility  
result stronger,
since if the blockchain's randomness is revealed
to the miner earlier, it makes mechanism design even harder.

\paragraph{Terminology and notation.}
Fix an arbitrary bid vector $\bfb = (b_1, \ldots, b_m)$. 
Let $S \subseteq \{b_1, \ldots, b_m\}$ denote a subset
of these bids to include in the block. We often call $S$ an {\it inclusion outcome}.
Note that if ``included = confirmed'', the miner is essentially choosing
which transactions are confirmed directly, too.
Whenever the miner picks an inclusion outcome $S \subseteq \{b_1, \ldots, b_m\}$, 
it can calculate its expected
utility denoted $\E(\mu | S)$
where the expectation is taken over  
the choice of the blockchain's random coins.
We use the notation $\E(\mu)$ to denote the miner's expected
utility under $\bfb$,
 had it executed the TFM honestly.


Fix an arbitrary bid vector $\bfb = (b_1, \ldots, b_m)$.
We say that an inclusion outcome 
$S \subseteq \{b_1, \ldots, b_m\}$ is {\it possible} (w.r.t. $\bfb$), 
if it is encountered with non-zero probability 
in an honest execution of the TFM over $\bfb$.

\begin{lemma}
Suppose that a randomized TFM satisfies weak MIC, and moreover, 
any transaction included in the blockchain must be confirmed.
Fix an arbitrary bid vector $\bfb = (b_1, \ldots, b_m)$.
For any possible inclusion outcome $S \subseteq \{b_1, \ldots, b_m\}$
it must be that $\E(\mu | S) = \E(\mu)$.

As a direct corollary, for any 
possible inclusion outcomes $S, S' \subseteq \{b_1, \ldots, b_m\}$
it must be that $\E(\mu | S) = \E(\mu | S')$.
\label{lem:minerindiff}
\end{lemma}
\begin{proof}
\elaine{TODO: we need to fix the defns for MIC for randomized case to incorprate
``expected'' utility}
Suppose 
that there is a possible inclusion 
outcome $S$ where $\E(\mu|S) \neq \E(\mu)$. 
It must be that there is  a possible inclusion 
outcome $S^*$ where $\E(\mu|S^*) > \E(\mu)$. 
In this case, instead of choosing the miner coins at random as prescribed
by the mechanism, 
it strictly benefits the miner 
to choose the specific inclusion outcome $S^*$.
This violates weak MIC.
\end{proof}

\ignore{
\begin{lemma}
Suppose that a randomized TFM satisfies 
1-weak-SCP, 
and moreover, 
any transaction included in the blockchain must be confirmed.
Fix an arbitrary bid vector $\bfb = (b_1, \ldots, b_m)$, and consider 
an arbitrary inclusion outcome $S \subseteq \{b_1, \ldots, b_m\}$ that is possible
w.r.t. $\bfb$.
It must be that $S$ includes the highest $|S|$ 
number of bids\footnote{This does not preclude having two possible inclusion
outcomes $S$ and $S'$ 
that include different number of bids.}
from the bid vector $\bfb$.
\label{lem:inclhighest}
\end{lemma}
\begin{proof}
Suppose there is a possible inclusion outcome $S$ 
that does not include the highest $|S|$ bids, that is, there is some bid
$b_j$ that gets included and thus confirmed, 
but another bid $b_i > b_j$ is not included.
Imagine that the users' true values are reflected by the vector $\bfb$. 
The miner 
now colludes with the $i$-th user and the coalition adopts the following strategy.
The miner executes the inclusion rule honestly. If the inclusion outcome happens 
to be $S$, then 
the miner asks user $i$ to 
bid $b_j < b_i$ instead of its true value
$b_i$. \elaine{TODO: this also needs some clarification in modeling}
The miner now includes 
the set $S$ where the coordinate $b_j$ is replaced with user $i$'s bid (which
also equals $b_j$ now).
In comparison with truthful bidding and honest execution of the inclusion rule, 
the coalition's utility increases by 
$\Pr[S] \cdot (b_i - b_j) > 0$
if it adopts the above strategy. 
\end{proof}
}

\begin{lemma}
Suppose that a randomized TFM satisfies 
weak MIC and 1-weak-SCP, and moreover, 
any transaction included in the blockchain must be confirmed.
Suppose that under some bid vector $\bfb = (b_1, \ldots, b_m)$, 
there is at least one possible inclusion outcome 
that includes $b_i$, and at least one possible inclusion outcome
that does not include $b_i$.
Then, 
consider any possible inclusion outcome $S$ that includes $b_i$, 
it must be that 
conditioned on $S$, user $i$ pays 
its full bid $b_i$ with probability $1$.
\label{lem:payfullbid}
\end{lemma}
\begin{proof}
Suppose that the claim does not hold, i.e., there is a 
possible inclusion outcome 
that includes $b_i$, but user $i$ pays $p_i < b_i$; and moreover,
there is at least one possible inclusion outcome 
that does not include $b_i$.
Let $S^*$ be a possible inclusion outcome 
that includes $b_i$
that minimizes the payment of user $i$.
In this case, 
the miner can form a coalition with user $i$, and the miner
can choose the inclusion  
outcome $S^*$ with probability $1$.
Due to weak MIC and Lemma~\ref{lem:minerindiff}, the miner's utility is 
still $\E(\mu)$ when it adopts this strategy, i.e., the same as playing honestly.
However, user $i$'s utility is positive and is maximized
under $S^*$. Furthermore, since there is at least one possible
inclusion outcome that does not include $b_i$ where user $i$'s utility is $0$, 
it must be that user $i$'s expected utility
is strictly greater under this strategy than playing honestly.
Therefore, the coalition strictly benefits 
under this strategy, which violates $1$-weak-SCP.
\end{proof}

\begin{lemma}
Suppose that a randomized TFM satisfies weak UIC, weak MIC, and 1-weak-SCP,
and moreover, any included transaction must be confirmed.
Then, 
the TFM must satisfy the constraints imposed by the Myerson's Lemma.
\label{lem:myerson-incleqconf}
\end{lemma}
\begin{proof}
\elaine{copied from the deterministic version of this proof, repeated text.
consider refactoring}
Recall that a mechanism disincentivizes an individual user from overbidding
or underbidding under the old utility notion, 
we say that it is user-DSIC (short for dominant-
strategy-incentive-compatible). Similarly, if a mechanism disincentivizes 
an individual user from
overbidding or underbidding under the new utility notion, we say that it 
is weakly user-DSIC.
Clearly, UIC implies user-DSIC and weak UIC implies weakly user-DSIC,
since in our definitions of 
(weak) UIC, the user can misbehave in more ways 
besides over- or under-bidding.
Since Myerson's lemma
holds for user-DSIC, it suffices 
to show that any (randomized) TFM where ``included = confirmed''
and satisfying weak user-DSIC,  
weak MIC, and 1-weak-SCP must also satisfy user-DSIC.

Suppose that 
this is not true, i.e., 
there is some TFM where ``included = confirmed''
and satisfying weak user-DSIC,
weak MIC, and 1-weak-SCP, however, the TFM does not  
satisfy user-DSIC.
Notice that if a user underbids, its utility is the same
under the old and new utility notions.  
Therefore, 
there must exist a bid vector $\bfb = (b_1, \ldots, b_m)$ 
some user $j \in [m]$,  and a bid $b'_j > b_j$, such that 
the user $j$ is incentivized to overbid  
under the old utility notion, but not incentivized to overbid
under the new utility notion.
There are the following cases, and we rule each one out, which allows
us to reach a contradiction. Below we use the terms ``included''
and ``confirmed'' interchangeably, and we define $\bfb' := (\bfb_{-j}, b'_j)$.
\begin{itemize}[leftmargin=5mm]
\item 
{\it Case 1: user $j$ is confirmed with probability $1$ under $\bfb'$.}
In this case, user $j$'s utility is the same under the old
and new utility definitions, and therefore, it is not possible
that user $j$ wants to deviate under the old utility but does not want to
under the new utility notion.
\item 
{\it Case 2: user $j$ is unconfirmed with probability $1$ under $\bfb'$.}
In this case, under the old utility, user $j$'s utility is $0$ even when it 
bids $b'_j$. Therefore, user $j$ does not want to deviate under
the old utility notion which contradicts our assumption.
\item 
{\it Case 3: user $j$ sometimes confirmed and sometimes unconfirmed under $\bfb'$.}
Since the TFM is weak MIC and 1-weak-SCP, and satisfies
``included = confirmed'', by 
Lemma~\ref{lem:payfullbid}, 
whenever the user 
$j$ is confirmed, it must pay its full bid. 
Therefore, under the old utility notion, if user $j$ bids $b'_j$ instead,
its utility is always $0$.
This means that the user does not want to deviate 
under the old utility notion, which contradicts our assumption.
\end{itemize}

\end{proof}



\paragraph{Proof of Corollary~\ref{cor:includedeqconfirmed}.}
We now continue with the proof of Corollary~\ref{cor:includedeqconfirmed}.
The proof of Lemma~\ref{lemma:randomInequality}
also makes use of the strategic deviation
where a user colluding with the miner 
overbids relative to its true value.
Specifically, 
the proof of Lemma~\ref{lemma:randomInequality}
relies on the fact that such overbidding 
comes for free if the offending transaction is not confirmed. 
In general, this is not true under the new utility function
associated with weak incentive compatibility.
However, we now argue that 
if ``included'' must be equal to ``confirmed'', then, 
the effect of this  
deviation (where the overbid transaction is not confirmed)
can alternatively be realized in a way 
that is free of charge.

More concretely, instead of having 
the colluding user actually carry out the overbidding,  
we instead exploit the miner's ability to include
an arbitrary  
subset of the mempool in the block.
Let ${\bf b}$ be the current mempool, which
includes the colluding user's bid $b$.
If in the proof of Lemma~\ref{lemma:randomInequality}, the miner wants 
the user to overbid $b' > b$ instead, it can simply 
pretend that the colluding user's bid  
is $b'$. In other words,
the miner can simulate running the mechanism
on ${\bf b} \backslash \{b\} \cup \{b'\}$.
As a result, the transaction $b'$ 
would not be confirmed, and thus $b'$ 
would not be included in the block, either. 
Therefore, the fact that the colluding user has not 
authorized/signed the transaction $b'$ does not matter 
in carrying out this deviation\footnote{In other words, $b'$ exists only in the simulation in the miner's head,
but is not released to the public network.}.

It is easy to see that as long as Myerson's Lemma holds
and any overbid transaction
that is not confirmed in the present block comes for free, 
Lemma~\ref{lemma:randomInequality}
still holds. 
Therefore, 
we conclude that Lemma~\ref{lemma:randomInequality}
still holds even under weak UIC and 1-weak-SCP,
if we insist that ``included'' be equal to ``confirmed''.
Now, 
as long as Lemma~\ref{lemma:randomInequality} and Myerson's Lemma still hold,
the proof of 
Theorem~\ref{theorem:randomized}
follows in the same way as before,
and
so does the proof of 
Corollary~\ref{cor:finiteblocksize}.
\elaine{double check that the corollary holds after it's fixed}
\Hao{If the mechanism is deterministic, then I think Corollary~\ref{cor:finiteblocksize} still holds, because the attack in the proofx is actually asking the user to ``underbid.'' However, if the mechanism is randomized, that means whether a bid is confirmed is probabilistic. Then, it is not clear what's inclusion = confirmation. I think ``inclusion = confirmation'' only makes sense if the mechanism is deterministic.}

If we restrict ourselves to {\it deterministic} mechanisms, we can actually
prove a counterpart of 
Corollary~\ref{cor:includedeqconfirmed}
without having to even rely on weak MIC. 
This is formally stated in the following corollary:
\begin{corollary}
Assume that all transactions included in the block must be confirmed. 
Then, no deterministic TFM $({\bf x}, {\bf p}, \mu)$ 
with non-trivial miner revenue
can satisfy weak UIC and 1-weak-SCP at the same time 
--- this impossibility holds no matter whether the block size
is finite or infinite.
Moreover, if the block size is finite, then the only deterministic TFM 
that achieves weak UIC and 1-weak-SCP 
is the trivial mechanism that always confirms nothing  
and pays the miner nothing.
\end{corollary}
\begin{proof}
Almost the same as the proof of 
Corollary~\ref{cor:includedeqconfirmed}, except that now, since the TFM is promised to be deterministic,
we can use  
Fact~\ref{fct:myerson-weakuic} instead of Lemma~\ref{lem:myerson-incleqconf}
to establish the fact that Myerson's lemma still holds.
\end{proof}

%% file: conclusion.tex
\section{Conclusion and Open Questions}
\label{sec:conclusion}
Mechanism design in decentralized settings (e.g., cryptocurrencies)
departs significantly from the classical literature
in terms of modeling and assumptions, 
and thus is relatively little understood.
For example, our work shows that even how to formally define incentive compatibility
is subtle and requires careful thought. 
Our work helps to unravel the mathematical structures of incentive compatible
 TFMs, and we hope that our  
definitional contributions 
can serve as a basis 
for future work in this space.
Part of our contribution is also to expose
the general lack of understanding in this fascinating
area of mechanism design. 
The following are some open questions we expose:
\ignore{
\item 
Our burning second-price auction is $1$-SCP, but it does not resist  
collusion between the miner and $2$ users.
For exmaple, if the miner colludes with two of the top $k-1$ users and  
replaces the $k$-th bid with another lower bid, 
then the coalition's joint utility increases. 
Does there exist a non-trivial TFM that simultaneously 
satisfies weak UIC, weak MIC, and $c$-weak-SCP for $c \geq 2$?
} 
\begin{itemize}[leftmargin=5mm,itemsep=1pt]
\item 
Are there other reasonable relaxations 
in the modeling and in incentive compatible notions  
that allow us to circumvent the impossibility results 
we showed?
\item 
Can we formally model the TFM as a repeated game 
and reason about potential
strategic behavior over a longer time scale?
Although our approach which adopts the $\gamma$ parameter 
captures potential costs in the future, right now we still
model the auction as one-shot.
\item 
Can cryptography help in the design of transaction fee mechanisms?
For example, the elegant work of 
Ferreira and Weinberg~\cite{commit-credible-auction} showed
that using cryptographic commitments can help overcome some
of the lower bound results 
shown by Akbarpour and Li~\cite{credibleauction}. 
Other works also use cryptography to help alleviate
strategic manipulations on transaction ordering~\cite{fairblock,fairordering}.  
Whether we can apply such techniques to TFMs is an interesting future direction,
and a subsequent work by Shi, Chung, and Wu~\cite{crypto-decentralized-mech} 
took an initial
step at exploring this direction.
\item 
So far, we assume that the block reward (besides the transaction fees)
paid to the miner is a constant 
that does not affect our game-theoretic analysis. 
It is interesting to explore what the mathematical implications are
when the block reward need not be a constant and can depend on the bids
included in the block. For example, this may mean that  
the miner revenue could even exceed the total user payment.
\item 
Our burning second price mechanism demonstrates
a tradeoff between resilience and efficiency of the mechanism
which is inevitable as shown by our impossibility result. 
However, we currently do know know whether 
the burning second price mechanism achieves optimal efficiency
(in terms of the utilization of on-chain space and money burnt)
for every choice of $\gamma$. This is another exciting direction for future work.
\end{itemize}

%% file: acks.tex
\section*{Acknowledgments}
We gratefully acknowledge helpful technical discussions 
with Kai-Min Chung during an early phase of the project.
We also thank T-H. Hubert Chan for insightful technical discussions.
We thank 
Tim Roughgarden and Matt Weinberg
for helpful discussions at EC'21 about whether blocks
should contain unconfirmed transactions that are just
there to ``set the price''.

%% file: ICcomparison.tex
\section{Relations Between Incentive Compatibility Notions}
\label{sec:ic-compare}

The notions UIC, MIC, and $1$-SCP are incomparable as depicted
in Figure~\ref{fig:ICgraph}.
\begin{figure}[h]
    \centering
    \input{ICgraph.tex}
    \caption{Relationship among incentive compatibility notions. The same chart holds
for UIC, MIC, and $1$-SCP under $\gamma$-strict-utility for any $\gamma \in [0, 1]$.}
    \label{fig:ICgraph}
\ignore{\qquad
     \begin{subfigure}[b]{0.38\textwidth}
    \input{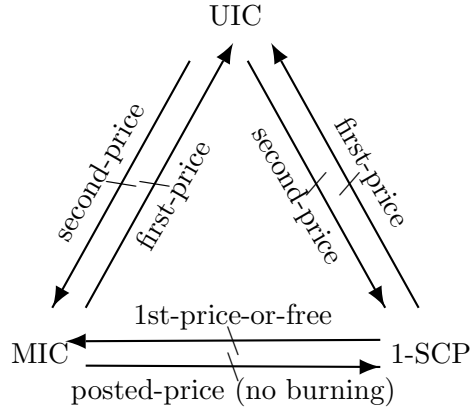}
\caption{Weak UIC, weak MIC, and $1$-SCP are incomparable.}
    \label{fig:ICgraph}
\end{subfigure}
}
\end{figure}

We explain Figure~\ref{fig:ICgraph} in more detail below:
\begin{itemize}[leftmargin=5mm]
\item
UIC $\not\Rightarrow$ MIC,  
UIC $\not\Rightarrow$ $1$-SCP: the second-price auction satisfies  
UIC, but does not satisfy MIC or $1$-SCP.  
This was pointed out  
in several earlier 
works~\cite{functional-fee-market,roughgardeneip1559,roughgardeneip1559-ec}.
Recall that in the second-price auction, the highest $B$ bids are included
in the block, the top $B-1$ are confirmed and they pay the $B$-th price,
where $B$ is the block size. The miner gets all payment. 
\item 
MIC $\not\Rightarrow$ UIC, $1$-SCP $\not\Rightarrow$ UIC:
the first-price auction satisfies MIC and $c$-SCP for any $c \geq 1$, but is
not UIC.
This was also pointed out
in earlier
works~\cite{functional-fee-market,roughgardeneip1559,roughgardeneip1559-ec}.
Recall that in the first-price auction, the top $B$ bids are included and confirmed,
they each pay their bid, and the miner gets all payment.
\item 
MIC $\not\Rightarrow$ $1$-SCP: 
the posted-price auction 
satisfies MIC but not $1$-SCP. Recall that in the posted-price auction,
there is a fixed reserve
price $r$.  Everyone  bidding at 
least $r$ is included and confirmed 
and pays exactly $r$. The miner gets all payment.
It is easy to check that the mechanism is indeed MIC. 
However, it is not $1$-SCP, since if a user's true value is $0 < r' < r$,
the miner can collude with the user, have the user bid $r$ instead, and the joint
utility of the coalition strictly increases.
\elaine{was this shown by roughgarden?}
\item 
$1$-SCP $\not\Rightarrow$ MIC: 
this is the most subtle to see. 
We construct the following 
``first-price-or-free'' mechanism which is $c$-SCP for any $c \geq 1$, but
not MIC.
The mechanism is not MIC since if there is only one bid, it makes sense
for the miner to inject a fake bid  
to increase its utility.
We show that the mechanism satisfies $c$-SCP for any $c \geq 1$ below.
\end{itemize}

\begin{mdframed}
    \begin{center}
    {\bf First-price-or-free mechanism} 
    \end{center}
    \begin{itemize}[leftmargin=5mm,itemsep=1pt]
    \item 
    Choose all bids in the current bid vector to include in the block.
    Let $\bfb = (b_1,\ldots,b_m)$ be the included bids of the block, where $b_1 \geq \cdots \geq b_m$.
    \item 
	Only the highest bid ($b_1$) is confirmed.
	Every other bid is unconfirmed.
	\item 
    If there is only one bid in the block ($m = 1$), the only confirmed user pays nothing.
    Otherwise, if $m \geq 2$, the only confirmed user pays $b_1$.
    \item 
    The miner gets all the payment.
    \end{itemize}
\end{mdframed}

\begin{theorem}
The first-price-or-free mechanism is $c$-SCP for all $c \geq 1$. 
\end{theorem}
\begin{proof}
Suppose there is only one user with true value $v^*$.
The miner and that user is the only possible coalition.
If they play honestly, that user's bid is the only bid in the block, so it must be confirmed.
Thus, in the honest case, the joint utility is $v^*$. 
If the coalition now deviates, 
then the confirmed bid is either a fake bid or the colluding user's bid.  
In either case, the coalition's utility cannot exceed $v^*$.

Suppose the number of users is $m \geq 2$.
There are two cases. 
First, suppose the highest bid $b_1$ does not belong to the
coalition if all colluding users are bidding truthfully. 
In this case, the coalition's utility is $b_1$ when it behaves honestly,
and $b_1 \geq v^*$ where $v^*$ denotes the highest true value of any colluding user.
Second, 
the highest bid $b_1$ belongs to the 
coalition if colluding users are bidding truthfully.
In this case, the coalition's utility is $b_1 = v^*$
if it behaves honestly.
In either case, we show that if the coalition deviates, it cannot gain. Suppose
$b'_1$ is the new highest bid after deviating.
If $b'_1$ belongs to the coalition, then the miner revenue offsets the coalition's payment,
and thus the coalition's utility cannot exceed  
the highest true value of any colluding user. 
If $b'_1$ does not belong to the coalition, it must be that $b'_1 \leq b_1$,
and the coalition's utility is $b'_1 \leq b_1$.

\end{proof}

\paragraph{Relationship for incentive compatibility notions 
under $\gamma$-strict-utility.}
Note that in Figure~\ref{fig:ICgraph}, for each arrow 
$X \not\Rightarrow Y$ shown by some example mechanism, 
it is easy to check that the same mechanism 
also shows that $X \not\Rightarrow \text{weak } Y$.
Thus,  Figure~\ref{fig:ICgraph}
in fact also holds for UIC, MIC, and $1$-SCP under $\gamma$-strict-utility,
for any choice of $\gamma \in [0, 1]$.

%% file: ICgraph.tex
\begin{tikzpicture}[
	squarednode/.style={rectangle, draw=white, minimum size=12mm},
	]
	\node[squarednode] (A) at (90:3) {UIC};
	\node[squarednode] (B) at (210:3) {MIC};
	\node[squarednode] (C) at (330:3) {$1$-SCP};
	\draw[decoration={markings,
	mark=at position 0.5 with \node[transform shape] (tempnode) {$\backslash$};,
	mark=at position 1 with {\arrow[scale=1.5,>={latex}]{>}}
	},postaction={decorate}, thick] 
	(A.225) -- node[rotate=60,above] {second-price} (B.75);
	\draw[decoration={markings,
	mark=at position 0.5 with \node[transform shape] (tempnode) {$\backslash$};,
	mark=at position 0 with {\arrow[scale=1.5,>={latex}]{<}}
	},postaction={decorate}, thick] 
	(A.255) -- node[rotate=60,below] {first-price} (B.45);
	\draw[decoration={markings,
	mark=at position 0.5 with \node[transform shape] (tempnode) {$\backslash$};,
	mark=at position 1 with {\arrow[scale=1.5,>={latex}]{>}}
	},postaction={decorate}, thick]
	(A.285) -- node[rotate=300,below] {second-price} (C.135);
	\draw[decoration={markings,
	mark=at position 0.5 with \node[transform shape] (tempnode) {$\backslash$};,
	mark=at position 0 with {\arrow[scale=1.5,>={latex}]{<}}
	},postaction={decorate}, thick]
	(A.315) -- node[rotate=300,above] {first-price} (C.105);
	\draw[decoration={markings,
	mark=at position 0.5 with \node[transform shape] (tempnode) {$\backslash$};,
	mark=at position 1 with {\arrow[scale=1.5,>={latex}]{>}}
	},postaction={decorate}, thick]
	(B.345) -- node[rotate=0,below] {posted-price (no burning)} (C.195);
	\draw[decoration={markings,
	mark=at position 0.5 with \node[transform shape] (tempnode) {$\backslash$};,
	mark=at position 0 with {\arrow[scale=1.5,>={latex}]{<}}
	},postaction={decorate}, thick]
	(B.15) -- node[rotate=0,above] {1st-price-or-free} (C.165);
\end{tikzpicture}

%% file: largeCoalition.tex
\section{Additional Results for Weak Incentive Compatibility}
In this section, we present some additional results that 
further unfold the mathematical landscape
of weakly incentive compatible mechanisms. 

\subsection{The Solitary Mechanism}
\label{section:solitary}
Earlier in Section~\ref{sec:lb-rand-weak}, we ruled out the existence 
of a deterministic, 2-user-friendly mechanism
that satisfies weak UIC and 2-weak-SCP simultaneously, assuming finite block size.
In this section, we show that the 2-user-friendly restriction is necessary
for this theorem to hold. 
In particular, we describe a mechanism called the solitary mechanism,
which always confirms a single transaction, and satisfies weak UIC, 
weak MIC, and $c$-weak-SCP for all $c \in \N$.
\Hao{I changed it to c-SCP.}  

\begin{mdframed}
    \begin{center}
    {\bf The solitary mechanism} 
    \end{center}
    \begin{itemize}[leftmargin=5mm,itemsep=1pt]
    \item 
    Choose the highest two bids to include in the block.
    \item 
    Only the highest bid is confirmed. 
    Other bids are all unconfirmed.
    The highest bid pays the second highest bid, and the miner is paid the second highest bid.
    \end{itemize}
\end{mdframed}

\begin{theorem}[The solitary mechanism]\label{theorem:solitary}
The solitary mechanism satisfies weak UIC, weak MIC, and $c$-weak-SCP for all $c > 0$.
The theorem holds no matter the block size is infinite or finite.
\end{theorem}
\begin{proof}
We prove the three properties one by one.

\paragraph{Weak UIC.}
A user has two kinds of strategies to deviate from the honest behavior:
to bid strategically or to inject fake transactions.
Since the payment is decided by the second highest bid, injecting fake transactions can only increase the payment, no matter whether the user is bidding truthfully or not.
Moreover, since 
this is exactly a classical second-price auction, bidding truthfully is known to 
be DSIC for an individual user, even under the old utility notion where overbidding
is never penalized.
Therefore, bidding untruthfully is not incentive compatible 
under the new utility notion as well.

\paragraph{Weak MIC.}
The miner has two kinds of strategies to deviate from honest behavior: not to choose the highest two bids and to inject fake bids.
\elaine{may need to add to defn: reordering bids.}
Without loss of generality, we assume that the miner chooses the included bids first, and replaces some of them with fake bids then.
We will show that both steps would not increase the miner's utility.

The miner's revenue is decided by the second highest bid that is included, 
so if miner does not choose the highest two bids to include, its revenue can only decrease or remain the same.
Next, suppose that the miner replaces 
one or both of the included bids with fake ones.
If after the replacement, the highest bid is a fake one, then the miner
has to pay the fee 
for the highest bid which is equal to its revenue. Therefore, the miner's utility cannot 
be greater than $0$.
If after the replacement, the highest bid is not a fake one but the second highest 
bid is a fake one,  
and the bid amount is $b$.
Then, the miner revenue is $b$. However, the cost to inject the fake bid $b$ is also $b$.
Thus, the miner does not gain overall. 
\ignore{
\elaine{this argument is a little buggy}
If the miner's injected bid is the highest bid, the miner pays the price to itself, and miner's utility becomes zero. 
If the miner's injected bid $b$ is the second highest bid, its revenue becomes $b$.
However, this fake transaction also costs miner $b$ since it is unconfirmed, so miner's utility is still zero.
Thus, any deviation does not increase miner's utility.
}

\paragraph{Weak $c$-SCP.}
Consider an arbitrary coalition of the 
miner and a subset of the users.
Suppose the coalition plays honestly: 
\begin{itemize}[leftmargin=5mm,itemsep=1pt] 
\item if the top confirmed bidder is in the coalition, 
the coalition's utility is top bidder's true value
denoted $v_1$;
\item 
if the top confirmed bidder 
is not in the coalition, then the coalition's utility is  
the true value of the 2nd bidder $v_2 \leq v_1$. 
\end{itemize}

Now, consider an arbitrary strategy 
where two bids are included and the higher of the two gets confirmed.
Each included bid can either come from some user, or 
is a fake bid.
Without loss of generality, we may equivalently assume that a fake bid belongs
to some imaginary user 
which belongs to the coalition, and its true value is $0$.
We may use the fake indices $0$ and $-1$
to refer to the one or two imaginary users.
There are the following cases:
\begin{itemize}[leftmargin=5mm,itemsep=1pt]
\item 
{\it Case 1: The confirmed user $i$ belongs to the coalition.}
In this case, the coalition's utility is upper bounded 
by (the possibly imaginary) user $i$'s true value $v_i$.
If $i$ has the highest true value, it means 
the coalition's utility is upper bounded by $v_1$; else
if $i$ does not have the highest true value, it means 
that the coalition's utility is upper bounded by $v_2$.
Either way, the coalition's utility cannot exceed the aforementioned honest case.

\item 
{\it Case 2: The confirmed user $i$ does not belong to the coalition.}
In this case, 
suppose that the included but unconfirmed user 
is denoted $j$ where $j$ is possibly
an imaginary user.
The coalition's utility
is $b_j - \max(0, b_j - v_j) \leq v_j$
where $b_j$ is user $j$'s bid, $v_j$ is its true value, and the
part $\max(0, b_j - v_j)$ is the penalty due to overbidding.
Since the confirmed user $i$ does not belong to the coalition, 
it  must be a real user and it must be bidding its truthful value, i.e., $b_i = v_i$.
Note that the $b_j$ cannot be bidding higher than $b_i$ since $b_j$ is unconfirmed
but $b_i$ is confirmed.
Therefore, 
$v_j \leq b_j \leq b_i = v_i$. This also implies that
that $v_j  \leq v_2$, and thus the coalition's utility is also upper bounded
by $v_2$. Recall that the coalition's utility
is at least $v_2$ had it played honestly; therefore, 
the coalition does not gain anything in comparison with 
playing honestly.
\end{itemize}


\end{proof}

\subsection{The Solitary-Or-Posted-Price Mechanism}
\label{sec:solitarypostedprice}
Earlier in Section~\ref{sec:lb-rand-weak}, we ruled out the existence 
of a deterministic, 2-user-friendly mechanism
that satisfies weak UIC and 2-weak-SCP simultaneously, assuming {\it finite block size}.
In this section, we show that the finite block size restriction is necessary for this lower bound
to hold, by showing a deterministic, 2-user-friendly mechanism that  
satisfies weak UIC, weak MIC, and 2-weak-SCP, but only under {\it infinite} block size.

\begin{mdframed}
    \begin{center}
    {\bf The solitary-or-posted price mechanism} 
    \end{center}
    \paragraph{Parameters:} a reserve price $r$.
    \paragraph{Mechanism:}
    \begin{itemize}[leftmargin=5mm,itemsep=1pt]
    \item 
Choose the top two bids as well as every other bid that is at least $r$ to include in the block.
\elaine{i changed this description. at least 2 must be included}
    \item 
    Every bid at least $r$ is confirmed, and the highest bid in the 
block is always confirmed (even if it is smaller than $r$).
    \item 
    Let $b_2$ be the second highest bid in the block.
    Every confirmed bid pays $\min(b_2, r)$, and miner is paid $\min(b_2, r)$. 
    The remaining payment is burnt.
    \end{itemize}
\end{mdframed}

\begin{theorem}[Solitary-or-posted-price mechanism]
Suppose the block size is infinite.
The solitary-or-posted-price mechanism 
satisfies weak UIC, weak MIC, and $c$-weak-SCP for all $c > 0$.
\end{theorem}
\begin{proof}
We prove the three properties one by one.

\paragraph{Weak UIC.}
A user has two kinds of strategies to deviate from the honest behavior:
to bid strategically or to inject fake transactions.
Since the payment is decided by $\min(b_2, r)$, injecting fake transactions can only increase the payment, no matter whether the user is bidding truthfully or not.

Suppose the real bid vector is $\bfb = (b_1,\ldots,b_m)$, where $b_1 \geq \cdots \geq b_m$.
Suppose user $i$ bids truthfully and other users may bid arbitrarily.
Let $v_i$ and $b'_i$ be user $i$'s true value and strategic bid, respectively.
\begin{itemize}
\item 
{\it Case 1: $b_2 \geq r$.}
User $i$ is facing a posted-price auction such that it if confirmed if and only if $b'_i \geq r$.
In this case, it is not hard to see that no matter user $i$ overbids or underbids, its utility does not increase.
\item 
{\it Case 2: $b_2 < r$.}
When $v_i = b_1$, then user $i$'t utility is $v_i - b_2 \geq 0$ in the honest case.
If user $i$ overbids, it still pays $b_2$ so the utility does not change.
If user $i$ underbids, it is either confirmed with the same payment, or becomes unconfirmed.
In either case, the utility does not increase.
\end{itemize}

\ignore{
\paragraph{Weak MIC.}
A miner has two kinds of strategies to deviate from honest behavior: not to choose the highest two bids and to inject fake bids.
\elaine{may need to add to defn: reordering bids.}
Without loss of generality, we assume that the miner chooses the included bids first, and replaces some of them with fake bids then.
We will show that both steps would not increase the miner's utility.

The miner's revenue is decided by the second highest bid that is included, 
so if miner does not choose the highest two bids to include, its revenue can only decrease or remain the same.
Next, if the miner's injecting bid is the highest bid, the miner pays the price to itself, and miner's utility becomes zero. 
If the miner's injecting bid $b \geq r$ is the second highest bid, its revenue becomes $r$.
However, this fake bid costs miner $r$ since it is confirmed and needs to pay, so miner's utility is still zero.
If the miner's injecting bid $b < r$ is the second highest bid, its revenue becomes $b$.
However, this fake bid also costs miner $b$ since it is unconfirmed, so miner's utility is still zero.
If the miner's injecting bid is not among the top $2$, it does not affect the revenue, while it may introduce extra cost, so miner's utility does not increase.
Thus, any deviation does not increase miner's utility.
}
\paragraph{Weak MIC.}
The miner has two strategies to deviate: not to choose the highest two bids and to inject fake bids.
Without loss of generality, we assume that the miner 
chooses the included bids first, and replaces some of them with fake bids then.
We will show that both steps would not increase the miner's utility.
The miner's revenue is decided by the second highest bid that is included, 
so if miner does not choose the highest two bids to include, its revenue 
can only decrease or remain the same.

Now, suppose the miner replaces some of the included bids with fake ones.
If any fake bid is confirmed, then the fake bid must be paying an amount
equal to the miner revenue, and thus the miner's utility is at most $0$.
If no fake bid is confirmed, and some fake bid is unconfirmed
and its bid amount is $b$. Then, 
it must be that $b$ is the second highest bid and the highest bid is smaller than $r$.
In this case, 
the miner gets revenue $b$; however, it costs $b$ 
to inject this fake bid. Thus, the miner does not gain overall.


\paragraph{Weak $c$-SCP.}
Suppose there are $m$ users, and their true values are $(v_1,\ldots,v_m)$ where $v_1 \geq \cdots \geq v_m$.
Henceforth, we also call the user with the highest true value $v_1$ the top user.
There are two possible cases.
\begin{itemize}[leftmargin=5mm,itemsep=1pt]
\item 
{\it Case 1: $v_2 < r$.}
When everyone behaves honestly, the miner's revenue is $v_2$, user $1$'s utility is $v_1 - v_2$, and all other users are zero since they are unconfirmed.
Suppose the miner colludes with a subset of users, and they prepare a bid vector $\bfe = (e_1,\ldots,e_m)$ where $e_1 \geq \cdots \geq e_m$.
Each bid $e_i$ in $\bfe$ is either a non-colluding bid coming from a non-coalition user in which case
$e_i = v_i$, or it is a colluding bid, i.e., one that comes from a colluding user
or a fake bid.
If only one user $i$ is confirmed in $\bfe$, there are two possibilities.
\begin{itemize}[leftmargin=5mm,itemsep=1pt]
    \item Suppose user $i$ is not in the coalition.
    The utility of the coalition is miner's revenue ($e_2$) minus potential extra cost
if there are overbid or fake bids that are unconfirmed. 
    If $i$ is not the top user, then $e_2 \leq v_i \leq v_2$.
    This means the utility of the coalition cannot exceed $v_2$, which can be achieved by playing honestly.

    If $i$ is top user, to make the miner's revenue larger than the honest case, 
it must be that $e_2$ is a colluding bid and $e_2 > v_2$. Let $v'$ be the true value
of this colluding user or $v' =0$ if $e_2$ is fake.
    Since $e_2$ is unconfirmed in $\bfe$, its utility becomes $v' - e_2$.
    Thus, the utility of the coalition cannot exceed $v' \leq v_2$, which can be achieved by playing honestly.

    \item Suppose user $i$ is in the coalition.
    The utility of the coalition is $v_i$ minus potential extra cost.
However, if top user is also in the coalition, 
the utility of the coalition is $v_1$ in the honest case, and is at most $v_i \leq v_1$ in the strategic case.
If the top user is not in the coalition, 
the utility of the coalition is $v_2$ in the honest case, and is at most $v_i \leq v_2$ in the strategic case.
\end{itemize}

If there are two or more confirmed bids in $\bfe$, the miner's revenue becomes $r$, 
while each confirmed user needs to pay $r$.
In this case, the utility of the top user decreases by $r - v_2$.
The utilities of all other bids (including fake ones) 
are non-positive, 
since if they are confirmed, they have to pay $r$ which is higher than their true values.
Since $v_2 < r$, there must be a colluding user (that is not the top user) 
bidding $b' \geq r$ or the miner injects a fake bid $b' \geq r$.
Let $v'$ be the true value of this colluding user or $v' = 0$ if the bid is fake.
The utility of this bid is $v' - r$.
The joint utility of this colluding or fake bid 
and the miner is $v' \leq v_2$, so the joint utility does not increase.
\item 
{\it Case 2: $v_2 \geq r$.}
When everyone behaves honestly, the miner's utility is $r$, 
user $i$'s utility 
is $v_i - r$ for all confirmed user $i$, and all other users are zero since they are unconfirmed.
In this case, miner's utility is already maximized.
Suppose the miner colludes with a subset of users, and they prepare a bid vector $\bfe$.
If there is only one confirmed bid in $\bfe$, then only the confirmed user can 
benefit by the deviation since its payment decreases.
However, the amount that the confirmed user gains is exactly what the miner loses, so the joint utility of any coalition does not increase.
If there are two or more confirmed bids in $\bfe$, then every user's utility is maximized when they bid truthfully, since the payment is fixed at $r$ regardless of others' bids.
\end{itemize}
\end{proof}


\subsection{Necessity of Burning}
Earlier in Section~\ref{sec:lb-rand-weak}, we ruled out the existence 
of a deterministic, 2-user-friendly mechanism
that satisfies weak UIC and 2-weak-SCP simultaneously, assuming {\it finite block size}.
In this section, we show that if the mechanism is not allowed to use a burning mechanism, i.e.,
if the miner's payment must be the sum of all users' payment, 
then, the same lower bound would hold even under infinite block size. 
This lower bound also shows that the burning in the solitary-or-posted-price mechanism
is necessary.

\begin{theorem}\label{theorem:twoweakSCP}
Let $(\bfx, \bfp,\mu)$ be a deterministic mechanism without burning.
If $(\bfx, \bfp,\mu)$ is 2-user-friendly, then it cannot achieve UIC and $2$-weak-SCP at the same time.
\end{theorem}

The remainder of 
this section will focus on proving Theorem~\ref{theorem:twoweakSCP}.
We first prove a useful lemma
that says in a mechanism satisfying the desired properties, 
if in some bid vector, all unconfirmed bids
are bidding $0$, then all confirmed bids must be paying $0$.

\begin{lemma}\label{lem:unconfirmed0}
Let $(\bfx, \bfp,\mu)$ be a deterministic mechanism without burning 
that is 
weak UIC and $2$-weak-SCP.
Suppose there exists a bid vector $\bfb = (b_1, \ldots, b_m)$ 
that confirms at least one bid, and moreover, all unconfirmed bids are $0$.
Then, all confirmed bids must pay $0$.
\end{lemma}
\begin{proof}
Due to Lemma~\ref{lemma:samePayment}, let $p := p(\bfb)$ denote the universal payment for $\bfb$.
Let $\epsilon < p/2m$ 
be a sufficiently small positive number. 
By Lemma~\ref{lemma:confirmInvariant} and Myerson's Lemma,
we can change all confirmed bids in $\bfb$ to 
$p + \epsilon$
such that all confirmed bids in $\bfb$ remain confirmed, and their payment unaffected.
Let $\bfb'$ be the resulting bid vector, and let $u \in [m]$
be the number of confirmed users in $\bfb'$, and recall all unconfirmed users bid $0$. 
Since there is no burning, $\mu(\bfb') = u \cdot p$.
Let $i$ be a confirmed user in $\bfb'$.
Now, suppose that the real bid vector is actually $\bfb'' := (\bfb'_{-i}, p-\epsilon)$, and this also
represents everyone's true value. User $i$ becomes unconfirmed in $\bfb''$ by Myerson's Lemma.
Thus $\mu(\bfb'') \leq (p+\epsilon) \cdot (u - 1)$.
In this case, the miner can collude with user $i$ and ask it to bid $p+\epsilon$ instead.
In this case, user $i$'s utility is $-\epsilon$, 
however, the miner's revenue is $p \cdot u > \mu(\bfb'') + \epsilon$.
Thus, the coalition 
can strictly gain 
from this deviation, which violates $1$-weak-SCP.

\ignore{
Due to Lemma~\ref{lemma:ordered}, one largest bid denoted $b_i$ must be confirmed.
By Lemma~\ref{lemma:confirmInvariant}, we can increase $b_i$ to an arbitrarily large amount
without affecting its payment. 
Suppose we have increased $b_i$ to some sufficiently large $\Gamma > |\bfb|_1$, and 
let $\bfb' := (\bfb_{-i}, \Gamma)$ be the resulting vector.
}

\end{proof}

\begin{lemma}\label{lemma:twointervals2}
Let $(\bfx, \bfp,\mu)$ be a deterministic mechanism without burning which is $2$-user-friendly, weak-UIC, $2$-weak-SCP, and with non-trivial miner revenue.
Then, there exists a bid vector $\bfb = (b_1,\ldots,b_m)$ where two different 
users $i,j$ are confirmed, and moreover, $\mu(\bfb)> 0$, $b_i > p_i(\bfb)$ and $b_j > p_j(\bfb)$.
\end{lemma}
\begin{proof}
It suffices to show that 
there exists a bid vector $\bfb$ such that $\mu(\bfb) > 0$ and at least two users' bids are confirmed.
If so, we can use the same argument in the proof of Lemma~\ref{lemma:twointervals}
to show that there exists 
a bid vector $\bfb'$ such that $\mu(\bfb') > 0$, and moreover at least two users' bids are confirmed,
and they are both bidding 
strictly higher than their payment.
In particular, cases 1 and 2 follow just like the proof of Lemma~\ref{lemma:twointervals}.
For case 3, 
suppose that $\bfb$ is a bid vector such that $\mu(\bfb) > 0$ and two different
users $i$ and $j$ are confirmed, and both bid exactly their payment.
In this case, the proof of Lemma~\ref{lemma:twointervals} constructed a new bid vector $\bfb'$ 
which is otherwise equal to $\bfb$ except that $b_i$ and $b_j$ now bid  
$b_i + \Delta$ for an arbitrary $\Delta > 0$,  
and showed that under $\bfb'$, both $i$ and $j$ are confirmed and bidding strictly above payment.
Here, 
we only need to additionally argue that $\mu(\bfb') > 0$.
This can be achieved by 
choosing $\Delta$ to be sufficiently small, and then applying Lemma~\ref{lem:minerutilkto0}.


Therefore, below, we focus on proving that 
there exists a bid vector $\bfb$ such that $\mu(\bfb) > 0$ and at least two users' bids are confirmed.
Suppose this is not true. In other words,  
for any bid vector $\bfb$ satisfying $\mu(\bfb) > 0$, only one user is confirmed. 
We will show that this contradicts $2$-weak-SCP.
\ignore{
For the first case, given $\bfb$, the argument in the proof of Lemma \ref{lemma:twointervals} guarantees there must exist a bid vector we want, including non-trivial miner revenue.
Henceforth, we will focus on the second case, and we are going to show that it must violate $2$-weak-SCP.
}

\ignore{
Suppose $\bfb = (b_1,\ldots,b_m)$ is the bid vector satisfying $\mu(\bfb) > 0$, while only user $i$ is confirmed.
We will show that there must be another user $j$ such that $b_j > 0$.
For the sake of reaching a contradiction, suppose $b_i$ is the only non-zero bid.
Since there is no burning and other users' bids are all zero, we have $\mu(\bfb) = p_i(\bfb) > 0$.
For any $0 < \epsilon < \mu(\bfb)$, imagine that the real bid vector is $(\bfb_{-i}, p_i(\bfb) - \epsilon)$.
Since user $i$'s bid is lower than $p_i(\bfb)$, it is unconfirmed.
However, other users' bids are all zero, so miner's revenue is zero.
The miner can collude with user $i$, and ask it to bid $p_i(\bfb)$ instead.
In this case, user $i$ is confirmed due to Myerson's Lemma, and its 
utility becomes $-\epsilon$, while miner's utility becomes $\mu(\bfb)$.
Thus the coalition strictly gains, and this 
violates $1$-weak-SCP. Therefore, there must be another user $j$ such that $b_j > 0$.
}

Since the mechanism is $2$-user-friendly, Lemma~\ref{lemma:twointervals} 
guarantees that there exists a bid vector $\bfc = (c_1,\ldots,c_m)$ such that $x_1(\bfc) = x_2(\bfc) = 1$ and $c_1 > p(\bfc)$ and $c_2 > p(\bfc)$.
By our assumption, it must be $\mu(\bfc) = 0$.
Because there is no burning, we have $p(\bfc) = 0$.
By Lemma~\ref{lemma:confirmInvariant}, we can increase 
user $1$'s and user $2$'s bids arbitraily without changing their confirmation and payment.
As a result, we obtain $\bfc' = (\Gamma, \Gamma, c_3, \ldots, c_m)$, where $\Gamma = |\bfc|_1$.
By Lemma~\ref{lemma:paychangeslow2}, we can now reduce each bid 
$b_3, \ldots, b_m$
down to zero one by one, without changing user $1$'s and user $2$'s confirmation.
Formally, we obtain a bid vector $\bfc'' = (\Gamma, \Gamma, 0, \ldots, 0)$ 
such that $x_1(\bfc'') = x_2(\bfc'') = 1$. By our assumption, both confirmed users in $\bfc''$ are 
paying $0$.

Since the mechanism has non-trivial miner revenue, there must exist
a bid vector $\bfb$ where $\mu(\bfb) >0$. By our assumption, only one user denoted $i$ is confirmed
in $\bfb$. By Lemma~\ref{lem:unconfirmed0}, there must be another user $j$ bidding non-zero.
Suppose $\bfb$ also represents everyone's true value.
In this case, miner's utility is $\mu(\bfb)$, user $i$'s utility is $b_i - \mu(\bfb)$, 
and user $j$'s utility is zero.
The miner can collude with user $i$ and user $j$, and ask them to bid $\Gamma$ instead.
Then, the coalition prepares a bid vector $\bfc'' = (\Gamma, \Gamma, 0, \ldots, 0)$ where the first two bids
are user $i$ and user $j$'s bids.
In this case, the miner's utility is zero, while user $i$'s utility is $b_i$ and user $j$'s utility is $b_j$.
The joint utility increases by $b_j > 0$, which violates $2$-weak-SCP.
\end{proof}

\paragraph{The proof of Theorem \ref{theorem:twoweakSCP}.}
By Lemma~\ref{lemma:samePayment} and 
Lemma~\ref{lemma:twointervals2}, there exists a bid vector $\bfb^{(0)} = (b_1,\ldots,b_m)$ such that $\mu(\bfb^{(0)}) > 0$, $x_1(\bfb^{(0)}) = x_2(\bfb^{(0)}) = 1$, $b_1 > p(\bfb^{(0)})$ and $b_2 > p(\bfb^{(0)})$ ---
note that we can always relabel the bids to make the first two bids represent two confirmed bids.

Now, one by one, we reduce every unconfirmed bid down to zero and increase 
every confirmed bid to a sufficiently large value $\Gamma > |\bfb|_1$.
Formally, 
for $i = 1,\ldots,m$, we define \[
    \bfb^{(i)} = \left\{\begin{matrix}
        (\bfb^{(i-1)}_{-i}, 0),      & \text{ if } x_i(\bfb^{(i-1)}) = 0,\\ 
        (\bfb^{(i-1)}_{-i}, \Gamma), & \text{ if } x_i(\bfb^{(i-1)}) = 1.
       \end{matrix}\right.
\]
By Lemma~\ref{lemma:confirmInvariant}, 
when increasing user $1$'s and user $2$'s bids, their confirmation, payment, and miner revenue do not change. 
Later on, when increasing any confirmed user $i$'s bid where $i > 2$,  
any previous user bidding $\Gamma$ 
would still remain confirmed and pay the same.
When decreasing any unconfirmed user $i$'s bid to $0$ where $i > 2$, since $\Gamma$
is sufficiently large, and by Lemma~\ref{lemma:paychangeslow2}, 
any previous user bidding $\Gamma$ would remain confirmed, and although their payment
may change, change in the payment is slow.
Thus, at the end, users $1$ and $2$ are confirmed in the final vector  
$\bfb^{(m)}$.
\ignore{
both of them will bid $\Gamma$ in $\bfb^{(m)}$.
By Lemma \ref{lemma:paychangeslow2}, the payment grows slowly, so we have $\Gamma > p(\bfb^{(m)})$.
Then, by Lemma \ref{lemma:unconfirmedPayment}, every user who bids $\Gamma$ must be confirmed in $\bfb^{(i)}$ for all $i$.
}

\ignore{
Next, there are two possible cases.
\begin{enumerate}
    \item 
    $\mu(\bfb^{(m)}) > 0$; and 
    \item 
    $\mu(\bfb^{(i)}) = 0$ for some $i = 1,\ldots,m$.
\end{enumerate}
}
By Lemma~\ref{lem:unconfirmed0} and the fact that there is no burning, 
it must be that $\mu(\bfb^{(m)}) = 0$.
\ignore{
We start from the first case. 
Conceptually, once we reduce confirmed users bid continuously, miner's revenue must ``suddenly drops'' at some moment.
At that moment, the miner has incentive to sign a contract with that user.
Formally, let $\epsilon = p(\bfb^{(m)}) / 2$.
By Lemma \ref{lemma:confirmInvariant}, we can tune each confirmed user's bid to $p(\bfb^{(m)}) + \epsilon$ without changing their comfirmation, payment, and miner's revenue.
As a result, we obtain $\bfb'$ such that $x_i(\bfb') = 1$ for all $i$ such that $x_i(\bfb) = 1$.
Additionally, $p(\bfb') = p(\bfb^{(m)})$, and $\mu(\bfb') = \mu(\bfb^{(m)})$.
Let $t$ be the number of confirmed users in $\bfb^{(m)}$.
Imagine the real bid vector is $(\bfb^{(m)}_{-1}, p(\bfb^{(m)}) - \epsilon)$.
Since user $1$'s bid is lower than $p(\bfb')$, it is unconfirmed and pays nothing.
In this case, user $1$'s utility is zero, and miner's utility is at most $(t-1)(p(\bfb') + \epsilon)$.
The miner can sign a contract with user $1$, and ask it to bid $p(\bfb')$ instead.
In this case, user $1$'s utility becomes $-\epsilon$, while miner's utility becomes $\mu(\bfb') = t\cdot p(\bfb')$.
It violates $2$-weak-SCP.
}
Let $i^*$ be the smallest integer $i \in \{1,\ldots,m\}$ such that $\mu(\bfb^{(i)}) = 0$.
Then, we have $\mu(\bfb^{(i^* - 1)}) > 0$ and $\mu(\bfb^{(i^*)}) = 0$.
By Lemma~\ref{lemma:confirmInvariant}, increasing 
a confirmed user's bid does not change miner revenue, so user $i^*$ must be unconfirmed in $\bfb^{(i^* - 1)}$.
Imagine the real bid vector is $\bfb^{(i^* - 1)}$ which also represents everyone's true value.
In this case, the miner's revenue is $\mu(\bfb^{(i^* - 1)})$, user $1$'s utility is $\Gamma - p(\bfb^{(i^* - 1)})$, and user $i^*$'s utility is zero.
The miner can collude with user $1$ and user $i^*$, and ask user $i^*$ to bid $\Gamma$ instead.
The coalition now prepares a bid vector $\mu(\bfb^{(i^*)})$ where the second coordinate $\Gamma$ actually comes from user $i^*$ and $b_{i^*} = 0$ is a fake bid injected by the miner.
Since there is no burning and $\mu(\bfb^{(i^*)}) = 0$, the payment must be zero.
Therefore, the miner's revenue becomes zero, 
while user $1$'s utility becomes $\Gamma$, and user $i^*$'s utility becomes $b_{i^*}$.
By Lemma~\ref{lem:minerutilkto0}, we have $b_{i^*} \geq \mu(\bfb^{(i^* - 1)})$, and thus
the coalition strictly gains from this deviation, which 
violates $2$-weak-SCP.

%% file: mechComparison.tex
\section{A More Detailed Discussion of Related Work}
\label{sec:detailedrelatedwork}
In this section, we discuss some known fee mechanisms, monopolistic price \cite{zoharfeemech}, random sampling optimal price (RSOP)\cite{competitiveauction,zoharfeemech}, BEOS mechanism \cite{functional-fee-market}, and analyze why they fail to achieve UIC, MIC, and $1$-SCP at the same time.
Roughgarden~\cite{roughgardeneip1559,roughgardeneip1559-ec}
also provided a summary of known results --- we provide some more details
in this section.

\subsection{Monopolistic Price}
The monopolistic price mechanism was introduced by  
Lavi et al.~\cite{zoharfeemech}. We describe the mechanism below.
\begin{mdframed}
\begin{center}
{\bf Monopolistic Price}
\end{center}
\paragraph{Parameters:} the block size $B$

\paragraph{Mechanism:}
\begin{itemize}[leftmargin=5mm,itemsep=1pt]
\item 
{\it Inclusion rule.}
Given the bid vector $b_1 \geq b_2 \geq \cdots$, the miner calculates
\begin{equation}\label{eq:monoPrice}
	k^* = \newargmax_{k \in [B]} k \cdot b_k.	
\end{equation}
Then, the miner chooses $(b_1,b_2,\ldots,b_{k^*})$ to be the block.
\item 
{\it Confirmation rule.}
All transactions in the block are confirmed.
\item 
{\it Payment rule and miner revenue rule.}
All transactions in the block pay the lowest bid in the block, and all payments go to the miner.
\end{itemize}
\end{mdframed}

As Example 2.2 in Lavi et al.~\cite{zoharfeemech} pointed out, monopolistic price is not UIC.
However, as conjectured by Lavi et al.\cite{zoharfeemech} and proven by Yao~\cite{yaofeemech}, the strategic gain of the users by deviating from truthful bidding goes to zero as the number of users goes to infinity.
Moreover, monopolistic price is also not $1$-SCP.
Consider the following example.
Suppose there are only two bids in the mempool, $(10,6)$.
If the miner is honest, it will choose both bids. 
In this case, the miner's revenue is $2 \cdot 6 = 12$, and the utility of the second user is $0$.
However, the miner can sign a contract with the second user and ask it to bid $10$ instead.
Now, the miner's revenue becomes $2 \cdot 10 = 20$, and the utility of the second user becomes $-4$.
Their joint utility becomes $16$, which increases by $4$.
It violates $1$-SCP.

As the authors in \cite{zoharfeemech} pointed out, a myopic miner has no incentive to deviate.
We formulate the notion as the following proposition.

\begin{proposition}
Monopolistic price is MIC.
\end{proposition} 
\begin{proof}
The miner has two kinds of strategies to deviate from honest behavior: 
not to choose the highest $k^*$ bids and to inject fake bids, where $k^*$ is decided by Eq.(\ref{eq:monoPrice}).
Without loss of generality, we assume that
the miner chooses the included bids first, and replaces some of them with fake bids then. 
We will show that both steps would not increase the miner's utility.

Notice that choosing the highest $k^*$ bids always gives the optimal revenue by Eq.(\ref{eq:monoPrice}).
Next, suppose the miner may or may not follow the inclusion rule, and prepares $\bfc = (c_1,\ldots,c_t)$, where $\bfc$ may or may not include some fake bids and $c_1 \geq \cdots \geq c_t$.
In this case, the revenue is $t \cdot c_t$.
Now, suppose the miner injects one more fake bid $f$.
If $f \geq c_t$, the revenue becomes $(t+1) \cdot c_t$, while the miner needs to pay $c_t$ for injecting $f$.
Thus, the overall utility does not increase.
If $f < c_t$, the revenue becomes $(t+1) \cdot f$, while the miner needs to pay $f$ for injecting $f$.
Thus, the overall utility becomes $t \cdot f$, which is less than not injecting $f$.

Finally, notice that the argument above holds no matter $B$ is finite or infinite.
\end{proof}

\subsection{Random Sampling Optimal Price}
Random sampling optimal price (RSOP) was introduced by Goldberg et al.~\cite{competitiveauction}, and Lavi et al.~\cite{zoharfeemech} analyzed the incentive compatibility in the context of fee mechanism.
We describe the mechanism below.

\begin{mdframed}
\begin{center}
{\bf Random Sampling Optimal Price (RSOP)}
\end{center}
\paragraph{Parameters:} the block size $B$

\paragraph{Mechanism:}
\begin{itemize}[leftmargin=5mm,itemsep=1pt]
\item 
{\it Inclusion rule.}
Choose the highest $B$ bids into the block.
\item 
{\it Confirmation rule.}
Let $C$ and $D$ be empty sets.
For each bid, it is put in $C$ with probability $1/2$ and in $D$ with probability $1/2$.
In other words, $C$ and $D$ form a partition for all bids in the block.
Let $C = (c_1,\ldots, c_{|C|})$ where $c_1 \geq \cdots \geq c_{|C|}$, and $D = (d_1,\ldots, d_{|D|})$ where $d_1 \geq \cdots \geq d_{|D|}$.
Then, a bid $b$ is confirmed if it is in $C$ and $b \geq c_{k_D^*}$, where $k_D^* = \newargmax_{k \in [|D|]} k \cdot d_k$ or it is in $D$ and $b \geq d_{k_C^*}$, where $k_C^* = \newargmax_{k \in [|C|]} k \cdot c_k$.
\item 
{\it Payment rule.}
All confirmed transactions in $C$ pay $c_{k_C^*}$ and all confirmed transactions in $D$ pay $d_{k_C^*}$.
\item 
{\it Miner revenue rule.}
All payments go to the miner.
\end{itemize}
\end{mdframed}


\Hao{The following example says that RSOP is not UIC if a user can submit multiple bids.
Suppose there are only two users in the mempool, and the true value of the first user is $4$.
Now, the second user bids $b_2 = 10$ (its true value does not matter).
If the first user bids truthfully, $b_1$ is confirmed if $b_1$ and $b_2$ fall into the same set, and the utility of the first user is $4$.
However, if they fall into differcent sets, $b_1$ would not be confirmed, so the utility is $0$.
The expected utility of the first user is $2$ if it bids truthfully.
On the other hand, suppose the first user submits two bids $b_1 = 1$ and $b_1' = 1$.
We define user $1$'s utility such that it earns its true value if at least one of its bids is confirmed, and the overall utility is the true value minuses the cost.

Without loss of generality, we assume $b_2$ falls into set $A$.
Then, there are three cases:
\begin{itemize}
	\item Both $b_1$ and $b_1'$ fall into set $A$. 
	The utility of the first user is $4$.
	This event happens with probability $1/4$.
	\item Either $b_1$ or $b_1'$ (but not both) falls into set $A$. 
	The utility of the first user is $3$.
	This event happens with probability $1/2$.
	\item Neither $b_1$ nor $b_1'$ falls into set $A$. 
	The utility of the first user is $0$.
	This event happens with probability $1/4$.
\end{itemize}
Thus, the expected utility of the first user is $5/2 > 2$ if it bids strategically.
}

RSOP is UIC by our definition, because each user's confirmation and payment 
do not depend on its own bid, but depends on the bids in the other set.
As pointed out in Example 5.3 and Example 5.5 in Lavi et al.~\cite{zoharfeemech}, the miner is incentived to inject some fake bids or not to choose the highest bids from the mempool, so RSOP is not MIC.

By modifying Example 5.3 in \cite{zoharfeemech}, we obtain the following example saying that RSOP is not $1$-SCP.
Suppose there are only two users in the mempool, and the true value of the first user is $1$.
Now, the second user bids $3$ (its true value does not matter).
If both the miner and the first user are honest, two bids either fall to the same set or fall into differcent sets, each with probability $1/2$.
If they fall to the same set, the miner's revenue is zero, and the utility of the first user is $1$;
if they fall into differcent sets, the miner's revenue is $1$, and the utility of the first user is $0$.
Thus, the expected joint utility is $1$.
However, the miner can sign a contract with the first user and ask it to bid $2$ instead.
If they fall to the same set, the miner's revenue is zero, and the utility of the first user is $1$; 
if they fall into differcent sets, the miner's revenue is $2$, and the utility of the first user is $0$.
Their joint utility becomes $3/2$, which violates $1$-SCP.

\subsection{The BEOS Mechanism}

The following mechanism was introduced by Basu et al.\cite{functional-fee-market}, and we formulate in terms of inclusion, confirmation, payment, and miner revenue rule.
Roughgarden~\cite{roughgardeneip1559,roughgardeneip1559-ec} argued that BEOS mechanism does not satisfy UIC, MIC, or $1$-SCP, and we briefly explain why below for completeness.

\begin{mdframed}
\begin{center}
{\bf BEOS mechanism}
\end{center} 
\paragraph{Parameters:} 
\begin{itemize}[leftmargin=5mm,itemsep=1pt]
	\item the block size $B$
	\item fill threshold $K$, where $0 \leq K \leq B$
	\item interval $I$ for profit-sharing 
	\item minimum entry fee $f$
\end{itemize}

\paragraph{Mechanism:}
\begin{itemize}[leftmargin=5mm,itemsep=1pt]
\item 
{\it Inclusion rule.}
Choose the highest $t$ non-zero bids from mempool to include in the block, so the block is $(b_1, \ldots, b_t)$ where $b_1 \geq \cdots \geq b_t$.
If $t < K$, the miner further chooses one of the following options with higher revenue.
\begin{itemize}
	\item The miner declares the mempool is too empty.
	\item The miner chooses to pay fill penalty, which is $b_t \cdot (K - t)$.
\end{itemize}
\item 
{\it Confirmation rule.}
If the miner declares the mempool is too empty, only the transactions $\geq f$ are confirmed.
Otherwise, all transactions in the block are confirmed.
\item 
{\it Payment rule.}
If $t \leq K$, all transactions in the block pay the lowest confirmed bid in the block. 
Otherwise, if $t > K$, the highest $K$ bids all pay $b_K$, and all other bids pay nothing.
\item 
{\it Miner revenue rule.}
Let $S$ denote the fee of the current block, which is sum of all payment in the block and the fill penalty (if there is any).
The miner is paid $S/I$ and the corresponding fees in the upcoming $I - 1$ blocks.
\end{itemize}
\end{mdframed}

Because the bid ($b_k$) who decides the payment is also confirmed, that user can try to bid lower strategically so that it is still confirmed.
Therefore, the BEOS mechanism violates UIC.

As the authors \cite{functional-fee-market} pointed out, the BEOS mechanism is not MIC, while the strategic gain of the miner by deviating from the mechanism decreases as the reward is splitting into more future blocks ($I$).
To see why it is not MIC, consider the following example.
Suppose $K = 4$ and $I = 3$, and the mempool is $(5, 5, 5, 1)$.
Since it is possible to make the block full, an honest miner should include all four transactions in the mempool, and earn $(1 \cdot 4)/3 = 4/3$ from this block.
However, if the miner only includes $(5, 5, 5)$ and pays the fill penalty $5$, the fee from this block is $(5 \cdot 4)/3 = 20/3$.
Thus, the miner's utility from this block is the difference between the total fee and the fill penalty, $20/3 - 5 = 5/3 > 4/3$.

The same example also suggests the mechanism is not $1$-SCP.
In the honest case, the miner's revenue is $4/3$ and the utility of the fourth user is $v - 1$, where $v$ is the true value.
However, the miner can sign a contract with the fourth user and ask it to bid $5$ instead.
In this case, the miner's revenue becomes $20$ and the utility of the fourth user becomes $v - 5$.
Their joint utility increases.

\Hao{Actually, I conjecture that if $f = 0$ and $K \leq I$, their protocol is MIC.}

\subsection{Dynamic Posted-Price}

Ferreira et al.~\cite{dynamicpostedprice} proposed the dynamic posted-price mechanism.
On the single block level, their mechanism is exactly the posted-price without burning.
Their contribution is to design a rule to decide the posted-price for each block given the blockchain history.
As explained in Appendix \ref{sec:ic-compare}, the posted-price auction satisfies UIC and MIC, but it is not $1$-SCP.